\documentclass[11pt, oneside]{article}      
\usepackage[text={6.5in,8in},centering]{geometry} 
\usepackage{pdflscape}

\usepackage{macros-hao}

\begin{document}
\title{Randomized Graph Cluster Randomization\thanks{Authors are listed in alphabetical order. We thank Guillaume Basse, Dean Eckles, and Aaron Sidford for valuable discussions, as well as seminar participants at 2019 MIT Conference on Digital Experimentation. This work was supported in part by NSF grant IIS-1657104.}}
\author{
Johan Ugander\thanks{Department of
Management Science and Engineering, Stanford University, Stanford,
CA 94305. E--mail: \texttt{jugander@stanford.edu}.} 
\and Hao Yin\thanks{Institute for Computational and Mathematical Engineering, Stanford University, Stanford, CA 94305. E--mail: \texttt{yinh@stanford.edu}.}
}
\date{\today}
\maketitle


\abstract{
The global average treatment effect (GATE) is a primary quantity 
of interest in the study of causal inference under network interference.
With a correctly specified exposure model of the interference, 
the Horvitz-Thompson (HT) and H\'ajek estimators of the GATE 
are unbiased and consistent, respectively, yet known to 
exhibit extreme variance under many designs and in many settings of interest. 
With a fixed clustering of the interference graph,
graph cluster randomization (GCR) designs have been shown to greatly reduce variance compared 
to node-level random assignment, but even so the variance is still often prohibitively large.
 
In this work we propose a randomized version of the GCR design, descriptively named
randomized graph cluster randomization (RGCR),
which uses a random clustering rather than a single fixed clustering.
By considering an ensemble of many different cluster assignments, this design 
avoids a key problem with GCR where a given node is sometimes ``lucky'' or ``unlucky'' 
in a given clustering.
We propose two inherently randomized graph decomposition algorithms
for use with RGCR designs, \emph{randomized $3$-net} and \emph{1-hop-max}, 
adapted from prior work on multiway graph cut problems and 
the probabilistic approximation of (graph) metrics.
We also propose weighted extensions of these two algorithms with 
slight additional advantages.

When integrating over their own randomness, all these algorithms furnish network exposure probabilities that can be estimated efficiently. 
We develop upper bounds on the variance of the HT estimator of the GATE under assumptions on the metric structure of the graph driving the interference. Where the best known variance upper bound for the HT estimator under a GCR design is exponential in the parameters of the metric structure, we give a comparable variance upper bound under RGCR that is instead polynomial in the same parameters. 
We provide extensive simulations comparing RGCR and GCR designs, observing substantial reductions in the mean squared error for both HT and H\'ajek estimators of the GATE in a variety of settings.
}

\clearpage
\tableofcontents
\clearpage

\section{Introduction}

Interest in the design and analysis of randomized experiments under interference has accelerated in recent years~\cite{hudgens2008toward,fienberg2012brief,aronow2017estimating,savje2017average,chin2018central,jagadeesan2020designs}, motivating work on efficient estimators of the global average treatment effect (GATE) \cite{eckles2017design,saint2019using,chin2019regression}. GATE estimation seeks to understand the difference between placing {\it all} units in treatment vs.~placing {\it all} units in control, a natural estimand capturing the full average treatment effect net of all ``network effects.'' 
A major motivation for studying the GATE comes from experiments run on online social networking platforms~\cite{saveski2017detecting,pouget2018optimizing,pouget2019testing} and online marketplaces~\cite{johari2020experimental,hathuc2020counterfactual}, where the interactions are either between social relations or between marketplace competitors. In these settings
a platform designer typically has full control over treatment assignments and is specifically
interested in understanding which condition, when assigned to all units, has the best average outcome.
  

In the case of a binary intervention, a so-called A/B test of treatment versus control, 
the GATE is defined as the difference between the average of outcomes when all individuals 
are exposed to the treatment condition vs.~when all individuals are exposed to control.
Formally, let $\bZ \in \{0,1\}^n$ be a length-$n$ vector representing the
treatment assignment of a population of $n$ individuals, where the value of 1 and 0 
corresponds to treatment and control, respectively. Let $Y_i(\bZ = \bz)$ be 
the $i$-th individual's outcome or response; the mean outcome of all units to $\bZ$ is
$$
\mu(\bz) \triangleq \frac{1}{n} \sum_{i=1}^n Y_i(\bZ = \bz),
$$
and the GATE is then $\tau \triangleq \mu(\bOne) - \mu(\bZero)$.

Exact measurement of the GATE is not possible because the scenarios $\bz=\bOne$ and $\bz=\bZero$ are strongly counterfactual: it is not possible to simultaneously observe the entire population in treatment and the entire population in control. The GATE is typically estimated through randomized experiments, but to connect the outcome of a randomized experiment with the GATE,
assumptions are required to make $\mu(\bz)$ identifiable, e.g.~the \emph{no interference}~\cite{cox1958planning} or \emph{stable unit treatment value assumption} (SUTVA)~\cite{rubin1974estimating}.
However, in many situations there is unavoidable interference between individuals, in the sense that their outcome depends directly on the treatment or outcome of others. 
In the presence of interference,  estimators derived under the SUTVA assumptions are generally biased~\cite{sobel2006randomized, aronow2017estimating}. 
A variety of alternative assumptions have been made in attempts to bring reasonable power to potential outcome inferences under interference, including monotonicity assumptions on the individual treatment effect~\cite{manski2013identification, choi2017estimation, eckles2017design, pouget2018optimizing}. In this work, we don't require a monotonicity assumption for our results to hold, but instead commit to an exposure model framework~\cite{manski2013identification,sussman2017elements,forastiere2020identification}. 

In prior efforts to estimate the GATE, a promising approach has been to replace the SUTVA assumption with a less restrictive {\it exposure model}~\cite{manski2013identification,aronow2017estimating,ugander2013graph}. 
An exposure model identifies, for each unit $i$, the condition when the unit has the same response as
if all units are assigned to treatment or control. We use $E_i^\bz$ to denote the events---defined by subsets of the space of global assignment vectors, to be formally specified later on---where node $i$ responds as if exposed to global treatment ($\bz=\bOne$) or global control ($\bz=\bZero$). For network experiments, $E_i^\bOne$ and $E_i^\bZero$ then capture conditions under which we consider $i$ to be ``network exposed to treatment'' vs.~``network exposed to control''. Throughout this work we will focus our attention on the \emph{full-neighborhood exposure model}, discussed further in Section~\ref{subsec:exposure}.

The Horvitz-Thompson (HT) estimator~\cite{horvitz1952generalization} of the mean outcomes $\mu(\bz)$ is
\begin{equation}   \label{Eq:Est-MO-GCRfirst}
\hat \mu(\bz)   = \frac 1 n \sum_{i=1}^n \frac{\indic{E_i ^\bz}\cdot Y_i(\bz)}{\prob{E_i^\bz}},
\end{equation}
and consequently the HT estimator for the GATE is $\hat \tau = \hat \mu(\bOne) - \hat \mu(\bZero)$. 
Arronow and Samii have shown that, assuming the exposure model 
is properly specified, a standard consistency assumption on the potential outcomes~\cite{vanderweele2009concerning}, 
and that the probability of every node being network exposed to treatment and control is positive, 
then the estimators $\hat \mu(\bOne)$, $\hat \mu(\bZero)$, and $\hat \tau$  are unbiased~\cite{aronow2017estimating}. 

While we focus our analysis of GATE estimation on HT estimators, some of our results extend to the related H\'ajek estimator~\cite{hajek1971comment}, also called the self-normalized estimator~\cite{tukey1956conditional,swaminathan2015self}, of the mean outcome
\begin{equation}   \label{Eq:hajek1}
\tilde{\mu} (\bz) = 
\left(\sum_{i=1}^n \frac{\indic{E_i ^\bz}}{\prob{E_i^\bz}}\right)^{-1}
\sum_{i=1}^n \frac{\indic{E_i ^\bz}\cdot Y_i(\bz)}{\prob{E_i^\bz}},
\end{equation}
with the H\'ajek GATE estimator taking the form $\tilde \tau = \tilde \mu(\bOne) - \tilde \mu(\bZero)$. 
Notice that the \hajek\ and HT estimators utilize the same exposure probabilities for a given design.
The H\'ajek estimator is typically biased but often preferable to the HT estimator under a strong bias--variance trade-off.  

Under independent node-level Bernoulli($p$) randomization---where units are assigned tor treatment with probability $p$ and control with probability $(1-p)$---the variance of the HT GATE estimator quickly blows up if there are units $i$ for which the exposure conditions $E_i^\bOne$ or $E_i^\bZero$ require many independent assignments to all come up heads or all come up tails. For exposure models such as full-neighborhood exposure, where a unit and all of its network neighbors must be assigned to treatment together, $\prob{E_i^\bOne}$ and/or $\prob{E_i^\bZero}$ then quickly become very, very small.

The \emph{graph cluster randomization} (GCR)~\cite{ugander2013graph} experimental design scheme  
was proposed to combat this issue. 
Given a fixed clustering of the graph, i.e., the set of nodes has been partitioned into disjoint \emph{clusters},
GCR jointly assigns all nodes within each cluster into either treatment or control together. This randomization design can be viewed as a correlation imposed on the way in which assignment vectors $\bZ \in \{0,1\}^n$ are drawn, correlating neighbors in the graph  with the goal of broadly increasing the collections of probabilities $\prob{E_i^\bOne}$ and $\prob{E_i^\bZero}$, for all nodes $i$, for a given exposure model.   
GCR can be shown to achieve a considerable variance reduction under certain settings compared to independent assignment.
Eckles \etal \cite{eckles2017design} evaluated GCR for GATE estimation and showed that it reduces bias, variance, and mean squared error (MSE) 
in scenarios where there is a strong direct treatment effect and network spillover. 
However, they found that it often still exhibits considerable MSE, which can then exceed the MSE of independent assignment when spillover effects are small.

The GCR scheme operates using a pre-specified fixed clustering assignment,
and a known problem with GCR is that, informally, a node can get ``unlucky'' in the fixed cluster assignment,
adjacent to many clusters.
For such unlucky nodes, the probability of network  
exposure to treatment or control is then very low under GCR with that clustering, which greatly inflates the variance of the HT GATE estimator $\hat \tau$.
Therefore, even though GCR has been shown to theoretically give considerable variance reductions compared with node-level randomization, the variance can still be very, very large.
Another disadvantage of GCR is the incompatibility with {\it complete
randomization} at the cluster level due to a violation of the positivity
assumption required by both the HT and \hajek\ GATE estimators.

We propose an extension of the GCR scheme whereby the graph cluster randomization is itself based on a randomized clustering. We descriptively call this scheme \emph{randomized graph cluster randomization (RGCR)}. We find that RGCR can greatly reduce the variance of the HT GATE estimator both in theory and in extensive simulations, compare to ordinary GCR. 
Further simulations using the \hajek\ GATE estimator, while lacking theoretical support, show that it too
benefits from RGCR (vs.\ GCR) and is often preferable to the HT estimator for a given design.
Most importantly,
we find that these variance reductions are considerable enough to bring RCGR into the realm of being ``useful'' in 
many situations where GCR would fail to deliver a GATE estimate with actionable MSE.

\begin{figure}[t] 
   \centering
\begin{minipage}{0.35\linewidth}\centering
   \includegraphics[width=1.3in]{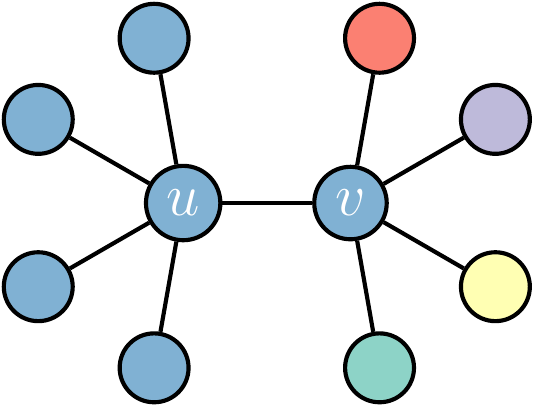} 
   \\{\small clustering $\bc_1$   }
\end{minipage}
\begin{minipage}{0.35\linewidth}\centering
   \includegraphics[width=1.3in]{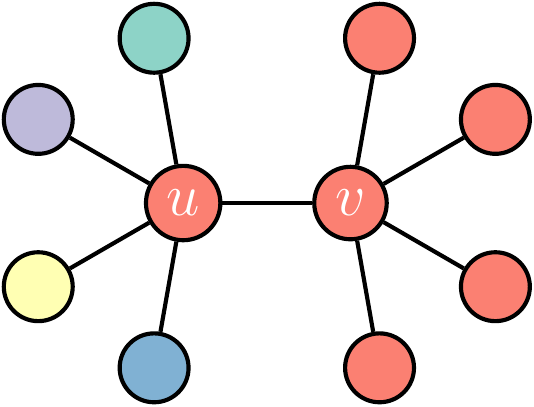} 
   \\{\small clustering $\bc_2$   }
\end{minipage}
\caption{An illustration of variance reduction with randomized graph cluster randomization (RGCR), considering two different clusterings $\bc_1$ and $\bc_2$, where colors denote clusters. For the given network, consider GCR with a fixed clustering and clusters assigned to treatment or control with probability $p=1/2$. 
The full-neighborhood exposure probabilities of nodes $u$ and $v$ 
are either $2^{-1}$ and $2^{-5}$ (under $\bc_1$) or $2^{-5}$ and $2^{-1}$ (under $\bc_2$) respectively, contributing $2 + 2^5 = 34$ to the variance of the 
HT estimator of the GATE. In contrast, when randomizing evenly between $\bc_1$ and $\bc_2$, 
the exposure probabilities of $u$ and $v$ both become $(2^{-1} + 2^{-5})/2$,
contributing $2 / (2^{-1} + 2^{-5}) * 2 \approx 7.5$ to the variance.
}
\label{Fig:mix_two_partitions}
\end{figure}

The intuition that motivates using a random cluster partition is illustrated in~\Cref{Fig:mix_two_partitions}. Essentially, when averaging across different cluster assignments, 
the distribution of individual network exposure probabilities $\prob{E_i^\bz}$ will be less skewed because different nodes will be ``unlucky'' in different clusterings.
Averaging across many clusterings washes out extremely small probabilities, greatly reducing GATE  estimator variance.

One can consider two approaches to randomized graph clustering. First, consider employing a uniform mixture of $K$ graph clusterings, each obtained via a (potentially different) black box clustering algorithm. In this setting, we can compute the exposure probabilities (needed for the HT and H\'ajek estimators of the GATE) simply by averaging the exposure probabilities across clusterings. That said, computing many clusterings of a large graph can be very computationally expensive. As a more appealing approach, we consider employing inherently randomized graph clustering algorithms where it is potentially tractable to consider the exposure probabilities when integrating over the full randomness of the algorithm. For at least one of the algorithms we consider in this work, \emph{randomized 3-net}, we show that the exact computation of the full-neighborhood exposure probabilities is NP-hard. Even so, we are able to construct Monte Carlo estimators of the probabilities with relative errors that can be bounded at a reasonable computational cost. The Monte Carlo estimation procedure we employed is practically equivalent to generating $K$ clusterings from the randomized algorithms and then averaging, but we do not need to store all $K$ clusterings at any point.

\xhdr{Mulit-way cuts and randomized partitioning} 
The randomized clustering algorithms we analyze in this work stem from the literature on probabilistic approximations of graph metrics. Randomized graph decompositions have a rich history \cite{linial1993low} originally driven by interests in distributed graph computations~\cite{alon1986fast,miller2013parallel}. The algorithm we call \emph{1-hop-max} is closely related to the CKR partitioning algorithm~\cite{calinescu2005approximation}, developed as an approach to the 0-extension problem~\cite{karzanov1998minimum}, a metric generalization of the multi-way cut problem on graphs~\cite{dahlhaus1992complexity}. Our 1-hop-max algorithm runs the CKR algorithm with centers (or ``terminals'') selected at random, as is also done in the closely related FRT algorithm for metric approximation~\cite{fakcharoenphol2004tight}, and with a fixed radius of one. The other algorithm we consider, \emph{randomized $3$-net clustering}, comes from the related literature on metric approximation in bounded geometries~\cite{gupta2003bounded} with applications to nearest neighbor search~\cite{karger2002finding}. Graph cluster randomization with a fixed $3$-net clustering was previously analyzed in the original work on GCR~\cite{ugander2013graph}. In the randomized setting of RGCR, we find 1-hop-max more amenable to theoretical analysis, while simulations indicate that RGCR with 1-hop-max and randomized $3$-net do comparably well in diverse settings.

\xhdr{Restricted growth conditions} 
The conceptual notion of a (graph) metric with bounded geometry is very useful for considering the design of good clustering algorithms for social networks, 
as social networks arguably exhibit a version of bounded growth.
Let $G=(V,E)$ be a graph,  
$\dmax = \max_{i \in V}\{\degree{i}\}$ denote the maximum degree, 
$\dist(i,j)$ the shortest path distance on $G$, 
and let $B_r(i) = \{j \in V \mid \dist(i,j) \leq r\}$ for $r>0$
denote the $r$-hop neighborhood of node $i$, also sometimes called the $r$-ball at node $i$.

As a motivating empirical observation, 
due to apparent tendencies towards clustering,
the size of social network neighborhoods $\lvert B_r(i)\rvert$ 
tend to grow slower than $(\dmax)^r$ in $r$ \cite{ugander2011anatomy}.
There are two ways to operationalize this empirical tendency.
First, borrowing a definition from the literature on metric approximation~\cite{karger2002finding}, 
one could consider experimental designs that perform well under
a condition of \emph{bounded growth}, whereby
there is a constant $\eta > 0$ such that
$$\lvert B_{2r}(i)\rvert \leq \eta \lvert B_r(i)\rvert, \forall r \ge1,$$ 
for all nodes $i$. 
Second, the original GCR work identified and developed results under 
a less restrictive metric property of 
\emph{restricted growth}~\cite{ugander2013graph},
which assumes there is a constant $\kappa > 0$ such that
$$\lvert B_{r+1}(i)\rvert \leq \kappa \lvert B_r(i)\rvert, \forall r\geq 1,$$
for all nodes $i$.
Notice that bounded growth implies restrictive growth. 
The constants $\eta$ and $\kappa$ here are called the
\emph{bounded growth} and \emph{restrictive growth coefficients}, respectively.
We emphasize that both of these definitions start at a radius of $r \ge 1$, 
and thus we do not require any relationships to hold between $B_0$ and $B_1$ 
(otherwise we would have $\kappa = 1 + \dmax$), 
and it can be easily verified that $\kappa \leq \dmax$.
Our goal, building on the initial analysis of GCR, is to exploit degree bounds and/or restricted growth structure to design algorithms that work provably well when $\dmax$ and/or $\kappa$ are modest. 

We note that a separate approach to causal inference under network interference has recently assumed metric growth conditions of a slightly different variety~\cite{leung2019causal}. That work follows recent work on limit theorems for network-dependent random variables where growth conditions appear as part of sufficient conditions~\cite{kojevnikov2019limit}.

\xhdr{Bounded geometry of empirical social networks}
While bounded geometry assumptions play a central role in the previous theoretical analysis of GCR~\cite{ugander2013graph}
and other recent work~\cite{kojevnikov2019limit,leung2019causal}, 
the empirical growth rates of $r$-balls in social networks has not been well-documented. 
The average degree, degree distribution, and path-length distribution of 
large-scale social networks have all been the subject of extensive empirical 
investigations~\cite{leskovec2008planetary,backstrom2012four,ugander2011anatomy},
with the path length distribution being the central object of study in the large literature
on ``degrees of separation'' inspired by Milgram \cite{travers1969experimental}. 
Less attention has been given to the empirical structure of neighborhood sizes in at different distances, though 
some intuition for the relationship between friend counts and friend-of-friend
counts can be derived from prior work~\cite{ugander2011anatomy,radaelli2018quantifying,su2016effect}.

Our empirical analysis, given in Appendix~\ref{app:growth}, documents that for Facebook college social networks
$\kappa$ is typically on the order of $25-50\%$ of $\dmax$. 
As an aside, recall that the coefficient $\kappa$ describes a worst-case coefficient. We observe that $\kappa$ is typically pushed up by a few bad nodes where, 
e.g., a degree-1 node $u$ is connected to a high degree node, making $|B_2(u)| /|B_1(u)|$ very high and thus $\kappa$ high for the graph as a whole. 
As part of Appendix~\ref{app:growth} we investigate the empirical growth of $r$-balls in fine-grained detail. It's possible that new
paths forward for studying estimators (and limit theorems~\cite{kojevnikov2019limit}) on social networks may be more suited to an alternative formulation of restricted growth, not yet formulated.

\xhdr{Bounds on the HT variance for the GATE} 
Our main theoretical result is to show that under a restricted growth condition (``$\kappa$''), RGCR delivers qualitatively better bounds on the HT variance compared to GCR (which is already known to be qualitatively better than independent randomization). More specifically, in a graph on $n$ nodes with restricted growth coefficient $\kappa$ and cluster assignment probability $p$, previous results~\cite{ugander2013graph} have shown that  the variance of the HT estimator of $\hat \mu(\bOne)$ under GCR with a fixed $3$-net clustering is upper bounded by 
$$
\var{\hat \mu(\bOne)} \le \frac{1}{n} \dmax \kappa^5 p^{- \kappa^6} \Theta(1),
$$ 
polynomial in the maximum degree but exponential in $\kappa$. In the absence of a restricted growth condition but in the presence of a max degree bound, a variance upper bound of
$\var{\hat \mu(\bOne)} \le \frac{1}{n} \dmax^{6}  p^{- \dmax} \Theta(1)$, 
can be obtained (via a more direct argument than one that sets $\kappa=\dmax$ above).
Returning to the setting of restricted growth, in this work we show that under RGCR with a randomized 1-hop-max clustering we can upper bound the variance by 
$$
\var{\hat \mu(\bOne)} \le \frac{1}{n} \dmax^2 \kappa^4 p^{-1} \Theta(1),
$$ 
polynomial in both $\kappa$ and $\dmax$. 
In the absence of a restricted growth condition but in the presence of a max degree bound, we obtain
$\var{\hat \mu(\bOne)} \le \frac{1}{n} \dmax^{6}  p^{- 1} \Theta(1)$.
We do not derive any theoretical bounds for the variance of the H\'ajek estimator, but our simulations (\Cref{sec:simulation}) explore the empirical behavior of the H\'ajek estimator extensively. 
 
The bounds on the variance of the HT GATE estimator $\hat \tau$ are analogous to these bounds for the mean outcome $\hat \mu(\bOne)$. The difference between these two variance bounds, under GCR vs.\ under RGCR, is striking both in the setting of a fixed modest $\kappa$ and in settings where $\kappa$ is on the order of $\dmax$. 
Recall that the latter setting is empirically quite common per analysis in \Cref{app:growth},
and our analysis furnishes an upper bound on the HT GATE variance under RGCR that is exponentially lower than the comparable bound under vanilla GCR.

The switch to 1-hop-max instead of $3$-net is for analytical convenience: the two algorithms are very similar, but once randomized, the distribution of clusterings produced by the randomized $3$-net algorithm are not as amenable to analysis. For comparison, non-randomized GCR with a single fixed 1-hop-max clustering has a HT variance upper bound of $\frac{1}{n} \dmax \kappa^3 p^{-\dmax} \Theta(1)$, exponential in the max degree (and thus worse than $3$-net when $\kappa$ is modest). A summary of our variance bounds for HT estimators is given in~\Cref{tab:summary_theory} in \Cref{sec:rgcr}, which also shows slightly improved bounds based on weighted variations of both 1-hop-max and randomized 3-net. In our work we do not perform any analysis under the weaker bounded growth condition (``$\eta$''), owing to the well-known fact that social networks have a very limited effective diameter~\cite{leskovec2007graph}, with the vast majority of node pairs appearing within a hop distance of six~\cite{backstrom2012four}, limiting the utility of a bounded relationship between $B_{2r}$ and $B_r$.

The connection between existing techniques for optimizing randomized graph decompositions and designing low-variance network experiments is intuitive---both problems aim to cut a graph into many small, well-separated parts---but we emphasize that at present the connection we make here is only intuitive. Minimizing the variance of the HT estimator of the GATE, as an objective, is not merely a matter of finding a good graph cut in any traditional sense. To this point, our successful theoretical analysis not of randomized $3$-net but of $1$-hop-max (which is related to CKR partitioning~\cite{calinescu2005approximation}) under a restricted growth condition stands in contrast to the metric approximation literature, where $r$-net algorithms are those that yield a powerful analysis under growth restrictions~\cite{gupta2003bounded}.

\xhdr{Curse of large clusters}
The GCR scheme suffers from large variance 
when nodes are connected to many clusters. A naive solution 
to this specific problem would be 
to partition the network into only a few, say $K$, clusters
where $K = O(1)$ is pre-specified and independent of the size of the network.
However, such an approach fails when the nodes' outcome exhibits
homophily or some other global drift pattern such that nodes at a short distance 
have similar response outcomes. 
If there is significant difference in the response of nodes in different clusters,
and only a few ($K=O(1)$) clusters,
then the observed difference $\hat \tau = \hat \mu(\bOne) - \mu(\bZero)$ 
will be sensitive to this cluster-level variation, with additional variance incurred
that does not then decay with the size of the network.

For the RGCR scheme, in \Cref{sec:ring} we show that this issue persists, and using a random clustering
with large clusters (of size $\Theta(n)$, so $K=O(1)$) prohibits the HT estimator variance from
converging to zero even as $n \rightarrow \infty$. Specifically, analyzing a ring network
where the optimal balanced $K$-partitions are obvious,
when selecting one of the optimal balanced $K$-partition uniformly at random in RGCR,
we show that
\[
\var{\hat \tau} \rightarrow \Omega\left( \frac{b^2}{K} \right)
\]
as $n \rightarrow \infty$, where $b$ is a homophily-like measure of the magnitude 
in the cluster-level average different in nodes' response.
Therefore, if $b > 0$ and $K = O(1)$, then we have $\var{\hat \tau} = \Omega(1)$
even with $n \rightarrow \infty$.
This result provides an important insight on the choice of random clustering
used in RGCR scheme: the number of clusters in the output random clustering
should increase with the number of individuals to let $\var{\hat \tau} \rightarrow 0$ 
as $n \rightarrow \infty$,
a necessary condition on the random clustering strategy.
Consequently, various graph clustering algorithms
such as spectral partitioning~\cite{spielman1996spectral,shi2000normalized,lee2014multiway}, 
balanced label propagation~\cite{ugander2013balanced},
or reLDG~\cite{nishimura2013restreaming,saveski2017detecting} are 
not good clustering strategies for RGCR if the number of clusters in the output
is small.

\xhdr{Simulations} 
Extending our analysis beyond theoretical results on variance bounds under bounded geometries, 
we provide a extensive simulation-based analysis of various RGCR schemes. 
We vary many aspects of the simulation to understand the efficacy of RGCR-based
experiments for GATE estimation. We observe dramatic variance reduction
for the HT GATE estimator used RGCR compared to GCR, bringing a useless variance ($~10^{50}$) down to 
a potentially useful variance ($~10^0$). We vary the structure of the underlying network,
the randomized clustering algorithm,
the possible weighting used in the algorithm, whether randomization is independent or complete, 
and whether the estimator is HT or H\'ajek. 

A specific innovation in our simulations is a rich graph-aware response model,
exhibiting both degree-correlated responses and homophily in responses. 
Specifically, if two nodes have short graph distance, their responses
tend to be close, resembling responses in many real-world settings \cite{mcpherson2001birds}
not captured in typical response models used in beyond-SUTVA simulations. 
Note that a failure to capture homophily in the response model can result in
preferring a random clustering algorithms that generates few large clusters, 
concealing the issue of large clusters as developed in the previous discussion
and presented more fully in \Cref{sec:ring}.
In our response model, homophily is added to the model using techniques from 
spectral graph theory~\cite{von2007tutorial}, constructing a (non-constant) function on the 
node set where responses of graph neighbors are similar.

We find that for both HT and H\'ajek estimator
of the GATE, RGCR tends to dramatically improve on GCR in our rich simulations, 
while varied adjustments to
the specific RGCR scheme can have additional gains.  
We find a RGCR scheme using degree-weighted randomized $3$-nets with complete randomization 
to generally be the lowest variance.

\xhdr{Paper roadmap} 
The remainder of this paper is organized as follows. 
After a detailed introduction to preliminary definitions in \Cref{sec:preliminary}, 
we formally propose the RGCR scheme in \Cref{sec:rgcr}.
In \Cref{sec:properties} we develop key theoretical 
properties of RGCR (e.g., variance reduction) under HT estimation, 
 with a focus on the two families of random clustering algorithms we consider in this work, 
 the 3-net and 1-hop-max algorithms, as well as their weighted variants.
We also discuss the bias of the related \hajek~estimator under RGCR.
In \Cref{sec:ring}, we formalize a theory for the curse of large clusters,
which provides a necessary condition on the random clustering
algorithm for the variance to converge to zero as a network grows large.
In \Cref{sec:simulation} we provide extensive simulation results comparing
different RGCR and GCR schemes.
\Cref{sec:conclusion} concludes.

\section{Preliminaries}
\label{sec:preliminary}

\subsection{Networks and growth rates}

Throughout this work we will consider interference in 
network settings as modeled by an undirected, unweighted network $G = (V, E)$, dubbed the \emph{interference graph},
where the node set $V = \{1, 2, \dots, n\}$ represents the units/individuals and 
$E$ is the collection of edges that represent pairwise response
dependencies that underly the interference. 
For each individual $i$, let $\nbr{i}$ be the set of its neighbors
on the network, and $\degree{i} \triangleq \lvert \nbr{i} \rvert$ be its \emph{degree}.
We use $\dmax = \max_{i \in V}\{\degree{i}\}$ to denote the maximum degree
of all nodes in the network. 
A natural distance between a pair of nodes $i$ and $j$ on network $G$ is the 
\emph{shortest path distance} denoted as $\dist(i,j)$, \ie, the length of the shortest path connecting them.
With a positive integer radius $r > 0$, we use $B_r(i) = \{j \in V \mid \dist(i,j) \leq r\}$ 
to denote the $r$-hop neighborhood of node $i$.
For example, with $r=1$, $B_1(r)$ contains node $i$ itself and all its neighbors, and thus $\lvert B_1(i)\rvert = 1 + \degree{i}$.

Throughout this work we make broad use of the idea of a decomposition of a graph into  
\emph{clusters}.
A \emph{clustering} is a partition of all nodes in the network into
some non-overlapping clusters, which is also referred as a \emph{partition}.
We denote a clustering as a vector $\bc = [c_1, \dots, c_n] \in \real{n}$
such that nodes $i$ and $j$ belongs to the same cluster if and only if $c_i = c_j$.
Ideally clusters are internally densely connected while relatively separated from 
the rest of the network, though our definitions require no such thing.

\subsection{GATE estimation under exposure models}
\label{subsec:exposure}
In many online and social settings, the presence of interference introduces bias
in the estimation of global average treatment effects if a no-interference assumption,
\eg~SUTVA, is incorrectly specified. 
More relaxed assumptions than SUTVA can be made that, if correct, can enable reasonable inference.
As a first example, the class of
\emph{constant treatment response} (CTR)~\cite{manski2013identification} 
assumptions identify, for each individual $i$, an \emph{effective treatment mapping}
$g_i$ that captures equivalence classes of the global assignment vectors $\bz$: if $g_i(\bz_1) = g_i(\bz_2)$
for two global assignments $\bz_1$ and $\bz_2$, then $Y_i(\bz_1) = Y_i(\bz_2)$.
SUTVA is a special case of CTR with $g_i(\bz) = z_i$, \ie, where each individual's
response depends only on the treatment assignment of itself.

The \emph{neighborhood treatment response} (NTR)~\cite{aronow2017estimating}
assumption is another case of a CTR assumption, which allows some treatment-based
spill-over effect: for any two global assignments $\bz_1$ and $\bz_2$, 
$g_i(\bz_1) = g_i(\bz_2)$ if $\bz_1[B_1(i)] = \bz_1[B_1(i)]$, \ie, an individual's
response depends only on the treatment assignment of itself and its neighbors.
Consequently, individuals generate the same response as under the global treatment 
($\bz = \bOne$) assignment (a condition termed \emph{network exposed to treatment})
if they and all their neighbors are assigned to the treatment group;
similarly, they generate the same response as under the global control
($\bz = \bZero$) assignment (a condition termed \emph{network exposed to control}) 
if they and all their neighbors are assigned to the control group.
Ugander \etal termed this pair of network exposure conditions as the
\emph{full-neighborhood exposure model}, and other more relaxed neighborhood 
exposure models have also been 
discussed~\cite{manski2013identification,ugander2013graph}. 

In this work we focus on
the full-neighborhood exposure model due to it being the most restrictive 
neighbor exposure model. It greatly simplifies our theoretical analysis, relative to
other more complicated exposure models, while still providing conclusions
that generalize, at least at the level of intuition, to more relaxed neighborhood
exposure models .
Throughout this work we use the events $E_i^\bz$ specifically for full-neighborhood exposure, letting 
$E_i^\bz$ denote the event (a subset of the global assignment vectors in $\{0,1\}^n$)
where node $i$ is network-exposed to treatment ($\bz=\bOne$) or control ($\bz=\bZero$).

Both the Horvitz-Thompson (HT) and H\'ajek estimators require the following positivity
assumption on the network exposure probabilities in order to be well-defined.
\begin{assumption}\label{assm:positive}
At every node $i$ and for both $\bz \in \{\bOne, \bZero\}$,
the network exposure probability is positive: $\prob{E_i ^\bz} > 0$.
\end{assumption}
Aronow and Samii have shown that 
assuming the exposure model is properly specified and a standard consistency assumption on the potential outcomes applies,
the estimators are unbiased.
They derive the variance of the HT estimators under these assumptions~\cite{aronow2017estimating}. 
Specifically, the variance of the HT estimator of the mean outcome, $\hat \mu(\bz)$, is
\begin{equation}   \label{Eq:var-mean_outcome}
\begin{array}{rcl}
\var{\hat \mu(\bz)} &=& \textstyle \frac{1}{n^2} \left[ \sum_{i=1}^n \left( \frac{1}{\prob{E_i ^\bz}} - 1 \right) Y_i(\bz)^2  \right.
\\ & & \qquad \left. +  \sum_{i=1}^n  \sum_{j = 1, j \neq i} ^n \left(\frac{\prob{E_i ^\bz \cap E_j ^\bz}}{\prob{E_i ^\bz} \prob{E_j ^\bz}}- 1 \right)Y_i(\bz)Y_j(\bz) \right],
\end{array}
\end{equation}
for $\bz = \bOne, \bZero$, and the variance of GATE estimator is then
\begin{equation}   \label{Eq:var-GATE}
\var{\hat \tau} = \var{\hat \mu(\bOne )} + \var{\hat \mu(\bZero)} - 2 \cdot \cov{\hat \mu(\bOne), \hat \mu(\bZero)},
\end{equation}
where the covariance is
\begin{equation}   \label{Eq:covar-mean_outcome}
 \cov{\hat \mu(\bOne), \hat \mu(\bZero)} =  \frac{1}{n^2} \left[ \sum_{i=1}^n  \sum_{j = 1, j \neq i} ^n \left(\frac{\prob{E_i ^\bOne \cap E_j ^\bZero}}{\prob{E_i ^\bOne} \prob{E_j ^\bZero}}- 1 \right)Y_i(\bOne)Y_j(\bZero) - \sum_{i=1}^n  Y_i(\bOne) Y_i(\bZero) \right].
\end{equation}

The variance of the H\'ajek estimator can be approximated via a standard Taylor series linearization~\cite{sarndal2003model}.
In this work we do not derive any theoretical results for the variance of the H\'ajek estimator. When the variance
of the H\'ajek estimator is studied in \Cref{sec:simulation}, it is estimated from extensive simulations.

\subsection{Graph Cluster Randomization (GCR)}
The network exposure probabilities $\prob{E_i ^\bz}$, as well as
the joint exposure probabilities $\prob{E_i ^{\bz_1} \cap E_j ^{\bz_2}}$,
are properties of the experimental design. 
With node-level independent randomization, where we assign each node into the treatment
or control group independently, the exposure probability of each node
is exponential to the node degree, and thus it can be extremely small in a large network
with high-degree nodes. The variance of HT estimator is a monotone decreasing function 
in any single exposure probability, meaning that small probabilities beget large variances.  
As a result, the HT estimator variance can be exponentially
large in the largest degree $\dmax$ and not practical~\cite{ugander2013graph}.

To overcome the issue of exponential variance, Ugander \etal proposed to randomize
at the cluster level, the \emph{Graph Cluster Randomization} scheme~\cite{ugander2013graph}:
with a clustering $\bc$ of the network, one can jointly assign all nodes in each cluster
into the treatment or control group. 
A definition of HT estimator for $\mu(\bz)$ was given in the introduction, but restating it more formally in the context of GCR,
\begin{equation}   \label{Eq:Est-MO-GCR}
\hat \mu _{\bc}(\bz)   =  \frac 1 n \sum_i \frac{\indic{E_i ^\bz}\cdot Y_i(\bz)}{\prob{E_i ^\bz \mid \bc}},
\end{equation}
where the subscript indicates that this estimator is based on the design associated
with clustering $\bc$. Under this design, the exposure probability
of each node is exponential not in its degree, but in the number of clusters intersecting
with its 1-hop neighborhood, and thus should reduce the variance in the HT estimators if a reasonable clustering is in use.
Specifically, Ugander \etal show that, if the clustering is generated from the 3-net clustering algorithm,
and the graph satisfies the restricted growth condition with coefficient $\kappa$, 
then the variance is upper bounded by a linear function of the maximum degree of the graph:
\begin{equation}   \label{Eq:GCRVarUB}
\var{\hat \mu_{\bc}(\bOne)} \leq \frac{1}{n} \cdot  \dmax \kappa^5 p^{-\kappa^6} \cdot \Theta(1).
\end{equation}

Despite significant variance reduction compared with node-level
independent randomization, the GCR scheme has one main disadvantage: 
the variance of estimation is still potentially enormous, due to the existence 
of extremely small exposure probabilities. With a single fixed clustering of the network,
a node may be ``unlucky" and directly connect to many clusters. 
For such node to be network exposed to treatment or control,
all the adjacent clusters have to be assigned into the treatment or control
group respectively, making the exposure probability exponentially small.

A naive solution to this issue would be to partition the network into only a few clusters,
so each node can be adjacent to at most the number of clusters in the clustering.
However, this solution is prohibited due to two concerns.
First, partitioning the network into few but large clusters makes the estimated
result very sensitive to network homophily, as discussed in \Cref{sec:ring}, introducing
an additional source of variance that does not decay with the network size.
Second, with just a few clusters, independent randomization at the cluster level
may cause significant imbalance in treatment/control assignment.
For example, with a bisection of the network, if each cluster is assigned independently
into the treatment group with probability $1/2$, then there is a 25\% chance
that both clusters (and consequently all nodes in the network) are assigned into the treatment group,
and we collect no information about the control condition.
To maintain balance with two clusters, one would need to assign the clusters
to opposite conditions (treatment, control), the method of complete randomization.

However, a secondary disadvantage of the GCR scheme is that it is incompatible
with complete randomization at the cluster level,
due to potential violation of the positivity assumption (Assumption~\ref{assm:positive}).
For example, with GCR with few clusters and complete randomization, a node 
connected to all the clusters will always have some neighbors in treatment and some in control,
making it impossible for that node to be full-neighborhood exposure to either
treatment or control.

\section{Randomized Graph Cluster Randomization}
\label{sec:rgcr}

In this section, we present the Randomized
Graph Cluster Randomization (RGCR) scheme of experimental design
and analysis. Different from the 
original Graph Cluster Randomization (GCR) approach~\cite{ugander2013graph} 
that is associated with a single fixed clustering $\bc$,
the RGCR scheme is based on random clusterings. 

Formally, let $\mathcal P$ be a \emph{random clustering generator}, \ie,
an algorithm whose output $\bC$ is a clustering of the input graph, and
the output is random. Without ambiguity of notation, we also use 
$\mathcal P$ to denote the distribution of the randomly generated clustering,
\ie, $\mathcal P(\bc)$ is the probability of the clustering $\bc$ being generated.
The design and analysis of the RGCR scheme are both tailored to the
random clustering generator $\mathcal P(\cdot)$, or equivalently,
the resulting distribution of random clusterings.

\xhdr{Design}
With a random clustering generator $\mathcal P$, the experimental design 
is based on a two-step process.
First, we realize a clustering $\bc$ from the random clustering $\bC$. 
Second, like in the GCR scheme, we perform treatment/control assignment at the cluster level, 
jointly assigning all nodes within each cluster of $\bc$ into the treatment group 
with probability $p$, or into control otherwise.

In the second step of the above cluster-level randomization, 
GCR assign each cluster using \emph{independent randomization}.
For RGCR, besides independent randomization, we also consider 
\emph{complete randomization}, where we further introduce stratification.
In the case of $p = 1/2$, we first stratify the clusters of $\bc$ into pairs, by size (measured by the number of nodes):
the two largest clusters are a pair, the third and fourth largest cluster are a pair,
and so on. We then assign each pair of clusters together, with one into the treatment
and the other into the control group. Complete randomization with other values of $p$
is implemented analogously.
Complete randomization guarantees an equal number of clusters in treatment and control, 
thereby balancing the number of individuals as well. Stratification further tightens this balance.

Balance guarantees are especially important when the clustering contains only few clusters.
For example, in the case of a clustering formed by a graph bisection, 
under independent randomization
the probability that both clusters are assigned into the treatment group or both assigned
into the control group is 0.5, 
an unpleasant scenario where we collect information about only the treatment group 
or only the control group. 
In contrast, with complete randomization we always have one cluster assigned 
to the treatment group and the other to the control group.
Moreover, complete randomization may increase $\prob{E_i ^\bOne \cap E_j ^\bZero}$
for distant nodes, which increases the covariance of $\hat \mu(\bOne)$ and $\hat \mu(\bZero)$
and thus further reduces variance according to \Cref{Eq:covar-mean_outcome,Eq:var-GATE}.
Such variance reduction is consistent with our observation
in our simulation in \Cref{sec:simulation}.

Under GCR, complete randomization can violate the positivity assumption.
For example,
if a node $i$ is adjacent to a pair of clusters that are determined to
be oppositely assigned into the treatment and control group,
then it is impossible for node $i$ to be full-neighborhood exposed 
to treatment or control, \ie, $\prob{E_i ^\bOne} = \prob{E_i ^\bZero} = 0$.
Without positivity, the HT estimators (\Cref{Eq:Est-MO-GCR})
are ill-defined. 
For RGCR, we highlight in \Cref{sec:expo_prob_unweighted} that as a consequence of 
Theorem~\ref{Thm:prob_LB}, RGCR using our randomized $3$-net and 
$1$-hop max clustering algorithms always satisfies node-level positivity for the
full-neighborhood exposure condition  (and related fractional conditions).

\xhdr{Analysis}
With both independent or complete randomization, the exposure probability 
of each node $i$ conditioned on the generated clustering $\bC=\bc$, \ie, 
$\prob{E_i ^\bz \mid \bC=\bc}$, can be computed as in the GCR scheme. 
While we focus on full-neighborhood exposure throughout this work, we 
note that this observation applies to, e.g., partial neighborhood exposure 
conditions~\cite{ugander2013graph} as well.
In the analysis phase of an RGCR experiment,
we use the exposure probabilities unconditional on the clustering in use,
which only depends on the clustering distribution $\mathcal P$.
Formally, since the random clustering in use is generated from the distribution $\mathcal P$,
the network exposure probability of each node $i$, due to the Law of Total Expectation, is
\begin{equation}   \label{Eq:ExpoProb_mix}
\prob{E_i ^\bz \mid \mathcal P} = \sum_{\bc} \mathcal P(\bc) \prob{E_i ^\bz \mid \bc} 
= \expect[\bc \sim \mathcal P]{\prob{E_i ^\bz \mid \bc}}.
\end{equation}
Consequently, the HT estimators are
\begin{equation}   \label{Eq:Est-MO-RGCR}
\hat \mu _{\mathcal P}(\bz)   =  \frac 1 n \sum_i \frac{\indic{E_i ^\bz}\cdot Y_i(\bz)}{\prob{E_i ^\bz \mid \mathcal P}},
\end{equation}
where $\bz = \bZero$ or $\bz=\bOne$, and 
$\hat \tau_{\mathcal P} = \hat \mu _{\mathcal P}(\bOne) - \hat \mu _{\mathcal P}(\bZero)$ is the 
HT estimator of the GATE $\tau$.
Here the subscript $\mathcal P$ emphasizes that the estimator is based on a distribution of clusterings.
The H\'ajek estimators for $\mu(\bZero)$, $\mu(\bOne)$, and $\tau$ are analogous, using the unconditional
exposure probabilities in place of the conditional probabilities.

\xhdr{Putting design and analysis together}
There are a number of important challenges in going from using a single fixed clustering to 
using a random clustering in the graph cluster randomization scheme. 
Not all randomized clustering algorithms are suitable for RGCR. In the next section we discuss key properties that make an algorithm suitable for RGCR, and show that randomized $3$-net and $1$-hop-max are both good algorithms in these regards. Most concretely, in the design phrase one needs to be able to efficiently generate a single random clustering to launch an experiment. As a complementary challenge in the analysis phrase, HT and H\'ajek estimators require per-node unconditional exposure probabilities, which may be more or less difficult to compute, depending on the randomized clustering algorithm used.  We discuss and compare properties of different random clustering strategies in the following section.

\section{Theoretical properties of RGCR}
\label{sec:properties}


\begin{table}[t]
\centering
\begin{tabular}{c  @{\hskip 3ex} c @{\hskip 5ex}  c @{\hskip 5ex}  c 
}
\toprule
clustering   & \multirow{2}{*}{scheme} &  $\prob{E_i ^\bOne}$  & $\var{\hat \mu(\bOne)}$  
\\ algorithm & & (lower bound) & (upper bound) 
\\ \midrule
--   & i.i.d. &  $p^{\dmax+1}$ & $\frac{1}{n} \dmax \kappa p^{-\dmax}$ 
\\ \midrule
\multirow{2}{*}{3-net}  
   & GCR    &  $p^{\kappa^6}$ & $ \frac{1}{n} \dmax \kappa^5 p^{- \kappa^6}$   
\\ \rule{0pt}{3ex}
   & RGCR &  $\frac{p}{(\dmax+1) \kappa }$ & -- 
\\ \rule{0pt}{3ex}
$\bw^*$-weighted   & RGCR &  $\frac{p}{\lambda^*}$ & -- 
   \vspace{1ex}
\\ \midrule
\multirow{2}{*}{1-hop-max}
   & GCR    &   $ p^{\dmax+1}$ & $  \frac{1}{n} \dmax \kappa^3 p^{-\dmax}$   
\\ \rule{0pt}{3ex}
   & RGCR  & $\frac{p}{(\dmax+1) \kappa }$ & $\frac{1}{n} \dmax^2 \kappa^4 p^{-1}$ 
\\ \rule{0pt}{3ex}
$\bw^*$-weighted   & RGCR & $\frac{p}{\lambda^*}$ & $\frac{1}{n} \lambda^* \dmax \kappa^3 p^{-1} $ 
   \vspace{1ex}
\\ \bottomrule
\end{tabular}
\caption{
A summary of bounds pertaining to the HT estimator of the GATE under various
randomization designs. 
The RGCR results apply for \emph{both} independent and complete randomization,
while the GCR bounds do \emph{not} support complete randomization because they violate
the positivity assumption. 
Each variance upper bound is up to a $\Theta(1)$ multiplicative constant.
Details are given in the corresponding subsections of \Cref{sec:properties}.
}
\label{tab:summary_theory}
\end{table}

In this section, we analyze the properties of the RGCR scheme. 
We focus on the Horvitz--Thompson (HT) estimator due to its theoretical amenability,
while some important insights on the \hajek\ estimator are discussed at the end.

Since the RGCR scheme requires a randomized clustering strategy,
we first consider two initial algorithms:
randomized $3$-net, a randomized version of the $3$-net algorithm considered
in the original analysis of the GCR scheme,
and $1$-hop-max, a new randomized clustering algorithm similar to $3$-net
but more easily amenable to a rigorous analysis.
We then also consider weighted versions of these two algorithms, 
which introduces node-level flexibility and can effectively 
balance the exposure probabilities of high- and low-degree nodes, 
addressing an imbalance found in the first two algorithms. 
The goal of this section is to provide an analysis of how
RGCR can lead to considerable variance reduction when 
compared with the vanilla GCR scheme based on a single clustering.
All but the simplest proofs are removed to Appendix~\ref{app:proofs}.

We summarize the results of this section in \Cref{tab:summary_theory}
and highlight some important observations. First, for each clustering algorithm,
by using  GCR with a single fixed clustering, the variance of the HT estimator is
upper bounded by an exponential function of either $\dmax$ or $\kappa$. 
Note that both quantities can be large in real-world networks, resulting in the
huge variance in the original GCR scheme. In contrast, with RGCR,
the variance is upper bounded by a polynomial function of $\dmax$ and $\kappa$.
Recall that if the graph has bounded degree $\dmax$ but the growth is not ``further''
restricted then we still have that $\kappa < \dmax$.
Therefore, the RGCR scheme can significantly reduce the 
estimator variance compared with GCR,
both with and without restricted growth.

Second, we highlight that variance reduction is achieved primarily by 
obtaining a much larger exposure probabilities, 
which are the inverse weights in the HT estimator 
and play a similar role in the H\'ajek estimator. 
With a fixed clustering, a node can be at the boundary of a cluster,
making it adjacent to many clusters and thus the exposure probability becomes 
exponentially small. However, with RGCR, such exponentially small
probabilities are ``washed out" by averaging with the clusterings where a node
is at the center of a cluster, and even have a tidy lower bound.

Finally, for each random clustering algorithm considered, completed randomization
is valid for RGCR, \ie, positivity (Assumption~\ref{assm:positive}) is satisfied. 
In contrast, the positivity assumption is generally 
violated in GCR with complete randomization. The results for RGCR summarized in
\Cref{tab:summary_theory} apply for both independent and complete randomization,
 while those for GCR apply only for independent randomization.

Beside extensive analysis on the HT estimator, we also present some key properties
of the \hajek~estimator under the GCR and RGCR schemes. Compared with the HT
estimator, \hajek~estimator enjoys much lower variance due to the self normalization,
while a potential drawback, widely known in the literature, is the potential issue of bias.
As a highlight of our discussion, we show that the \hajek~estimator is unbiased
under GCR and RGCR if the individual treatment effect
$\tau_i = Y_i(\bOne) - Y_i(\bZero)$ is constant across all nodes.
However, in practice the treatment effects are reasonably non-constant,
making the \hajek~estimator potentially biased. 
This result motivates us to use a non-constant individual treatment effect
to study the bias of \hajek~estimator in simulation experiments in \Cref{sec:simulation}.


\subsection{Randomized $3$-net and $1$-hop-max clusterings}
We now study our two random clustering algorithms and 
establish properties of a RGCR design
when each clustering algorithm is used. 
For notation brevity, our analysis is always conditioned on the distribution
of random clusterings in focus, unless stated otherwise.

\subsubsection{Algorithms}

\begin{figure}[t]
  \centering
\begin{minipage}{.75\linewidth}
\begin{algorithm}[H]
\DontPrintSemicolon
    \KwIn{Graph $G = (V, E)$}
    \KwOut{Graph clustering $\bc \in \real n$}
    $\pi \gets$ generate a uniformly random total ordering of all nodes\;
    $S \gets \emptyset$, unmark all nodes\;
    \For{$i \in \pi$} {
      \If{$i$ is unmarked} {
        $S \gets S \cup \{i\}$\;
        \For{$j \in B_2(i)$} {
          mark node $j$ if it is unmarked yet\;
        }
      }
    }
    \For{$i \in V$} {
      $c_i \gets \arg \min\{j \in S, j \rightarrow\dist(i, j) \}$, \ie, the id of the node in $S$
         with shortest graph distance to $i$ (arbitrary tie breaking)\;
    }    
    \Return $\bc$\;
\caption{3-net clustering.}
\label{Alg:3_net}
\end{algorithm}
\end{minipage}
\end{figure}

The first algorithm in consideration is the $3$-net clustering which is used in the
original analysis of the graph cluster randomization scheme~\cite{ugander2013graph}. 
Here we assume that a $3$-net clustering is generated from a random ordering of all nodes
and thus its output is random, while such randomness was not exploited 
in any part of the analysis of vanilla GCR, which was conditional on a single clustering
outputted by the algorithm.

Formally the randomized $3$-net clustering algorithm is given in \Cref{Alg:3_net}, 
which consists of three major steps. First, we generate a total ordering of all nodes 
sampled uniformly over all permutations. 
Second, construct a maximal distance-3 independent set of the network
(line 2--7) using a greedy algorithm proceeding according to the total ordering generated in line 1.
We call each node in the independent set a \emph{seed} node. 
Next we assign every node in the network to the seed node with smallest graph distance,
with ties broken by some arbitrary rule. These steps return a clustering partition.

In the returned clustering, since the seed nodes form a distance-3 independent set,
any 1-hop neighbors of a seed node will be assigned to the seed.
Therefore, the seeds nodes are guaranteed to be in the interior of a cluster, 
not connecting to any nodes in a different cluster. 
Consequently, the returned clustering consists of node-neighborhood clusters
known to form relatively good clusters (in terms of edges cut) 
in real-world networks~\cite{gleich2012vertex,yin2019local}.

A potential disadvantage of 3-net clustering algorithm is the runtime.
Even though parallel algorithms have been developed for the random 
maximal independent set problem~\cite{alon1986fast,blelloch2012greedy}, 
the runtime still increases with the size of the network, and thus
it is generally slow to sample a random 3-net clustering on a very large network, 
even by more complicated means. 

\begin{figure}[t]
  \centering
\begin{minipage}{.75\linewidth}
\begin{algorithm}[H]
\DontPrintSemicolon
    \KwIn{Graph $G = (V, E)$}
    \KwOut{Graph clustering $\bc \in \real n$}
    \For{$i \in V$} {
      $X_i \gets \mathcal U(0, 1)$\;
    }
    \For{$i \in V$} {
      $c_i \gets \max([X_j \text{ for } j \in B_1(i)])$\;
    }    
    \Return $\bc$\;
\caption{1-hop-max clustering.}
\label{Alg:1_hop_max}
\end{algorithm}
\end{minipage}
\end{figure}

As a second algorithm for RGCR, we propose  \emph{1-hop-max}, 
given in \Cref{Alg:1_hop_max}. This algorithm consists of two steps. 
First, every node $i$ independently generates a random number from the 
uniform distribution on $(0, 1)$.
Second, for every node $i$, find the maximum
of the generated numbers within node $i$'s 1-hop neighborhood. The unique
numbers define the clustering: 
nodes with the same 1-hop-maximum form a single cluster.

Similar to the $3$-net algorithm, the clustering returned by the 1-hop-max algorithm contains
neighborhood-like clusters: every cluster is associated with a center node.
On the other hand, the 1-hop-max algorithm has a much faster parallel runtime.
Formally, we have the following result in terms of the \emph{work} 
(\ie, total number of operations) and \emph{depth} (\ie, length of longest chain 
in the computation dependency graph)~\cite{blelloch1996programming}, 
key constraints in parallel computing.

\begin{theorem}   \label{Prp:Dep_u_max}
\Cref{Alg:1_hop_max} has $O(\log(\dmax))$ depth and $O(m)$ work.
\end{theorem}

\subsubsection{Network exposure probabilities}
\label{sec:expo_prob_unweighted}

The network exposure probabilities $\prob{E_i ^\bz \mid \mathcal P}$ of these algorithms
are key parts of the HT and H\'ajek GATE estimators under RGCR. 

Before discussing how to compute or estimate these probabilities, 
we first show a simple but useful lower bound of the full neighborhood exposure 
probabilities when using 3-net or 1-hop-max random clustering generator.
This result is crucial in 
both the analysis of a Monte Carlo method for estimating the probabilities 
in \Cref{sec:expo_prob_estimate}
and the variance analysis in \Cref{sec:var_unweighted}.

\begin{theorem}   \label{Thm:prob_LB}
Using either 3-net or 1-hop-max random clustering on a graph
with restricted growth coefficient $\kappa$, using either independent
or complete randomization at the cluster level,
the full-neighborhood exposure probabilities for any node $i$ satisfy
\[
\prob{E_i ^{\bOne} \mid \mathcal P} \geq \frac{p}{\lvert B_2(i)\rvert} \geq \frac{p}{(1+\dmax)\kappa}, 
\quad
\prob{E_i ^{\bZero} \mid \mathcal P} \geq \frac{1-p}{\lvert B_2(i)\rvert} \geq \frac{1-p}{(1+\dmax)\kappa}.
\]
\end{theorem}
A detailed proof is given in the \Cref{app:proofs}, while the high-level idea is as follows.
If a node $i$ is ranked first within $B_2(i)$ in a 3-net clustering algorithm  
(or generated the largest number in the 1-hop-max algorithm),
which happens with probability $1/\lvert B_2(i) \rvert$,
then all its 1-hop neighbors are guaranteed to be in the same cluster as node $i$,
and thus it is definitely network exposed to either treatment or control.

Several remarks are in order on the above result. First, this lower bound 
is much higher than an analogous lower bound for the GCR scheme.
With GCR, 3-net clustering, and independent randomization 
(but not complete randomization), we have
$\prob{E_i ^\bOne \mid \mathcal G} \geq p^{\dmax}$ in general and
$\prob{E_i ^\bOne \mid \mathcal G} \geq p^{\kappa^6}$ under a
restricted growth condition~\cite{ugander2013graph}.
That lower bound is exponentially small in the restrictive growth parameter
$\kappa$. In real-world networks, $\kappa$ can be of magnitude of 100,
making the exposure probabilities impossibly small. In contrast, with RGCR,
the exposure probability is lowered bounded by a polynomial function
of $\dmax$ and $\kappa$.

As a second remark, these lower bounds also hold
when we consider a partial neighborhood exposure model. 
If a node is full-neighborhood exposed, it must also be partial-neighborhood
exposed, and thus the partial-neighborhood exposure probability of each node
is no lower than that for full-neighborhood exposure.

As a third remark,
another significant implication of \Cref{Thm:prob_LB} is that it provides
a positive lower bound on the node-level exposure probabilities, making
complete randomization feasible. Note that complete randomization
is not feasible for the GCR scheme due to violation of the positivity assumption.
 However, for RGCR scheme, according to \Cref{Thm:prob_LB}, even under
complete randomization, the exposure probability of each node is guaranteed 
to be positive. 

The exposure probability lower bound in \Cref{Thm:prob_LB}
is obtained by solely considering scenario when a node generates
the largest number in its 2-hop neighborhood. Actually, one can obtain
an improved lower bound from more careful consideration on
node's ranking among its 2-hop neighborhood.

\begin{theorem}   \label{Thm:prob_LB_improved}
With 1-hop-max random clustering algorithm and independent randomization at the cluster level,
if $\lvert B_2(i) \rvert - \degree{i} \geq 1/(1-p)$, then
the full-neighborhood exposure probabilities for any node $i$ satisfy
\begin{eqnarray*}
\prob{E_i ^{\bOne} \mid \mathcal P} 
\geq \frac{1}{\lvert B_2(i)\rvert} \cdot \frac{p}{1-p}
,\qquad
\prob{E_i ^{\bZero} \mid \mathcal P} 
\geq \frac{1}{\lvert B_2(i)\rvert} \cdot \frac{1-p}{p}
.
\end{eqnarray*}
\end{theorem}
The proof of this result involves a more carefuly analysis and for $p=1/2$ the difference between the lower bounds in \Cref{Thm:prob_LB_improved} and \Cref{Thm:prob_LB} is merely a factor of 2.

\subsubsection{Estimating the exposure probabilities}
\label{sec:expo_prob_estimate}

Computing the exact network exposure probabilities can be challenging
as it potentially requires considering an exponential number of different clusterings 
in \Cref{Eq:ExpoProb_mix}. More formally, \Cref{Thm:Hard-epsnet} show that
with 3-net clustering, computation of the exact 
exposure probability for a single node is
NP-hard.
\begin{theorem}   \label{Thm:Hard-epsnet}
For the $3$-net random clustering algorithm, using either independent or complete randomization at the cluster level, 
exact computation of the full-neighborhood exposure probability for a node in 
an arbitrary graph is NP-hard.
\end{theorem}
Note that even though we don't have an analogous rigorous proof for the 1-hop-max clustering
strategy, we expect the analogous exposure probability computations to also be NP-hard.

Despite this negative result, 
the network exposure probabilities can be estimated 
using a relatively straight-forward Monte Carlo method with theoretical guarantees. 
The procedure begins by generating $K$ clusterings $\{\bc^{(k)}\}_{k=1}^K$ from our randomized 
clustering algorithm and compute the exact exposure probability of each node under each clustering.
The estimator of the exposure probability is then
\begin{equation}   \label{Eq:prob_estimator_MC}
\hat{\mathbb{P}}[E_i ^\bz \mid \mathcal P] = \frac{1}{K} \sum_{k = 1}^K \prob{E_i ^\bz \mid \bc^{(k)}}.
\end{equation}
We then have the following result 
on the mean square error (MSE) of relative error in this Monte Carlo estimator.
\begin{theorem}   \label{Thm:prob_MC_var}
For either 3-net or 1-hop-max random clustering algorithm, and with 
$K$ Monte-Carlo trials and any node $i$, the relative error of the
Monte-Carlo estimator is upper bounded in MSE as
\[
\expect{\frac{\hat{\mathbb{P}}[E_i ^\bOne \mid \mathcal P] - \prob{E_i ^\bOne \mid \mathcal P}}{\prob{E_i ^\bOne \mid \mathcal P}} ~\middle |~ \mathcal P}^2
\leq \frac{\lvert B_2(i)\rvert}{Kp}.
\]
\end{theorem}
The proof is given in the \Cref{app:proofs}, which is obtained from
the fact that the ground-truth exposure probability is bounded away from 0
as is shown in \Cref{Thm:prob_LB}.

Given this MSE guarantee, it is natural to use the 
estimated exposure probabilities as the inverse weights in (e.g.) an HT estimator. 
A potential issue is the possible violation of the positivity assumption for 
complete randomization: it is possible that for some node $i$, 
$\prob{E_i ^\bOne \mid \bc^{(k)}} = 0$ for all the generated clusterings,
and thus $\hat{\mathbb{P}}[E_i ^\bOne \mid \mathcal P] = 0$ (which
would make the HT estimate ill-defined).
A fix to this positivity issue is to use stratified sampling in generating the clustering samples.

To stratify our Monte Carlo estimator, we 
generate $Kn$ samples $\{\bc^{(k, i)}\}$ with $k \in \{1, 2, \dots, K\}$
and $i \in \{1, 2, \dots, n\}$ such that, 
if the 3-net clustering is in use, then the clustering $\bc^{(k, i)}$ is based on
a random node ordering conditional on node $i$ being ranked first among all nodes.
Analogously, if the 1-hop-max clustering is in use, then in the generation
of clustering $\bc^{(k, i)}$, node $i$ generates the largest $X_i$ among all nodes.
Consequently, under clustering $\bc^{(k, i)}$, node $i$ is guaranteed to be
the center of a cluster and thus $\prob{E_i ^\bOne \mid \bc^{(k, i)}} = p$.
Now the network exposure probability of node $i$ is estimated as
\[
\hat{\mathbb{P}}[E_i ^\bz \mid \mathcal P] 
= \frac{1}{nK} \sum_{k = 1}^K \sum_{j=1}^n  \prob{E_i ^\bz \mid \bc^{(k, j)}}.
\]
In total, each node $i$ is ``favored" exactly $K$ times among the $Kn$ samples,
and we have 
\[
\hat{\mathbb{P}}[E_i ^\bOne \mid \mathcal P]  \geq \frac pn > 0.
\]
Besides a guarantee of positivity in the estimated
exposure probabilities, this stratified sampling technique is also 
effective at reducing variance in the estimation. 
Therefore, when computationally feasible to sample at least $n$ clustering samples,
this stratified sampling method should be strictly preferred over independent sampling.

As a final but important note on probability computation and estimation, 
we point out that the potential
computational bottleneck of generating $K$ clusterings when using RGCR
should not pose practical concerns.
First, we highlight that the exposure probabilities are needed 
only in the analysis phase but not the design phase. To launch an experiment,
it suffices to generate a single clustering from a randomized algorithm 
and use it in assigning individuals to treatment or control; 
after the experiment has been launched, we can later sample other random clusterings 
to estimate the exposure probabilities. 
Second, we note that the estimated exposure probabilities can be shared 
across experiments as long as the interference network remains unchanged.
In practice, with hundreds of A/B testings running at the same time, practitioners
only need to estimate the exposure probabilities once.

\subsubsection{Variance of estimators}
\label{sec:var_unweighted}

We now analyze the variance of the Horvitz--Thompson (HT) estimator with RGCR.
We show that, with 1-hop-max clustering, the variance is 
upper bounded by a polynomial function in both the maximum degree
$\dmax$ and the restricted growth parameter $\kappa$,
which also decays as $n \rightarrow \infty$. 

We first present a useful property of the randomized $1$-hop-max clustering
algorithm, the local dependence, which distinguished it from $3$-net clustering.
\begin{lemma}   \label{Lem:1_hop_max-local_dep}
With 1-hop-max random clustering algorithm, for any node $i$, the joint distribution
of $\bC_{B_1(i)}$, \ie, the clusterings of all nodes in $B_1(i)$,
depends only on the structure of the graph induced 
on the node set $B_2(i)$.
\end{lemma}
\begin{proof}
Since the clustering of every node is $C_j = \max\{X_{j^\prime}: 
j^\prime \in B_1(j)\}$ with $X_{j^\prime} \sim \mathcal U(0, 1)$, the joint distribution of
$[C_j]_{j \in B_1(i)}$ depends only on the structure of the graph induced 
on the node set $B_2(i)$ and is independent of the rest of network.
\end{proof}

With this local dependence property, now we present the following result on the variance of mean-outcome HT estimator.
\begin{theorem}   \label{Thm:var_restricted_growth_general}
For RGCR with a 1-hop-max clustering, if every node's responses are within $[0, \bar Y]$ then
\[
\var{\hat \mu _{\mathcal P}(\bOne) }
\leq \frac{\bar Y^2}{n^2} \sum_{i=1}^n \frac{\lvert B_4(i) \rvert}{\prob{E_i ^\bOne \mid \mathcal P}}.
\]
for both independent and complete cluster-level randomization.
\end{theorem}
As an intuition for this result, by local dependence we have that 
the full-neighborhood exposure events of two nodes become independent 
(or negatively correlated) events if their graph distance is sufficiently big. 
This observation limits many cross-terms of the variance formula (\Cref{Eq:var-mean_outcome}),
yielding an upper bound. A formal proof is given in \Cref{app:proofs}.

A corollary of \Cref{Thm:var_restricted_growth_general} is the following.

\begin{theorem}   \label{Thm:var_restricted_growth_unweighted}
For RGCR with 1-hop-max clustering on a graph with maximum degree $\dmax$ and restricted growth coefficient $\kappa$, If every node's responses are within $[0, \bar Y]$ then
\[
\var{\hat \mu _{\mathcal P}(\bOne) }
\leq \frac{1}{n} \cdot  \bar Y^2 (d_{\max}+1) ^2 \kappa^4 p^{- 1},
\]
for both independent and complete cluster-level randomization.
\end{theorem}
\begin{proof}
From \Cref{Thm:var_restricted_growth_general} we have
\[ 
\var{\hat \mu _{\mathcal P}(\bOne)}
\leq \frac{\bar Y ^2}{n^2}  \sum_{i=1}^n \left[ \frac{\lvert B_2(i) \rvert}{p} \cdot \lvert B_4(i) \rvert  \right] \leq \frac 1 n \cdot \bar Y^2 (1 + d_{\max})^2 \kappa^4 p^{-1},
\]
where the first inequality is due to $\prob{E_i ^\bOne \mid \mathcal P}\geq \frac{p}{B_2(i)}$
and the second inequality is due to $\lvert B_r(i) \rvert \leq (1 + d_{\max}) \kappa^{r-1}$.
\end{proof}

This upper bound is to be compared with \Cref{Eq:GCRVarUB}, the variance upper bound 
when using a single fixed clustering, which is exponential to the restrictive growth coefficient $\kappa$. 
In contrast, if a random graph clustering is used, the upper bound is a polynomial function of $\kappa$. 
This result provides a strong theoretical justification of variance reduction from using random graph partitioning in GCR. 

From the variance of the mean outcome estimator we can obtain the following
variance upper bound on the GATE estimator.

\begin{theorem}   \label{Thm:var_GATE_restricted_growth}
For RGCR with 1-hop-max clustering on a graph with maximum degree $\dmax$ and restricted growth coefficient $\kappa$, If every node's responses are within $[0, \bar Y]$ then
\[
\var{\hat \tau _{\mathcal P} } \leq \frac{2}{n} \cdot  \bar Y^2 (d_{\max}+1) ^2 \kappa^4 (p^{- 1} + (1-p)^{-1}),
\]
for both independent and complete cluster-level randomization.
\end{theorem}

All of our analysis thus far has been non-asymptotic (finite-$n$) results. As such, we have not assumed that $\dmax$ or $\kappa$ are fixed in $n$. As a corollary of \Cref{Thm:var_GATE_restricted_growth} then, we have the following sufficient condition for convergence of the HT GATE estimator, which extends beyond the regime of bounded-degree graphs.

\begin{theorem}   \label{Thm:convergence_GATE}
Let $G_n$ be a sequence of graphs on $n$ nodes with maximum degree $\dmax$ and restricted growth coefficient $\kappa$ both possibly dependent on $n$.  Let all responses be within $[0, \bar Y]$. Then for RGCR with 1-hop-max, a fixed cluster-level randomization probability $p$, and either:
\begin{itemize}
\item $\kappa$ fixed, $\dmax$ = $o(n^{1/2})$ or
\item $\kappa, \dmax$ = $o(n^{1/6})$,
\end{itemize}
we have
$\var{\hat \tau _{\mathcal P} } \rightarrow 0
$
as $n \rightarrow \infty$, for both independent and complete cluster-level randomization.
\end{theorem}

If $\kappa$ is fixed then the analogous sufficient condition for GCR (from \Cref{Eq:GCRVarUB}) requires $\dmax$ to be only $o(n)$. But if $\dmax$ and $\kappa$ are of similar order---as appendix~\ref{app:growth} suggests they often are empirically in social networks---the analogous GCR sufficient condition requires $\dmax$ to be $o(\log n)$, a significantly stronger requirement than under RGCR. 

The proof of the variance upper bound in \Cref{Thm:var_GATE_restricted_growth} does not apply to RGCR under a
randomized 3-net clustering. 
The reason the analysis breaks down is that local dependence (\Cref{Lem:1_hop_max-local_dep})
does not hold for the 3-net clustering algorithm. Specifically, the distribution of $\bC_{B_1(i)}$
depends on the structure of the whole network. 
For example, adding a single edge could make the incident nodes
less likely to be part of the seed set, and such a change of probability then 
reaches across the entire network, making each node more or less likely to be part of the seed set. 
Despite blocking our theoretical analysis, we still expect randomized $3$-net clustering
to undergo similar variance reduction
when using randomized clustering versus a single fixed clustering.
In \Cref{sec:simulation}, we show via simulation that the variance of RGCR with 3-net clustering
is much lower than with GCR, and it is in fact lower than that of RGCR with 1-hop-max clustering.


\subsection{Weighted randomized 3-net and 1-hop-max clusterings}

A drawback of both the 3-net and 1-hop-max clustering algorithms, shared by many 
existing approaches, is an implicit disadvantage for high-degree nodes: compared to
low-degree nodes
they are invariably connected to many more clusters and thus have much 
smaller exposure probabilities.
This phenomenon is supported by 
\Cref{Thm:prob_LB}, where we showed a exposure probability lower bound that
decreases with the size of its two-hop neighborhood.
Per \Cref{Thm:var_restricted_growth_general}, the smallest exposure probabilities 
(and thus, those for high degree nodes)
 dominate the variance in HT estimators.

To counteract the outsized contribution of high-degree nodes to the variance, 
we propose a weighted variant of both random clustering algorithms that
introduce additional node-level flexibility to adjust and balance the exposure probability of nodes. 
In particular, we can choose to prioritize high-degree nodes in these weighted clustering algorithms.
After introducing the algorithms in \Cref{sec:weighted_alg}, we presents
properties of these algorithms when they are used in RGCR,
highlighting similarities and differences when compared to unweighted counterparts.

\subsubsection{Algorithms}
\label{sec:weighted_alg}

Recall that, in the 1-hop-max clustering algorithm (\Cref{Alg:1_hop_max}),
we first independently generate a random number from the uniform distribution
and construct a clustering based on these generated random numbers: nodes with
higher numbers dominate their neighbors and are more likely to be in the center of a cluster.
Since the numbers are generated from a uniform distribution,
the probability that a given node generates a larger number than any other is always 1/2,
making higher-degree nodes less likely to dominate all their neighbors.

\begin{figure}[t]
  \centering
\begin{minipage}{.75\linewidth}
\begin{algorithm}[H]
\DontPrintSemicolon
    \KwIn{Graph $G = (V, E)$,  node weights $\bw \in \real{n}_+$.}
    \KwOut{Graph clustering $\bc \in \real n$}
    \For{$i \in V$} {
      $X_i \gets \mathcal \beta(w_i, 1)$\;
    }
    \For{$i \in V$} {
      $c_i \gets \max([X_j \text{ for } j \in B_1(i)])$\;
    }    
    \Return $\bc$\;
\caption{Weighted 1-hop-max clustering.}
\label{Alg:1_hop_max-w}
\end{algorithm}
\end{minipage}
\end{figure}

Our proposed fix to this problem is to change the first step of the algorithm,
generating numbers $X_i$ from a different non-uniform distribution at each node.
Let each node $i$ be associated with a weight $w_i > 0$ and then generate 
its number from a Beta distribution, $X_i \sim \beta(w_i, 1)$.
The full algorithm of weighed 1-hop-max is given in \Cref{Alg:1_hop_max-w}.

To understand the intuition behind the weighted scheme, we first note the following 
basic and well-known properties of the beta distribution, proven for completeness in \Cref{app:proofs}.
\begin{theorem}   \label{Prp:beta_dist_prop}
For independent random variables $X_i \sim \beta(w_i, 1)$, $X_j \sim \beta(w_j, 1)$,
we have 
\begin{enumerate}  [(a)]
\item $\prob{X_i > X_j} = \frac{w_i}{w_i + w_j}$,
\item $\max\{X_i, X_j\} \sim \beta(w_i + w_j, 1)$.
\end{enumerate}
\end{theorem}

According to part (a) of \Cref{Prp:beta_dist_prop},
a node with a larger weight is more likely to generate a larger number.
Thus, by adopting larger weights at high degree nodes,
we can make the large degree nodes more likely to dominate their neighbors,
correcting their disadvantage in the unweighted scheme.

\begin{figure}[t]
  \centering
\begin{minipage}{.75\linewidth}
\begin{algorithm}[H]
\DontPrintSemicolon
    \KwIn{Graph $G = (V, E)$,  node weights $\bw \in \real{n}_+$.}
    \KwOut{Graph clustering $\bc \in \real n$}
    \For{$i \in V$} {
      $X_i \gets \mathcal \beta(w_i, 1)$\;
    }    
    $\pi \gets \arg\sort([X_i]_{i \in V}, descend)$\;
    $S \gets \emptyset$, unmark all nodes\;
    \For{$i \in \pi$} {
      \If{$i$ is unmarked} {
        $S \gets S \cup \{i\}$\;
        \For{$j \in B_2(i)$} {
          mark node $j$ if it is unmarked yet\;
        }
      }
    }
    \For{$i \in V$} {
      $c_i \gets \arg \min\{j \in S, j \rightarrow\dist(i, j) \}$, \ie, the id of the node in $S$
         with shortest graph distance to $i$ (arbitrary tie breaking)\;
    }    
    \Return $\bc$\;
\caption{Weighted 3-net clustering.}
\label{Alg:3_net-w}
\end{algorithm}
\end{minipage}
\end{figure}

This idea of node weighting can also be applied to 3-net clustering. 
In the unweighted version, we first generate a uniform random ordering
of all nodes, which is used to form a seed set and partition the network.
In a uniform random ordering where each node has an equal probability of
ranking first, large degree nodes are at disadvantage of being selected into
the seed set and being the center of a cluster, and thus less likely to be
network exposed. To compensate for this disadvantage,
we can generate a non-uniform
random ordering where large degree nodes are more likely to rank high.
A non-uniform random ordering can be generated by
a combination of Beta-distributed samples and sorting. 
Specifically, if each node $i$
is associated with a weight $w_i$, then we can first generate $X_i \sim \beta(w_i, 1)$,
and sort the samples in decreasing order. In this way, nodes
associated with a larger weight are more likely to rank higher
after sorting. Formally this weighted 3-net clustering algorithm is given in
\Cref{Alg:3_net-w}.

We note two connections between the weighted 3-net and 1-hop-max
clustering algorithms and their original unweighted versions. First, the weighted
version can be considered an extension of the unweighted algorithms:
when all nodes have the same weight, the weighted 3-net and 1-hop-max
algorithm are equivalent to the original algorithm. Second, for either
3-net or 1-hop-max clustering, the distribution of the random clustering
returned from the unweighted and weighted algorithms have the same
support, \ie, for clusterings that has nonzero probability of being generated
from the unweighted version, the probability of being generated from
the weighted version is also nonzero, and vice versa. The difference lies in,
certain clusterings are more or less likely to be generated in the weighted
version. Consequently, conditioning on the generated clustering and using it
in a GCR scheme, there is no difference between which version is
used to generate the clustering. However, in RGCR, which is based on
a distribution of clusterings, the weighted version might have superior
properties due to its node-level adjustments.

\subsubsection{Properties with arbitrary node weights}
\label{sec:properties_weighted}

In this section, we discuss properties of the weighted 3-net and 1-hop-max
algorithms with an arbitrary set of node weights. The result motivates our discussion
on good choices of node weights in section that follows.

First, we have the following lower bound on exposure probabilities
at each node. Similar to \Cref{Thm:prob_LB}, the result is based on analyzing the probability
that a node is ranked first in its 2-hop-neighborhood. 
The proof is given in \Cref{app:proofs}.
\begin{theorem}   \label{Thm:prob_LB_weighted}
With the weighted 3-net or 1-hop-max random clustering algorithm, using either independent
or complete randomization at the cluster level,
the full-neighborhood exposure probabilities for any node $i$ satisfy
\[
\prob{E_i ^{\bOne} \mid \mathcal P} \geq \frac{w_i}{\sum_{j \in B_2(i) }w_j} \cdot p, 
\quad
\prob{E_i ^{\bZero} \mid \mathcal P} \geq \frac{w_i}{\sum_{j \in B_2(i) }w_j} \cdot (1-p).
\]
\end{theorem}

When all nodes have equal weights, then the weighted 3-net and 1-hop-max
algorithm degenerates to the original version, making \Cref{Thm:prob_LB_weighted}
a generalization of \Cref{Thm:prob_LB}.

Computing the exposure probability of each node might now be challenging,
but we again show that Monte Carlo estimation,
as in \Cref{Eq:prob_estimator_MC}, can efficiently achieve low relative error.
\begin{theorem}   \label{Thm:prob_MC_var_weighted}
Using either weighted 3-net or weighted 1-hop-max random clustering algorithm, and with 
$K$ Monte-Carlo trials, for any node $i$, the relative error of the
Monte-Carlo estimator is upper bounded in MSE as
\[\textstyle
\expect{\frac{\hat{\mathbb{P}}[E_i ^\bOne \mid \mathcal P] - \prob{E_i ^\bOne \mid \mathcal P}}{\prob{E_i ^\bOne \mid \mathcal P}} ~\middle |~ \mathcal P}^2
\leq \frac{1}{Kp} \cdot \frac{\sum_{j \in B_2(i) }w_j}{w_i}.
\]
\end{theorem}

As before, stratified sampling can also be adapted for the weighted clustering methods.
Similar to the procedure in \Cref{sec:expo_prob_estimate}, we generate
$Kn$ clustering samples $\{\bc^{(k, i)}\}$, where in clusterings $\bc^{(k, i)}$, 
$k \in \{1, 2, \dots, n\}$, node $i$ is ``favored" and deterministically placed first.
Note that the likelihood of node $i$ naturally generating the largest draw is proportional
to $w_i$, per \Cref{Prp:beta_dist_prop}. The sample $\bc^{(k,j)}$ 
should be weighted accordingly. The estimated exposure probabilities should then be
\[
\hat{\mathbb{P}}[E_i ^\bz \mid \mathcal P] 
= \frac{ \sum_{k = 1}^K \sum_{j=1}^n w_j \prob{E_i ^\bz \mid \bc^{(k, j)}}}{K \sum_{j=1}^n w_j}.
\]
Again this stratified method is preferred over Monte Carlo estimation 
with independent samples since it guarantees positivity in the estimated 
exposure probabilities and reduces variance in the probability estimation.

\subsubsection{Choice of node weights}
\label{sec:weighted_opt_choice}

With the node-level flexibility in the weighted 3-net and 1-hop-max clustering,
a natural subsequent question is to find a good choice of node weights. 
In this section, we discuss two heuristics which lead to different sets of node weights.
The first heuristic suggests node weights based on the eigenvector of an eigenvalue problem
associated with the network's squared adjacency matrix. The second heuristic suggests uniform weights, i.e., 
the unweighted versions of the algorithms.

\xhdr{Maximizing the minimal exposure probability lower bound} 
As is discussed in the previous sections, high-degree nodes are less likely
than low-degree nodes to be network exposed 
using the unweighted 3-net or 1-hop-max clustering.
To correct this disadvantage, it might be ideal if all nodes have the
same exposure probability, or at least the same lower bound. 

Given a graph $G=(V,E)$, let $G_2 = (V, E_2)$ denote the ``squared'' graph,
\ie, with the same node set $V$, and an edge $(i, j) \in E_2$ if $i \in B_2(j)$ 
in the original network.
The adjacency matrix of $G_2$ is an irreducible non-negative matrix,
and according to the Perron-Frobenius theorem, its spectral radius,
denoted as $\lambda^*$, is also its largest positive eigenvalue. Moreover,
for the eigenvector $\bw^*$ associated with this eigenvalue, \ie,
\begin{equation}   \label{Eq:EigenVec} \textstyle
\sum_{j \in B_2{(i)}} w^*_j   = \lambda^* w_i ^*,
\end{equation}
all the elements $w_i ^*$ are positive. Therefore, $\bw^*$ provides
a valid set of node weights, which  we call the \emph{spectral weights}.

Using these spectral weights in the weighted 3-net or 1-hop-max
scheme, we show that as a
corollary of \Cref{Thm:prob_LB_weighted,Eq:EigenVec} (the proof logic is identical),
all nodes now have the same exposure probability lower bound.

\begin{theorem}   \label{Thm:prob_LB_spectral}
With the spectral-weighted 3-net or 1-hop-max random clustering algorithm, 
using either independent or complete randomization the cluster level,
the full-neighborhood exposure probabilities for any node $i$ satisfy
\[
\prob{E_i ^{\bOne} \mid \mathcal P} \geq \frac{p}{\lambda^*}, 
\quad
\prob{E_i ^{\bZero} \mid \mathcal P} \geq \frac{1-p}{\lambda^*},
\]
a uniform lower bound on the full neighborhood exposure probability
of all nodes.
\end{theorem}

We then have the following corollary (of \Cref{Thm:var_restricted_growth_unweighted}) 
upper bound on the variance 
of HT GATE  estimators using RGCR with spectral-weighted 1-hop-max
random clustering.

\begin{theorem}   \label{Thm:var_restricted_growth_spectral}
Using RCGR with spectral-weighted 
1-hop-max clustering, if every node's response is within $[0, \bar Y]$ then
\[
\var{\hat \mu _{\mathcal P}(\bOne) } \leq \frac{1}{n} \cdot  \bar Y^2 \lambda^* (d_{\max}+1)  \kappa^3 p^{- 1},
\]
for both independent and complete cluster-level randomization.
\end{theorem}
\begin{proof}
We first note that, with an identical proof, one can verify that
\Cref{Thm:var_restricted_growth_general}
also hold with the weighed 1-hop-max clustering with any weights $\bw$.
Now similar to the proof of \Cref{Thm:var_restricted_growth_unweighted}, we have
\[ 
\var{\hat \mu _{\mathcal P}(\bOne)}
\leq \frac{\bar Y ^2}{n^2}  \sum_{i=1}^n \left[ \frac{\lambda^*}{p} \cdot \lvert B_4(i) \rvert  \right] \leq \frac 1 n \cdot \bar Y^2 \lambda^* (1 + d_{\max}) \kappa^3 p^{-1},
\]
where the first inequality is due to the exposure probability lower bound in 
\Cref{Thm:prob_LB_spectral}.
\end{proof}

As a final corollary, we have the following upper bound on the variance of the HT GATE estimator, by a proof 
identical to that of \Cref{Thm:var_GATE_restricted_growth}.
\begin{theorem}   \label{Thm:var_GATE_restricted_growth_spectral}
Using RCGR with spectral-weighted 
1-hop-max clustering, if every node's response is within $[0, \bar Y]$ then
\[\textstyle
\var{\hat \tau _{\mathcal P} } \leq \frac{2}{n} \cdot  \bar Y^2 \lambda^* (d_{\max}+1) \kappa^3 (p^{- 1} + (1-p)^{-1}),
\]
for both independent and complete cluster-level randomization.
\end{theorem}

Of note, according to the Perron-Frobenius theorem, we also have 
\[
\lambda^* \leq \max_i (\lvert B_2(i) \rvert) \leq (d_{\max}+1) \kappa.
\] 
As a result, this variance upper bound using spectral-weighted 1-hop-max clustering
can be used to furnish the variance upper bound for the unweighted 1-hop-max clustering  
(\Cref{Thm:var_GATE_restricted_growth}) as well.
These final inequalities are not necessarily strict improvements---they become equalities for a regular graph---but 
in practical settings they can lead to sizable improvements over unweighted clustering methods.

Having the same exposure probabilities at each node is ideal, whereas we note
that our spectral weights do not exactly achieve that. They merely maximize 
a uniform lower bound, the lower bound given in \Cref{Thm:prob_LB_weighted}.
The tightness of this lower bound might not be equal at each node,
since it only captures the scenario when the node is at the interior of a cluster. 
If a node is not in the interior and thus adjacent to multiple clusters, then
a lower-degree node is likely to be adjacent to fewer clusters and thus still
has higher exposure probability. Therefore, in reality, one might use a
weight where high-degree nodes are even more aggressively favored than
under spectral weighting. In \Cref{sec:simulation}, besides uniform weight  
and spectral weight, we also consider weighting each node by their degree directly.
Simulation results show that this aggressive degree weight strategy usually  
yields lower variance than both uniform weights and spectral weights.

\xhdr{Minimizing a variance proxy}
The above heuristic is intended to reduce the estimator variance, 
but a more direct approach would be to find the optimal weights that 
minimize the actual estimator variance.

That said, optimizing the variance, as formulated in
\Cref{Eq:var-mean_outcome,Eq:var-GATE,Eq:covar-mean_outcome},
is challenging because (i) it consists of cross-terms associated with the joint
exposure probability of node pairs that are hard to analyze,
and (ii) the nodes' response is unknown prior to the experiment, but can play
a significant role in determining the variance.
One compromise is to use a proxy objective function that resembles
the variance formula. We consider the following function
\begin{equation}   \label{Eq:ProxyVariance}
H(\bw) = \sum_{i=1} ^n \frac{1}{\prob{E_i ^\bOne \mid \mathcal P, \bw}},
\end{equation}
which overlooks the cross-terms and assumes a uniform response from all nodes.

Note that this proxy function is also intractable since one cannot efficiently
compute the exposure probability of each node given the weights. However,
one can obtain an upper bound of $H(\bw)$ using the exposure probability
lower bound in \Cref{Thm:prob_LB_weighted}, \ie,
\begin{equation}   \label{Eq:ProxyVarianceUB}
\bar H(\bw) =\frac{1}{p} \sum_{i=1} ^n \frac{\sum_{j \in B_2(i) }w_j}{w_i},
\end{equation}
and attempt to minimize this variance surrogate. We have the following result.
\begin{theorem}   \label{Thm:UniformOptimal}
The minimum of $\bar H(\bw)$ of is achieved with uniform weighting, \ie,
\[
\bar H(\bOne) \leq \bar H(\bw)
\]
for any $\bw \in \real n _+$.
\end{theorem}

The first heuristic increased the exposure probability of high degree
nodes, but came at the cost of decreasing the exposure probabilities of low degree nodes. 
Thus it is not certain whether this heuristic would actually reduces variance. 
It is therefore interesting that under this second heuristic, 
if trusting $\bar H(\bw)$ as a surrogate,
according to \Cref{Thm:UniformOptimal} the optimal weights are the uniform weights, 
corresponding to the unweighted 3-net or 1-hop-max clustering algorithms. 

The construction of the surrogate variance, $\bar H(\bw)$, 
is based on the lower bound exposure probability in \Cref{Thm:prob_LB_weighted},
whose tightness varies between high and low degree nodes. 
Specifically, for a low-degree node $i$, its exposure probability trivially satisfies
$\prob{E_i ^\bOne}\geq p^{\degree i}$, a bound that could potentially be much 
higher than the lower bound $p w_i / (\sum_{j \in B_2(i)} w_j)$.
Consequently, assigning $i$ a low weight would not significantly increase its inverse
exposure probability as penalized in $\bar H(\bw)$.
Therefore, for a real-world network with a wide range of node degrees,
it can certainly still be a good idea to use a weighted clustering algorithm with high weights for high degree nodes.
Our simulations in \Cref{sec:simulation} further demonstrate this intuition.

\subsection{\hajek~estimator bias}
\label{sec:bias_hajek}

The \hajek\ estimator is much less amenable to theoretical analysis than the Horvitz--Thompson (HT) estimator, and so our analysis of the \hajek\ estimator of the GATE is much less extensive. Both GATE estimators depend on the same exposure probabilities, so the general analysis fo the exposure probabilities under randomized $3$-net and 1-hop-max sheds light on the behavior of the \hajek\ estimator as well. That said, the variance much less straight-forward to analyze.

Regardless of these theoretical difficulties, the \hajek\ estimator has many intuitive advantages as a GATE estimator, relative to the HT estimator. We catalog these intuitive advantages briefly, and also contribute a possibly useful observation about the \hajek\ GATE estimator: it is unbiased when the individual treatment effect is constant. 

In our simulations in \Cref{sec:simulation} we offer a full side-by-side evaluation of both the HT and \hajek\ estimators, and find that RGCR also improves \hajek\ estimator performance. That said, RGCR tends to provide order-of-magnitude improvements in the variance of HT estimators, relative GCR. The added benefits of RGCR for the \hajek\ estimator are more modest.

As a first generic advantage of the \hajek~GATE estimator over the HT estimator,
the value of the \hajek\ estimator of a mean outcome, $\tilde \mu(\bz)$, is bounded within the range
of all units' responses, due to the estimator having the form of a convex combination of the responses of all exposed units 
(weighted by the inverse exposure probability). As a result, when the responses are bounded then the \hajek\ estimator variance is immediately bounded.
In contrast, the value of the HT estimator can be far outside this range
of responses, due to its sensitivity to extremely small exposure probabilities, and the HT variance can then be much, much larger as well.

As a second advantage, the variance of the \hajek~estimator is invariant to a shift in unit responses:
if every unit's response is increased or decreased (additively) by a constant, 
then the variance of \hajek~estimator remains unchanged. This, again, is not a property of
the HT estimator for the same estimand.

As a third advantage, for a given outcome $\bZ=\bz$, the \hajek~estimator depends only on the relative value of network exposure probabilities
of all nodes, and is invariant to their absolute value. Specifically, for two sets of 
node-wise exposure probabilities
$\{\prob[1]{E_i ^\bz}\}_{i=1} ^n$ and $\{\prob[2]{E_i ^\bz}\}_{i=1} ^n$ which may 
come from two different experiment designs, if there is a constant $c$ such that
$\prob[1]{E_i ^\bz} = c\cdot \prob[2]{E_i ^\bz}$ for every node $i$, then for a given outcome $\bZ=\bz$ the two sets
of exposure probabilities yield the same \hajek~estimator.
This property might imply an advantage for the RGCR scheme compared with GCR in \hajek~estimation,
as the RGCR scheme yields a more uniform network exposure probability of all nodes:
RGCR tends to increase small probabilities of ``unlucky" nodes and decrease large probabilities of ``lucky'' nodes
compared to a GCR scheme with a fixed clustering (\Cref{Fig:mix_two_partitions}).

Compared with the HT estimator, a potential drawback of the H\'ajek estimator, 
widely known in the literature, is the potential issue of bias, \ie, $\expect{\tilde \mu(\bz)} \neq \mu(\bz)$
and $\expect{\tilde \tau} \neq \tau$.
However, for GATE estimation in the setting where 
every node has the same individual treatment effect 
($\tau_i \triangleq Y_i(\bOne) - Y_i(\bZero)$),  
we observe that it is somewhat surprisingly 
an unbiased estimator for the GATE.
\begin{theorem}   \label{Thm:UnbiasedHajek}
If the treatment effect of every node is constant across all nodes, \ie, $\tau_i \equiv \tau$,
then using either GCR or RGCR scheme with $p = 0.5$, we have $\expect{\tilde \tau} = \tau$.
\end{theorem}

In practice, individual treatment effects $\tau_i$ are reasonably non-constant across individuals, 
making the H\'ajek estimator potentially biased. In our simulations in  \Cref{sec:simulation}, which feature non-constant individual treatment effects, we find that this bias is modest in our settings and the overall mean squared error (MSE) of the \hajek~GATE estimator is broadly superior to that of the HT GATE estimator.

\section{The curse of large clusters}
\label{sec:ring}

In this section, we use a specific network and simple response model to study how the variance of RGCR is affected by network homophily. We conclude that in our model if the number of clusters returned by the clustering algorithm is $O(1)$ in the size of the graph, a non-vanishing variance persists as part of the HT and \hajek~estimators. We consider both independent and complete randomization.

We consider a ring-like network, the cycle graph with $n$ nodes, where each node $i \in \{1, 2, \dots, n\}$ 
is connected to nodes $i-1$ and $i+1$ (except for node $1$ and $n$ being connected). 
We further consider the following simple response model with network drift. For each node $i$,
\[
Y_i = a + b h_i + \tau \frac{\sum_{j \in B_1(i)} z_i}{1 + \degree{i}} ,
\]
where $a$, $b$, and $\tau$ are scalar constants. 

In the second term $h_i$ is the homophily drift also used in our simulations in Section~\ref{sec:simulation} and described in detail there. Informally, $h_i$ is defined according to a natural disagreement minimization problem on the graph. On the cycle graph this problem has the well-known closed-form solution 
$$h_i = \sin \alpha_i, \text{ where } \alpha_i \triangleq \frac{i}{2\pi n}.$$
Here $\alpha_i$ can be thought of as the \emph{angle} of node $i$ along an evenly spaced cycle. The solution comes from basic properties of the cycle graph Laplacian, which is a symmetric circulant matrix. This $h_i$ term then effectively models how nearby nodes generate similar responses while distant nodes generate different reponses.

The third term in the model represents a linear-in-means treatment effect, where we seek to estimate the GATE $\tau$.
As a brief forward reference, we note that this present model is simpler than the response model we consider 
in our simulations in \Cref{sec:simulation}, yet still sufficient to induce the curse of large clusters 
we seek to demonstrate.

With a constant $k > 0$ that divides $n$, an oracle clustering of this network into $k$ clusters is the $k$-partition formed by breaking the ``ring'' into $k$ equally-sized connected arcs. Note that there are $n/k$ such different oracle $k$-partitions.

We study the variance of the RGCR scheme with a random oracle $k$-partition,
in the large-network scenario when $n \rightarrow \infty$.
We have the following results on the HT estimator, with the proof given in \Cref{app:proofs}.
\begin{theorem}   \label{Thm:RingNetworkVariance}
Suppose $p = 1/2$ and $k = o(n)$, then as $n \rightarrow \infty$,
\begin{itemize}
\item with independent randomization, we have
\[
\var{\hat \tau} \rightarrow \frac{(2a+\tau)^2}{ k} + \frac{b^2 k}{\pi^2} (1 - \cos({2\pi}/{k})) = [(2a+\tau)^2 + 2b^2]\cdot \Theta(1/k),
\]
\item with complete randomization, we have
\[
\var{\hat \tau} \rightarrow \frac{b^2k^2}{\pi^2(k-1)} (1 - \cos({2\pi}/{k})) = 2b^2 \cdot \Theta(1/k),
\]
where $\hat \tau$ is the HT estimator of the GATE.
\end{itemize}
\end{theorem}

\Cref{Thm:RingNetworkVariance} yields two important insights. 
First, if the clustering algorithm generates a fixed number of clusters, then the variance of
the HT GATE estimator, both for independent randomization and complete randomization, 
does not converge to 0 as $n \rightarrow \infty$. 
This is, in part or in full, due to
the issue of network homophily, a phenomenon commonly observed in real-world networks whereby closely connected nodes share common behaviors~\cite{mcpherson2001birds}.
In a large graph with few clusters, nodes in each cluster may generate different response than 
other clusters, obfuscate GATE estimation if we assign treatment/control at cluster level:
the difference in the responses of different clusters might be unrelated to the treatment effect, but instead due to endogenous node properties captured in the network topology~\cite{shalizi2011homophily}. 
Therefore, in order for the variance of the estimator to vanish under RGCR, the clustering algorithm needs to generate 
an increasing number of clusters as the network grows large.

Second, the analysis also shows a separate deficit of independent randomization: 
the variance increases quadratically with the average response $a$,
making the estimation sensitive to the scaling and shifting of the average responses.
In contrast, complete randomization does not suffer from this issue, with a variance 
under this response model that is independent of $a$. 
Therefore, we recommend that one should use complete randomization with 
RGCR whenever possible (when positivity is satisfied), 
a change from ordinary GCR where complete randomization
typically does not satisfy positivity for any relevant exposure model.

We also note that the above complete randomization result for the HT estimator applies equally for the \hajek~estimator, 
since these two estimators are asymptotically equivalent in this specific setting. 
Under complete randomization, and due to the fact that each cluster in the oracle $k$-partition contains the same number of nodes,
we always have a constant number of nodes in the treatment and control groups.
Moreover, due to the symmetry of the network, every node has the same exposure probability $\prob{E_1 ^\bz} \to 1/2$ in the limit of $n \to \infty$.
Therefore, the denominator of the \hajek~estimator concentrates at a constant $n$, making it equivalent to the HT estimator.
In summary, we also have non-vanishing variance in the \hajek~estimator if the number of clusters $k$ is bounded as $n \to \infty$.

\section{Simulation experiments}
\label{sec:simulation}

In this section we evaluate the performance of the randomized graph cluster randomization 
(RGCR) scheme in diverse simulations. 
After introducing the simulation setup in \Cref{sec:network,sec:response_model}, 
we examine the behavior of the HT estimator in \Cref{sec:sim_HT_reduction,sec:sim_HT_RGCR},
and the \hajek~estimator in \Cref{sec:simulation_hajek}.
For each estimator, we first demonstrate significant variance reduction 
(as well as bias reduction for the \hajek~estimator) under the RGCR scheme compared with GCR,
and then compare the bias, variance, and mean squared error (MSE) under RGCR
employing various random clustering algorithms.

As randomized clustering algorithms we consider both randomized 3-net and 1-hop-max,
both applied in unweighted, spectral-weighted, and degree-weighted forms.
Note that for RGCR designs we consider both independent and complete randomization while for 
GCR we only consider independent randomization (complete randomization is unattractive 
under GCR due to the potential violation of our positivity assumption).
We find that the spectral- and degree-weighted variants of $3$-net and $1$-hop-max clusterings 
further reduce the variance of the HT estimator and the bias and variance of the \hajek~estimator
(compared with the unweighted clustering algorithms).
In comparing complete randomization and independent randomization,
we find that complete randomization leads to lower variance 
in the HT estimator and the two approaches have comparable bias and variance for
the \hajek~estimator.

These estimators require exposure probabilities, 
which are estimated with Monte Carlo methods introduced in
\Cref{sec:expo_prob_estimate,sec:properties_weighted}.
In \Cref{sec:est_expo_prob} we demonstrate the high accuracy in estimation,
and visualize how the exposure probabilities vary under different random clustering algorithms.
Specifically, we observe that applying the spectral- or degree-weighting scheme
can increase the smallest exposure probabilities compared with the unweighted versions,
offering an explanation of why the variance of the HT estimator as well as the bias and variance 
of the \hajek\ estimator is reduced.

Besides examining the bias, variance, and mean square error (MSE) in each network,
we conclude this section by demonstrating how these quantities decay with the size $n$
of the network. In \Cref{sec:sim_diff_n}, we show that the bias and variance of
the \hajek~estimator decays with a much higher rate under RGCR compared with GCR,
which results in even more significant bias and variance reduction on large networks.
These results highlight the broad favorability of the RGCR scheme in practice.


\subsection{Networks}\label{sec:network}

We consider two interference networks across the experiments in this section.
The first network is drawn from a variation on the small-world network model proposed 
by Kleinberg~\cite{kleinberg2000small}, itself a modification of 
small-world model proposed by Watts and Strogatz~\cite{watts1998collective}.
Besides the two properties of the Watts--Strogatz model
 of high clustering and short average pairwise distance,
Kleinberg's small-world model is known for its navigability: individuals can find 
short chains from purely local information without centralized search~\cite{milgram1967small}.

The navigable small-world network is constructed from a periodic 2-dimensional lattice:
for each node, add a pre-specified number of long edges, where the other end of each edge
is randomly chosen on the network with probability proportional to the square of the inverse lattice
distance, \ie,
\[
\prob{\text{node $v$ is the end of random long edge from $u$}} \propto \dist[lattice](u, v)^{-2}.
\]
The network we use is generated from a $96 \times 96$ lattice, where the number of long edges
at each node is drawn from a power-law distribution~\cite{clauset2009power} with exponent $\alpha = 2.3$.
The resulting degree distribution is then heavy-tailed. We draw exactly one network from the model and fix 
it throughout the majority of our simulations. At a later point in the simulation discussion we consider
the effects of varying the network size within the framework of this model; 
we then sample a single graph for each lattice dimension.

Our second network for simulations is a snapshot the Facebook friendship network 
among Stanford students in 2005, included in the Facebook-100 dataset~\cite{traud2012social}. 
Some basic properties of these two main networks are given in \Cref{tab:network_stat}
with more detailed
growth statistics given in \Cref{app:growth}.

\begin{table}[t]
\centering
\begin{tabular}{c  @{\hskip 5ex} c @{\hskip 5ex}  c @{\hskip 5ex}  c @{\hskip 5ex}  c @{\hskip 5ex}  c
}
\toprule
network   & $n$ & $m$ & $\bar d$ & $\dmax$ & $\kappa$
\\ \midrule
\smallworld & ~9,216 & ~55,214 & 11.98 & 42 & 21.8
\\ \midrule
\stanford & 11,586 & 568,309 & 98.10 & 1172 & 586.5
\\ \bottomrule
\end{tabular}
\caption{Basic properties of the two interference networks studied in our simulations. For more detailed
growth statistics on these two networks, see \Cref{app:growth}.}
\label{tab:network_stat}
\end{table}

\subsection{Response model}\label{sec:response_model}

Our response model is intentionally more complicated than response models studied in pervious
simulations of network interference; the added complications are intended to 
inject realism into the simulations. We propose that this model is ``as simple as possible 
but not simpler'', where removing any one of these components can mislead one to 
conclude that overly simplistic designs or analyses would work well in practice.
We use the following response model throughout this simulation section:
\begin{eqnarray}
\label{Eq:Response0}
Y_i(\bZero)  &=&  (a + b \cdot h_i + \sigma \cdot \epsilon_i) \cdot \frac{\degree{i}}{\bar d}, \\
\label{Eq:ResponseZ}
Y_i(\bz)  &=&  Y_i(\bZero) \cdot \left(1 + \delta z_i + \gamma  \frac{\sum_{j \in \nbr{i}} z_j}{\degree{i}} \right).
\end{eqnarray}
The model has the following components.

\xhdr{Parameters}
The parameters $a$, $b$, and $\sigma$ are constants, 
where $a$ controls the shifting in average node's response,
$b$ controls the magnitude of homophily that results in a network drift effect (discussed below),
and $\sigma$ controls the noise level where $\epsilon_i \sim_{\iid} N(0, 1)$ 
is independent of any other node attributes.
In all our experiments we use $a = 1$, $b = 0.5$, $\sigma = 0.1$, 
and let $\delta =\gamma=0.5$ for the treatment effects.

\xhdr{Interference} 
Focusing first on the treatment effect, $\delta z_i$ represents the direct effect and 
$\gamma \frac{\sum_{j \in \nbr{i}} z_j}{\degree{i}}$ represents spillovers.
With this response model, the full-neighborhood exposure model is properly
specified, and we have 
\begin{equation}   \label{Eq:mode_tau_i}
\tau_i = Y_i(\bOne) - Y_i(\bZero) = (\delta + \gamma) \cdot Y_i(\bZero),
\end{equation}
and the GATE $\tau$ becomes
\begin{equation}   \label{Eq:mode_tau}
\textstyle
\tau = \frac{1}{n}\sum_i \tau_i = (\delta + \gamma) \cdot \mu(\bZero).
\end{equation}

\xhdr{Degree-correlated responses}
The role of the degree $\degree{i}$ and average degree $\bar d$
induce a strong correlation between node degree and control 
response, a realistic phenomenon~\cite{basse2018model} that 
also injects heavy-tailed-ness into the response distribution 
whenever the degree distribution is heavy-tailed.

\xhdr{Multiplicative treatment effect} 
Instead of the more common additive treatment effect
here the treatment effect is multiplicative at the node level. This multiplicative model,
which has also been studied elsewhere~\cite{aronow2017estimating}, 
caries forward the correlation between degree and control response to cause a heterogeneous
``individual'' global treatment effect. 
Note that, according to \Cref{Thm:UnbiasedHajek}, a heterogeneous treatment effect
is required to reveal the bias in \hajek~estimation.
A multiplicative treatment effect can also be deemed natural because 
the different exposure levels incur the same \emph{relative} change in a units' response.

\xhdr{Homophily}
Our use of a network homophily term $h_i$ is new to the literature on
causal inference under interference, and we believe it provides an important missing piece
for evaluating experimental designs under a more realistic response model.
This term represents the network drift phenomenon
where closely connected nodes usually have similar characteristics and consequently
similar response, while distant nodes can be dissimilar. 
For example, in the United States, many behaviors are correlated with geography,
while network structure is also very obviously correlated with geography~\cite{ugander2013balanced}.
Failure to consider this network drift in the response incurs additional variance 
in the estimation under graph cluster randomization: 
assigning all the East Coast users into the treatment while all West Coast users
to control would make GATE estimation sensitive to any variations
related to this network-level drift effect.
We note that such network drift is different from network correlation,
which is modeled by a Gaussian Markov Random Field (GMRF)~\cite{basse2018model}.
GMRF models only introduce local correlation in the nodes response,
while they do not impose a global drift in response at the network level.

We construct our $h_i$ feature as solving the following disagreement minimization problem:
\begin{equation}   \label{Eq:Opt}
\begin{array}{cl}
\min_{\bh \in \real{n}}  &  \sum_{(i,j)\in E} (h_i - h_j)^2 \\
\\
\text{subject to}  & \sum_{i = 1}^n \degree{i} h_i = 0, \\
\\
& \max\{|h_i|: 1 \leq i \leq n\} = 1.
\end{array}
\end{equation}
Without the constraints it is clear that any constant $h_i$ would minimize the objective,
but we are constructing $h_i$ to be the scalar function that minimizes
the disagreement across all edges, subject to their being non-zero disagreement.

This minimization problem is a classic problem in spectral graph theory,
where the solution is the eigenvector associated with the second smallest eigenvalue 
of the normalized graph Laplacian matrix $D^{-1} L$~\cite{von2007tutorial}, \ie,
\[
D^{-1} L \bh = \lambda_2 \bh,
\]
where $D$ is the diagonal degree matrix, and $L = D - A$ is the unnormalized 
graph Laplacian.
The problem notably arises in spectral clustering~\cite{von2007tutorial}
when the goal is to partition the graph into non-overlapping clusters
so that each cluster is internally well-connected while loosely linked between each 
other\footnote{According to Cheeger inequality~\cite{cheeger1969lower}
the eigenvalue $\lambda_2$ is related with the optimal conductance
of the graph bisection, and in the spectral clustering algorithm, 
a favorable bisection is obtained by thresholding~\cite{spielman1996spectral} 
or executing $k$-means (with $k=2$)
on the eigenvector $\bh$~\cite{shi2000normalized,von2007tutorial}.}.
The eigenvector $\bh$ can be obtained by various numerical linear algebra
algorithms, e.g., power iteration~\cite{golub2012matrix},
with computational complexity $O(m \log (1/\epsilon))$, where $m$ is the number of edges 
in the network and $\epsilon$ is the target accuracy in the output eigenvector.

A consequence of using this drift function $h_i$ is that, as $n \rightarrow \infty$,
we have $\mu(\bZero) \rightarrow a$.
Such convergence is due to $\sum_{i} h_i \degree{i} = 0$ as in \Cref{Eq:Opt},
and $\frac{1}{n} \sum_i \epsilon_i \degree{i} \rightarrow 0$ since $\epsilon_i \sim_{\iid} N(0,1)$
which is also independent of $\degree{i}$. As a result, $\frac{1}{n} \sum_i a (\degree{i} / \bar d) = a$.
Consequently, we have
\[
\tau = \mu(\bZero) \rightarrow 1,
\]
as $n \rightarrow \infty$.

In \Cref{Fig:ResponseHeatmap} 
we visualize the node homophily feature $h_i$
in our heavy-tailed small-world network
as well as the resulting response under global control $Y_i(\bZero)$.
Notice that adjacent nodes on the lattice have closer value in $h_i$,
while distant nodes tend to have significantly different $h_i$.
The variation of $h_i$'s along the lattice is not smooth due to the 
existence of a heavy-tailed number of long-range edges in the small-world network.

\begin{figure}[t] 
   \centering
\begin{minipage}{0.4\linewidth}\centering
   \includegraphics[height=2in]{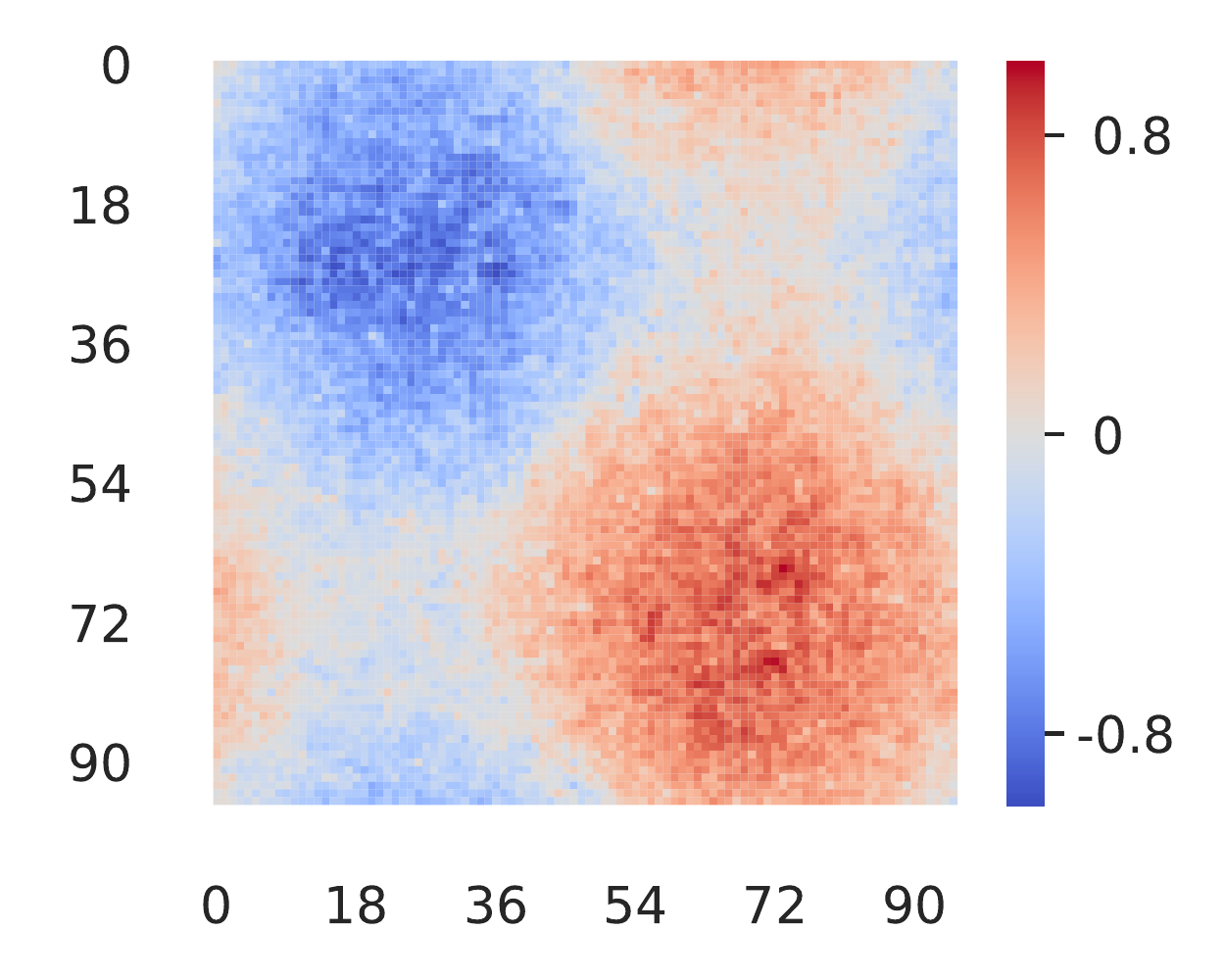} 
   \vspace{-1.5ex}
   \\{\footnotesize $h_i$   }
\end{minipage}
\begin{minipage}{0.4\linewidth}\centering
   \includegraphics[height=2in]{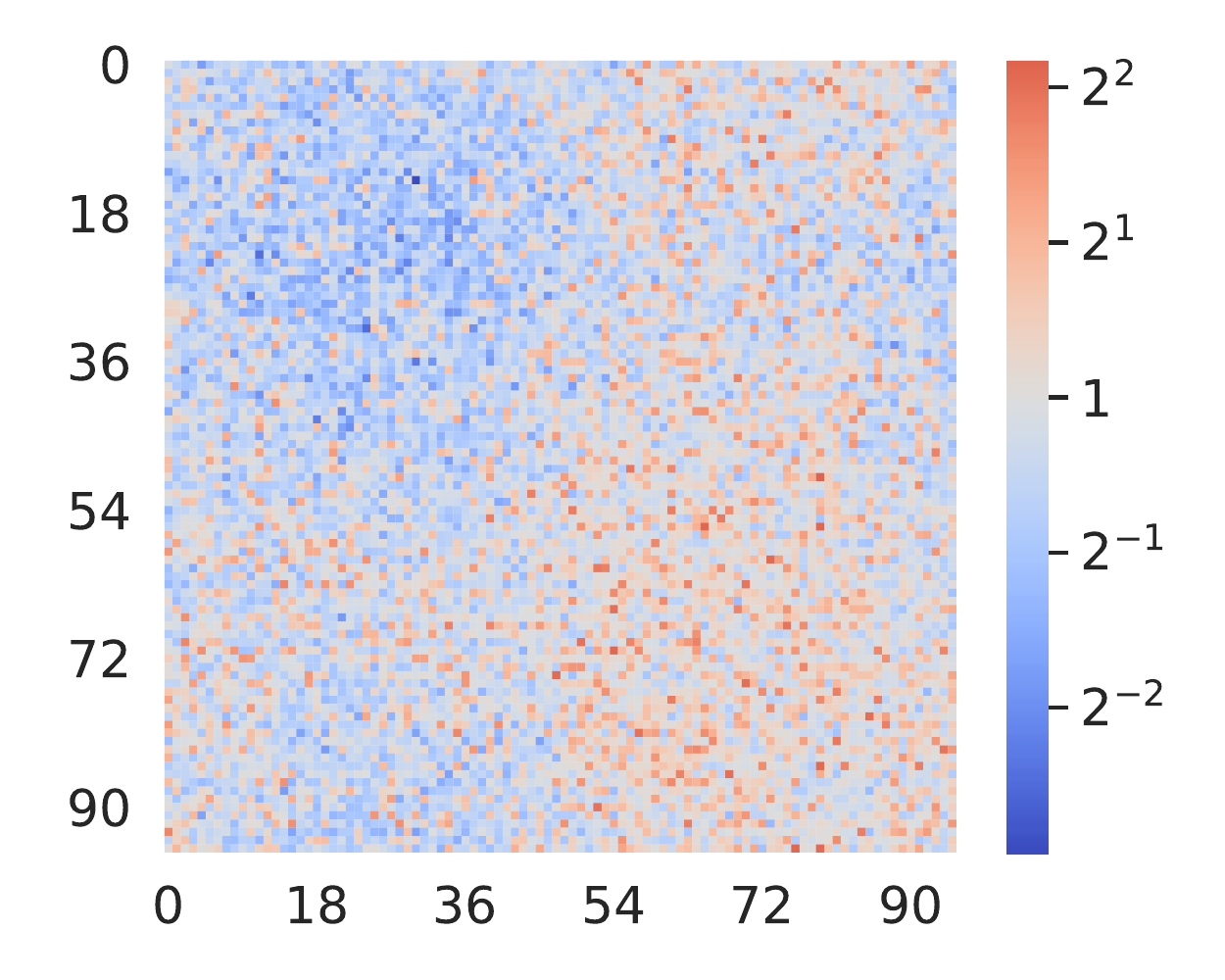} 
   \vspace{-1.5ex}
   \\{\footnotesize $Y_i(\bZero)$   }
\end{minipage}
\caption[Illustration of node locality feature and response]
{Heatmap of the node homophily feature $h_i$ (left) and
the corresponding response under global control $Y_i(\bZero)$ (right)
of a heavy-tailed small-world network based on the periodic $96 \times 96$ lattice.
}
\label{Fig:ResponseHeatmap}
\end{figure}


\subsection{Variance reduction of $\hat \tau$}\label{sec:sim_HT_reduction}

In this section, we demonstrate the significant variance reduction of the HT estimator 
under randomized graph cluster randomization (RGCR) compared with the standard GCR scheme. 

To make the benefits of randomization concrete, 
for each RGCR design we study the variance of HT GATE estimators 
when mixing $K$ clusterings, for varying values of $K$.
Specifically, we first generate $K$ random clusterings from the random clustering algorithm $\mathcal P$ and fix them. 
We then let the random clustering distribution $\mathcal P_K$ be the uniform distribution 
among the $K$ generated clusterings, $\{\bc^{(k)}\}_{k=1}^K$:
in the design phase we use a single clustering
out of the $K$ clusterings drawn uniformly at random;
and in the analysis phase, the exposure probability of each node under $\mathcal P_K$ is then
\begin{equation}   \label{Eq:FiniteMixProb}
\mathbb{P}[E_i ^\bz \mid \mathcal P_K] = \frac{1}{K} \sum_{k = 1}^K \prob{E_i ^\bz \mid \bc^{(k)}}.
\end{equation}

The benefit of analyzing mixtures of $K$ fixed clusterings is two-fold.
First, under this $K$-cluster design we can compute the exposure probabilities 
(\Cref{Eq:FiniteMixProb}) exactly (compared to estimating 
them using a Monte Carlo procedure when considering the full clustering distribution
$\mathcal P$).
Second, it unifies the GCR and RGCR schemes:
With $K = 1$ we are considering the standard GCR scheme 
with a single fixed clustering, while as $K \rightarrow \infty$,
it approaches the RGCR scheme based on the random clustering algorithm $\mathcal P$. 

In our simulations we contrast mixtures of $K \in \{1, 10, 10^2, \dots, 10^6\}$ clusterings
for both our heavy-tailed small-world network and the \stanford~friendship network.
The variance of the GATE estimators depend on the $K$ partitions being randomly generated,
and thus we repeat each partition generation process $400$ time to obtain a distribution.
We assign each cluster to the treatment group with probability
$p = 0.5$, the symmetric assignment scenario where we have 
$\prob{E_i ^\bOne \mid \mathcal P} = \prob{E_i ^\bZero \mid \mathcal P}$.

\begin{figure}[t] 
   \centering
 \begin{minipage}{0.4\linewidth}\centering
   \includegraphics[height=2.2in]{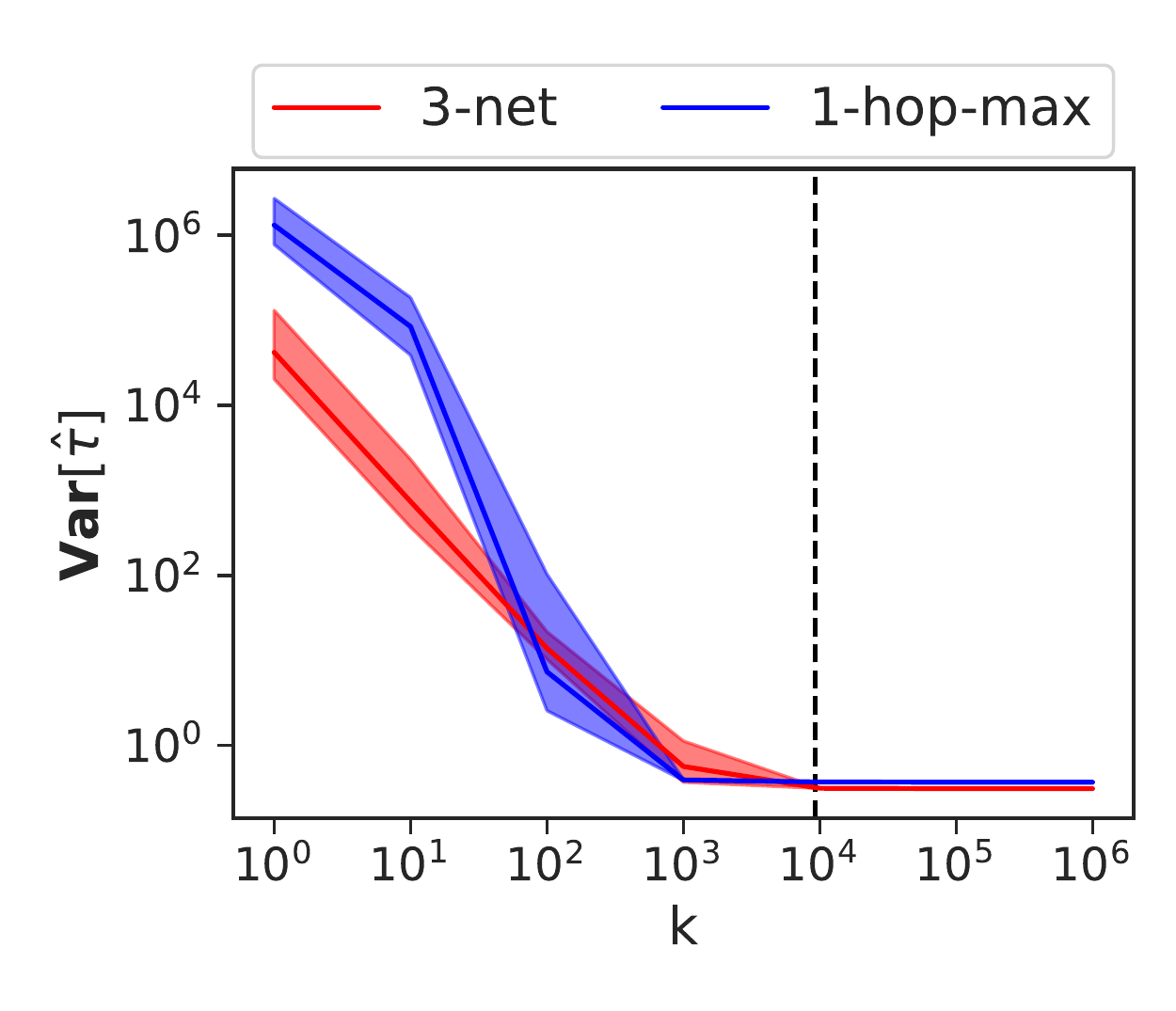} 
   \vspace{-1em}
   \\{\footnotesize \hspace{3em} \smallworld  }
 \end{minipage}
 \begin{minipage}{0.4\linewidth}\centering
   \includegraphics[height=2.2in]{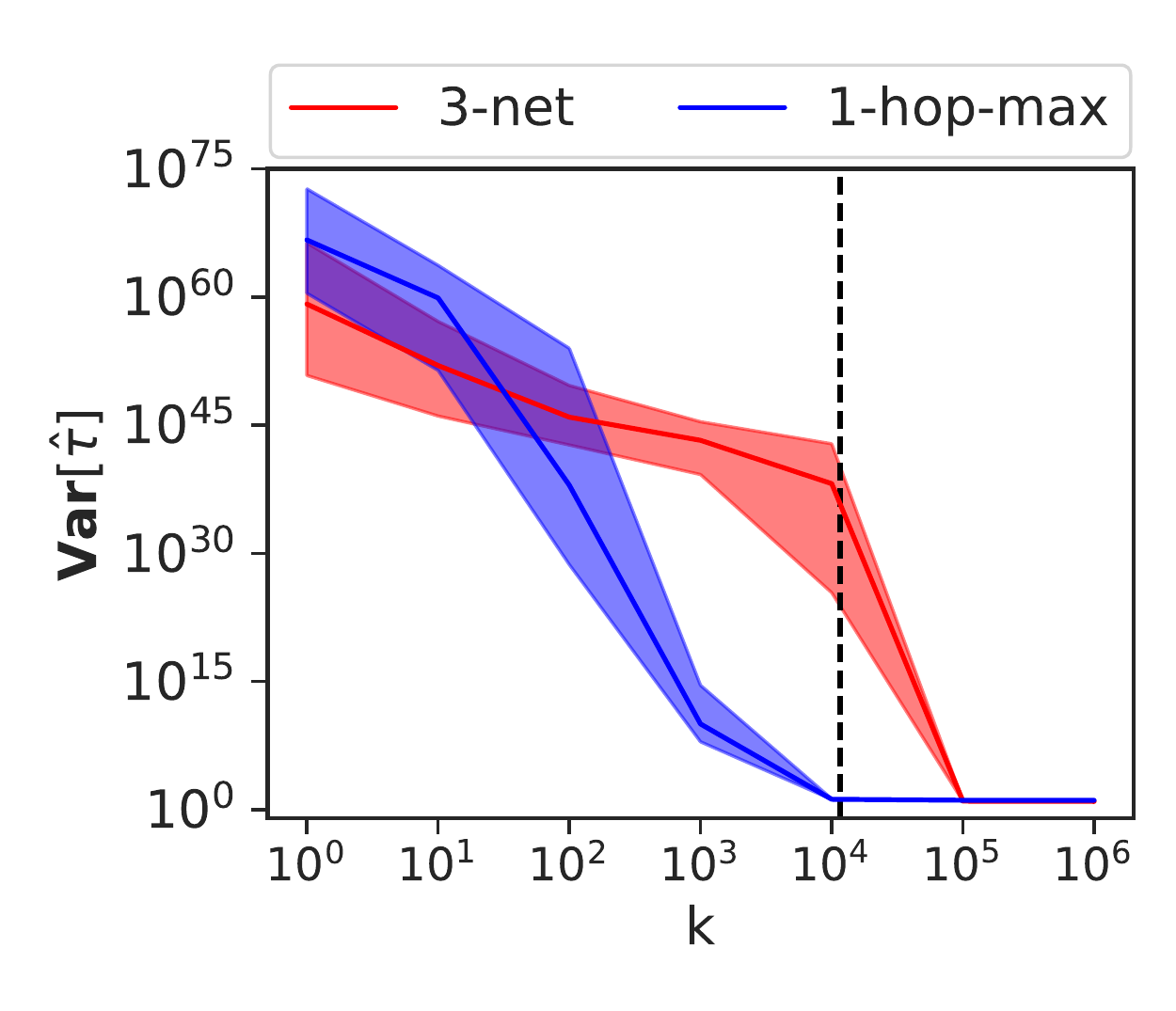} 
   \vspace{-1em}
   \\{\footnotesize  \hspace{3em}  \stanford   }
 \end{minipage}
\caption[Comparing $\var{\hat \tau}$ between GCR and RGCR schemes]
{The distribution of $\var{\hat \tau}$ when mixing $K$ random clusterings
from the unweighted randomized 3-net and 1-hop-max algorithms
in the heavy-tailed \smallworld\ network (left) and the \stanford\ network (right).
The number of nodes in each network is marked by a dashed line.
We plot the median (solid line) as well as the 2.5\% and 97.5\% quantiles (shaded area)
of the variance distribution from our simulations.
For both algorithms and both networks we observe enormous variance reduction from cluster mixing.
Similar trends are also seen in the weighted clustering methods (not shown).
}
\label{Fig:VarDist}
\end{figure}

\Cref{Fig:VarDist} shows the variance of the HT GATE estimator under each 
unweighted random clustering strategy with independent cluster-level assignment.
Within each scheme, we observe enormous variance reduction 
from randomized clustering in both the synthetic small-world network and 
the Facebook friendship network.
With $K = 1$, \ie, the standard GCR scheme with a fixed clustering,
the variance is uselessly high due to the existence of many exponentially
small exposure probabilities. 
As $K$ increases, the variance of the HT GATE estimator 
decreases monotonically and significantly.
Towards the limit of $K \rightarrow \infty$ which is an approximation to the RGCR scheme,
the variance is reduced to a realistic level of around $10^0$. The exact
value of the variance depends strongly on both the response model and the size of the networks 
we use in our simulations. 
Consistent with the theory developed in \Cref{sec:properties},
such variance reduction agrees with what one would expect from the 
exponentially small exposure probabilities being ``washed out", 
gradually growing to probabilities that are only polynomially small.

\subsection{Exposure probabilities}
\label{sec:est_expo_prob}

In the previous subsection, we approximate the RGCR scheme under 
the random cluster distribution $\mathcal P$ by using
a uniform clustering from a collection of $K$ clusterings from $\mathcal P$. 
In the next subsections, we examine the performance of the generic RGCR scheme.
Recall that due to the infeasibility of exactly computing the exposure 
probabilities (\Cref{Thm:Hard-epsnet}), 
for a generic $\mathcal P$ we must rely on estimated probabilities 
obtained via Monte Carlo
(\Cref{sec:expo_prob_estimate,sec:properties_weighted}).

We first validate the accuracy in this Monte Carlo procedure
 and then also examine the 
estimated probabilities, comparing them with the theory developed 
in \Cref{sec:properties}. 
Again all simulations are conducted under the scenario where we assign 
each cluster into the treatment group with probability $p = 0.5$, thus
$\prob{E_i ^\bOne \mid \mathcal P} = \prob{E_i ^\bZero \mid \mathcal P}$
and we only need to estimate $\prob{E_i ^\bOne \mid \mathcal P}$.

\xhdr{Accuracy in exposure probabilities estimation}
We demonstrate the accuracy of our Monte Carlo procedure for 
a $32\times32$
instance of our heavy-tailed small-world network model
 using the unweighted 3-net clustering as an example.  
To measure the relative error of the probabilities, as well as how they decay with the number 
of stratified samples in Monte Carlo estimation, we conduct the estimation procedure
with $Kn$ stratified samples, where $K$ ranges from $\{1, 2^1, \dots, 2^7\}$.
For each $K$, we repeat the estimation 10 times and estimate the relative
standard deviation
\[
\rstd\left(\hat{\mathbb P}_K\left[E_i ^\bOne \mid  \mathcal P \right] \right)
= \frac{\std{\hat{\mathbb P}_K\left[E_i ^\bOne \mid \mathcal P \right]}}
{{\mathbb P}\left[E_i ^\bOne \mid  \mathcal P \right]}
\]
of each node $i$ with Maximum Likelihood, where the subscript $K$
denotes that $Kn$ samples were used.
We then compute the average relative standard deviation of all nodes for each $K$,
and summarize the results in the left subfigure in \Cref{Fig:ProbEstVar}.
We observe from the figure that the average relative standard deviation decays as
\[
\hat{\rstd}\left(\hat{\mathbb P}_K\left[E_i ^\bOne \mid  \mathcal P \right] \right)
\propto K^{-1/2}
\]
When $K = 128$, the average relative standard deviation is around 1\%.

\begin{figure}[t] 
   \centering
 \begin{minipage}{0.4\linewidth}\centering
    \includegraphics[height=1.6in]{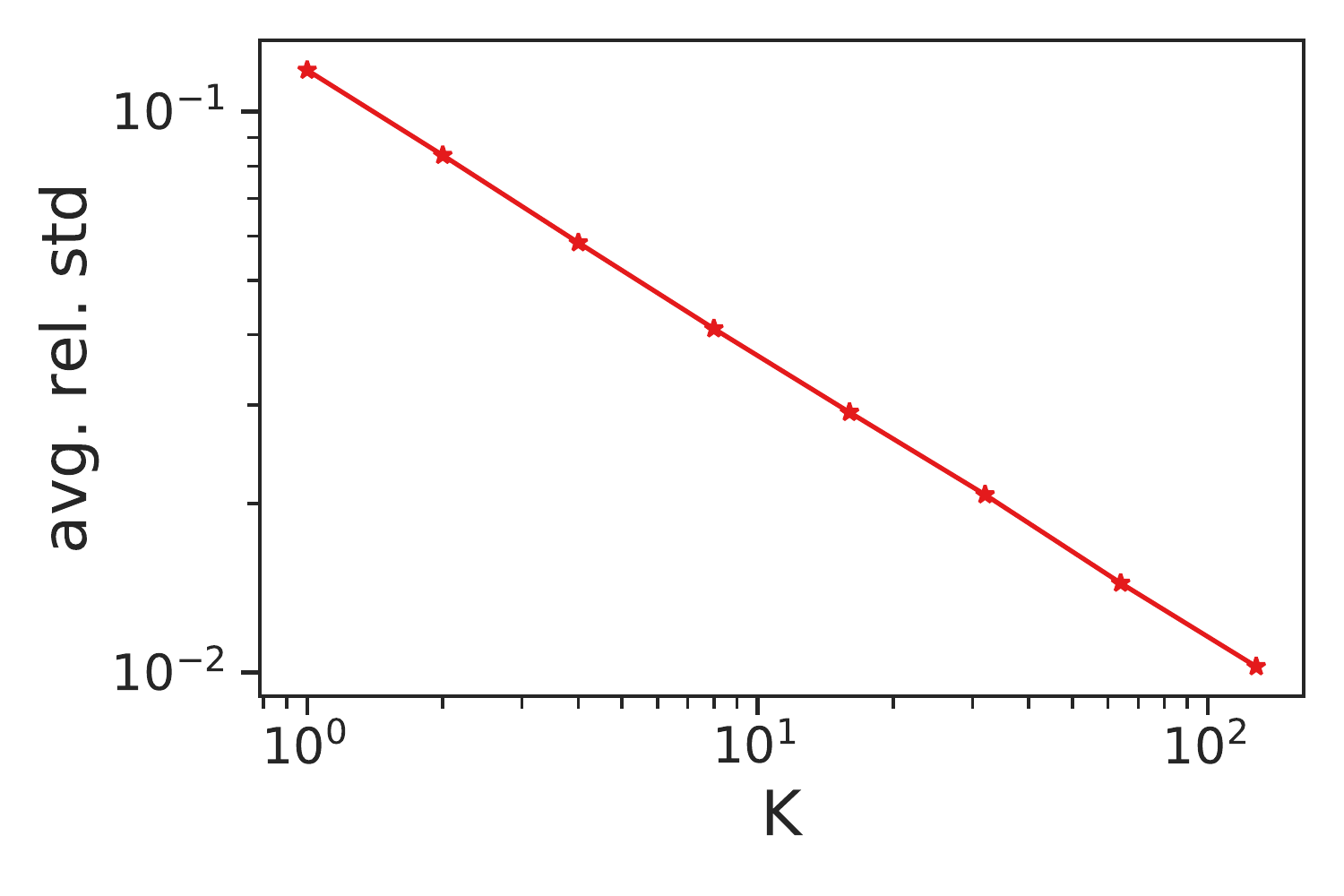} 
 \end{minipage}
 \begin{minipage}{0.4\linewidth}\centering
   \includegraphics[height=1.6in]{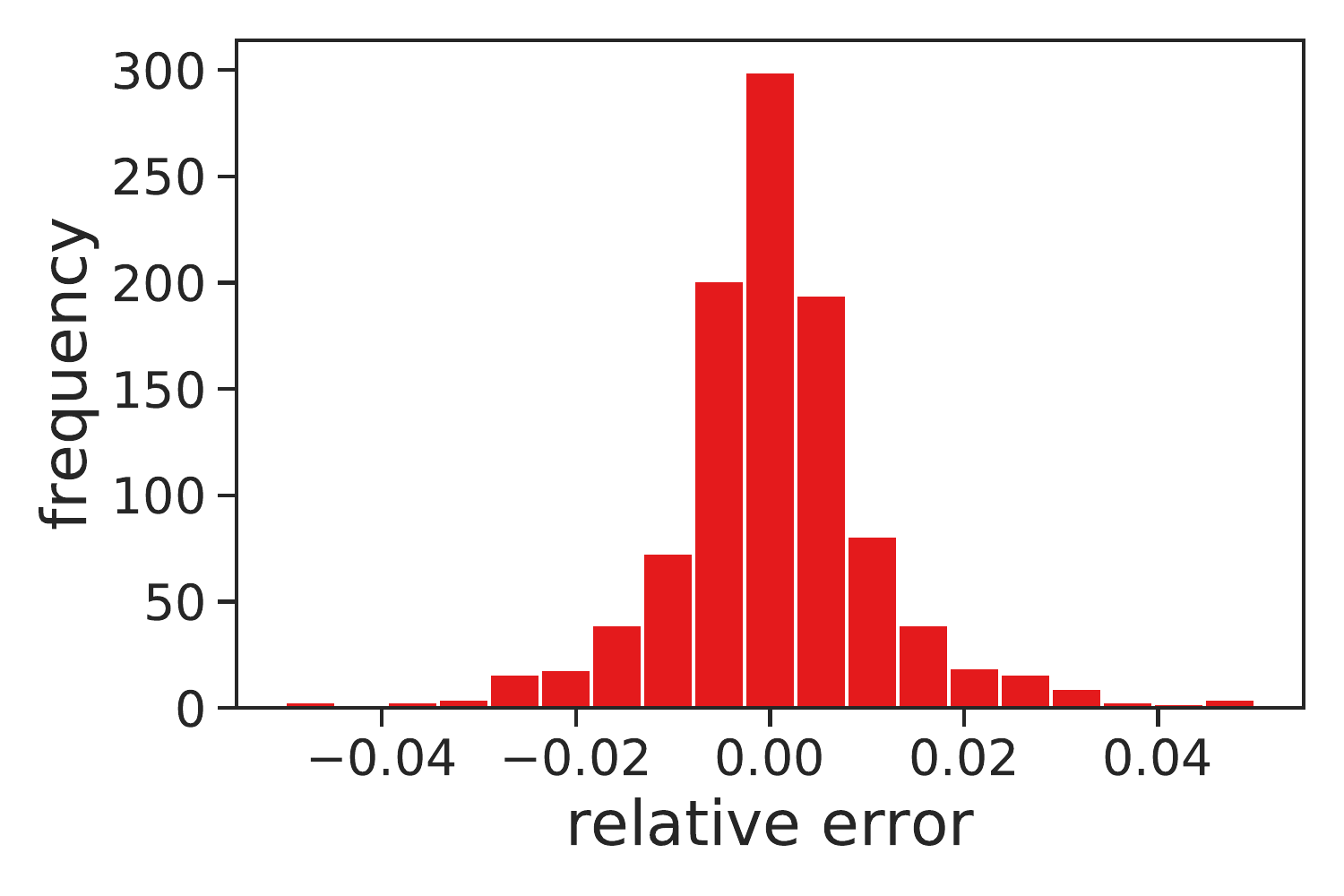} 
 \end{minipage}
\vspace{-1em}
\caption[Error analysis in exposure probability estimation]
{Left: Average relative standard deviation of the exposure probability 
estimator with $Kn$ stratified samples in Monte Carlo estimation.
Right: histogram of the relative error in the estimated exposure probabilities
of all nodes, where $128n$ stratified samples are used
in Monte Carlo estimation. Both plots are constructed with the heavy-tailed
small-world network.
}
\label{Fig:ProbEstVar}
\end{figure}

Besides the relative standard deviation, we also examine the distribution of the relative
error of a set of estimated exposure probabilities from $128n$ stratified samples,
\[
\mathbf{err}_i
=
\frac{\hat{\mathbb P}\left[E_i ^\bOne \mid  \mathcal P \right] - {\mathbb P}\left[E_i ^\bOne \mid  \mathcal P \right]}{{\mathbb P}\left[E_i ^\bOne \mid  \mathcal P \right]},
\]
where we use the average exposure probability across the 10 repetitions
as the ground truth exposure probability of each node.
The histogram of the relative errors at all nodes is given
\Cref{Fig:ProbEstVar} (right),
where we see the relative errors are bounded within $\pm 5\%$ and mostly within $\pm 2\%$.
Analogous results for other networks and clustering algorithms (not shown) 
confirm a broadly satisfying accuracy for the estimated exposure probabilities.
We use these estimated probabilities in place of the exact exposure probabilities
in all our uses of the HT and \hajek\ GATE estimators.

\begin{figure}[p] 
   \centering
   \includegraphics[width=0.8\linewidth]{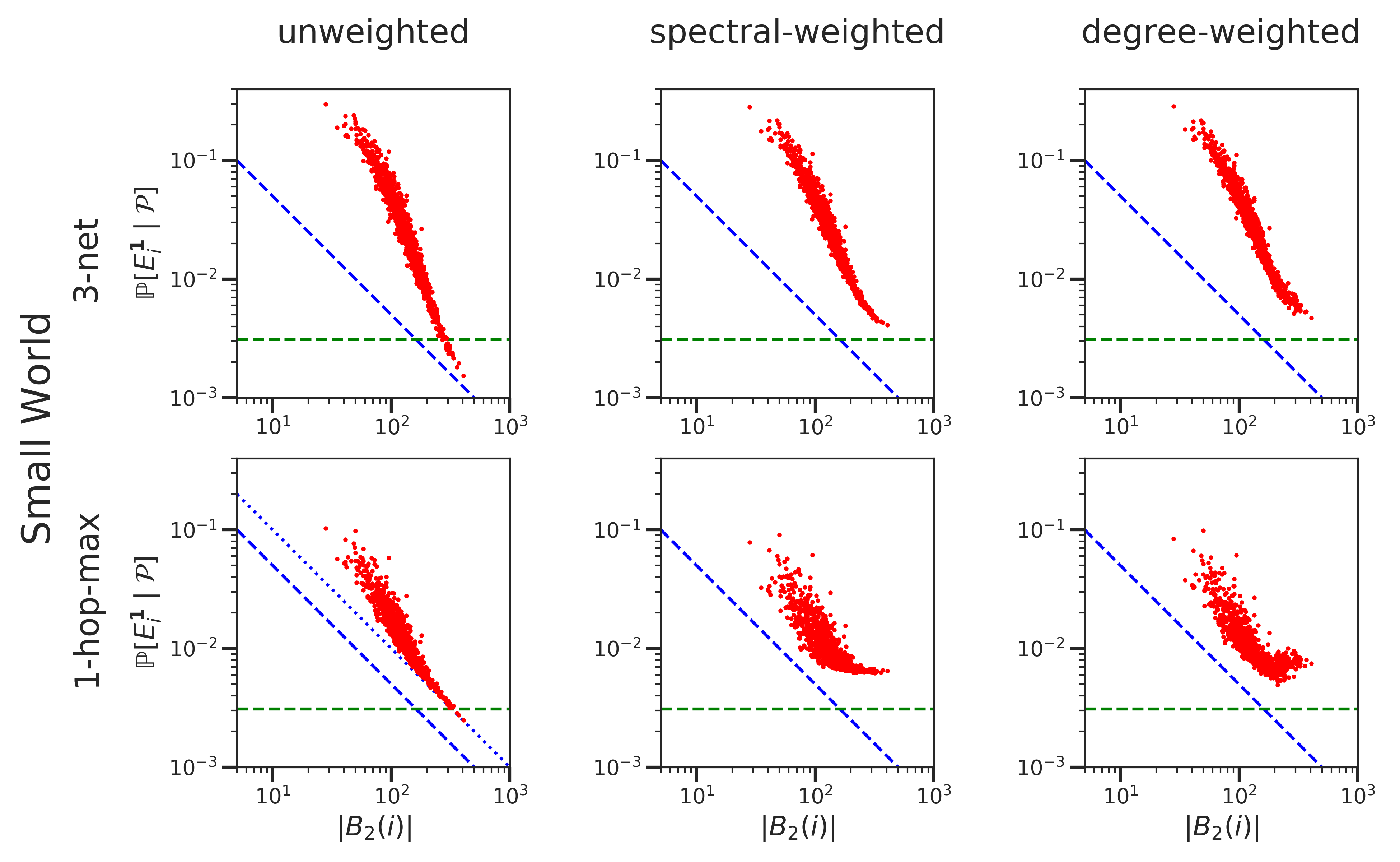} \\
   \includegraphics[width=0.8\linewidth]{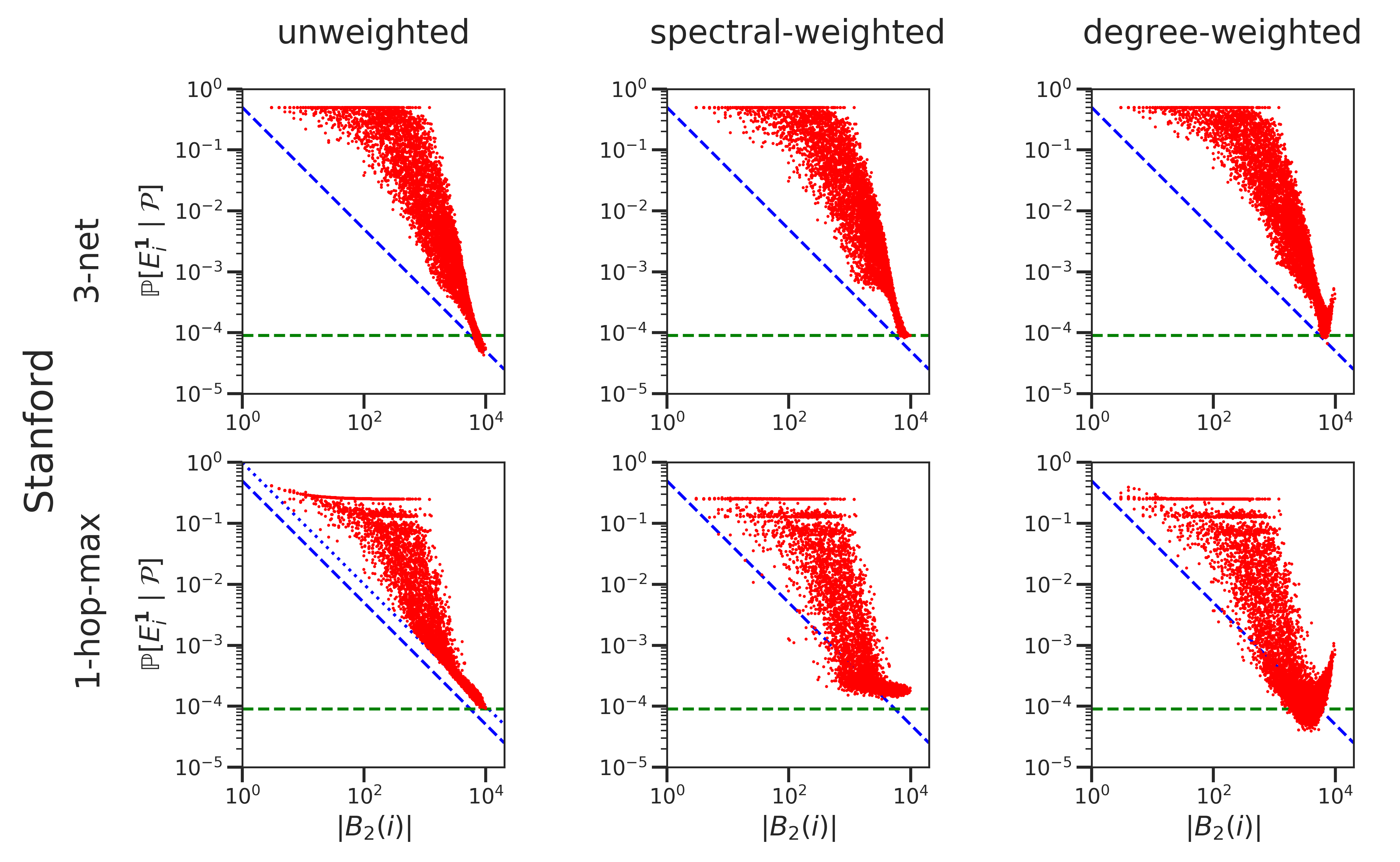}
\caption[Scatterplot of exposure probabilities under RGCR]
{Scatterplot of the exposure probability $\prob{E_i ^\bOne \mid \mathcal P}$ versus 
$\lvert B_2(i)\rvert$ at every node $i$ in the heavy-tailed small-world network
(top two rows) and the \stanford~network (bottom two rows), 
under each random clustering scheme with cluster-level independent randomization.
The blue dashed line represents the exposure probability lower bound for unweighted
3-net and 1-hop-max schemes (\Cref{Thm:prob_LB}). 
The blue doted line represents the slightly improved lower bound for 
unweighted 1-hop-max (\Cref{Thm:prob_LB_improved}).
The green dashed line represents
the uniform lower bound for the spectral-weighted schemes (\Cref{Thm:prob_LB_spectral}).
}
\label{Fig:ProbVS2Ball}
\end{figure}

\xhdr{Visualizing exposure probabilities}
\Cref{Fig:ProbVS2Ball} furnishes a scatterplot of the estimated exposure probabilities
under each randomized clustering strategy with cluster-level independent randomization 
versus the size of the 2-ball at each node. In the first column, where we see 
the unweighted randomized 3-net and 1-hop-max schemes, 
the lower bound provided in \Cref{Thm:prob_LB}
is verified (blue dashed line), 
as well as the slightly stronger lower bound for the unweighted 
1-hop-max scheme from \Cref{Thm:prob_LB_improved}.
These lower bounds fall off with $\lvert B_2(i)\rvert$, and thus nodes with larger 
2-hop neighborhood can have lower exposure probabilities. Moreover, 
we observe that the lower bounds are more tight for nodes with a larger 2-neighborhood.

The second and third columns of \Cref{Fig:ProbVS2Ball} are associated with 
the weighted 3-net and 1-hop-max clustering strategies. 
In the second column, we consider the spectral weighting developed
in \Cref{sec:weighted_opt_choice}, which obey a uniform lower bound
(green dashed line) on the exposure probability independent of $\lvert B_2(i)\rvert$.
While the bound is uniform, the slack is not, and we observe that it is again more tight for nodes 
with larger $\lvert B_2(i)\rvert$.
Besides spectral weighting, in the third column we consider another scheme
where each node is weighed by its degree.
Compared with the spectral weighting, this degree weighting scheme improves
the exposure probabilities at nodes with largest $\lvert B_2(i)\rvert$, while
it may harm other nodes: note the dip below the green dashed line 
in the 1-hop-max scatterplots for the \stanford\ network.

\begin{figure}[tp] 
   \centering
{\small \smallworld} \vspace{1em} \\ 
\begin{minipage}{0.25\linewidth}\centering
   {\small \hspace{2em} 3-net clustering  }\\
   \includegraphics[height=1.7in]{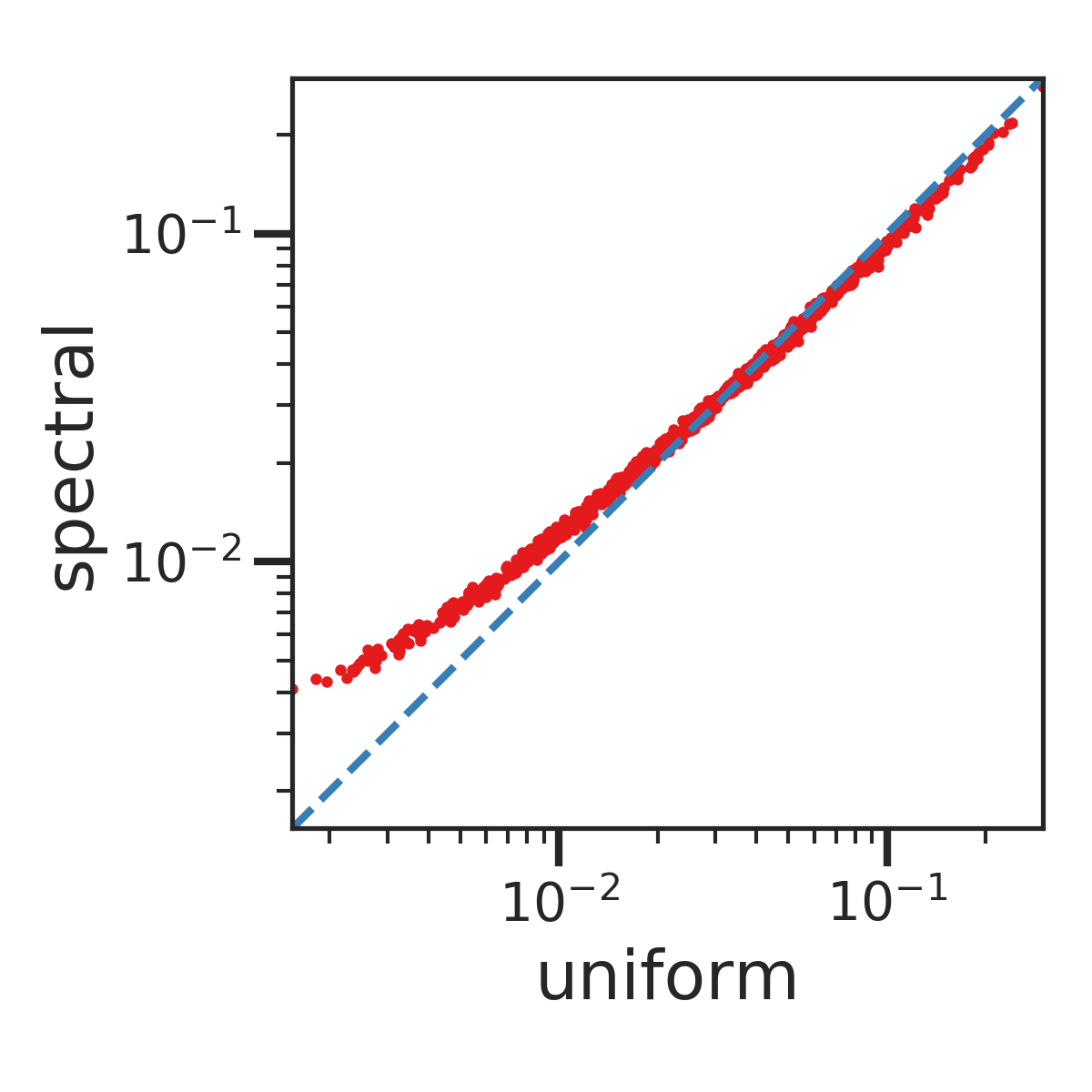} 
\end{minipage}
\begin{minipage}{0.25\linewidth}\centering
   {\small \hspace{2em} 1-hop-max clustering   }\\
   \includegraphics[height=1.7in]{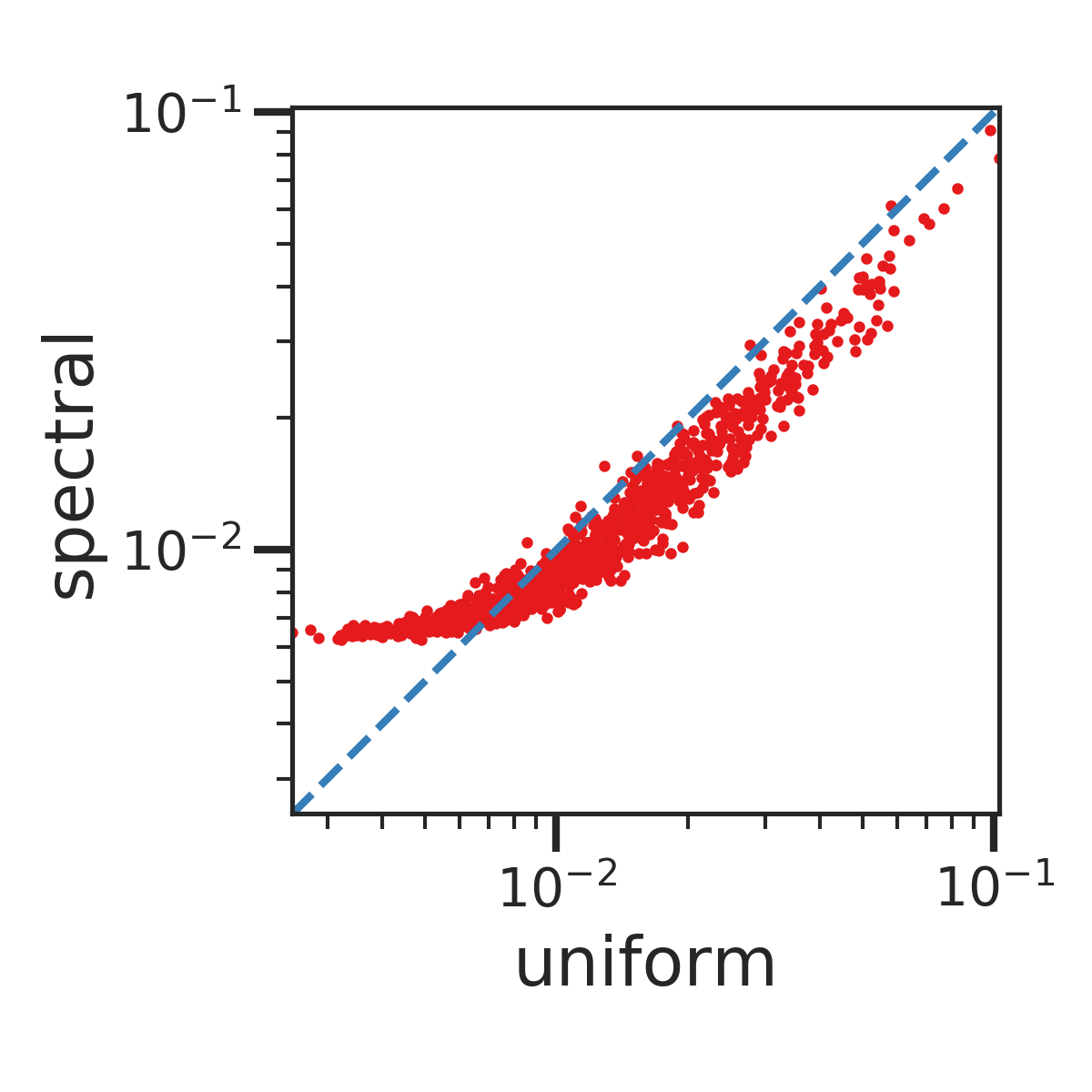} 
\end{minipage}
\begin{minipage}{0.25\linewidth}\centering
   {\small \hspace{2em} unweighted   } \\
   \includegraphics[height=1.7in]{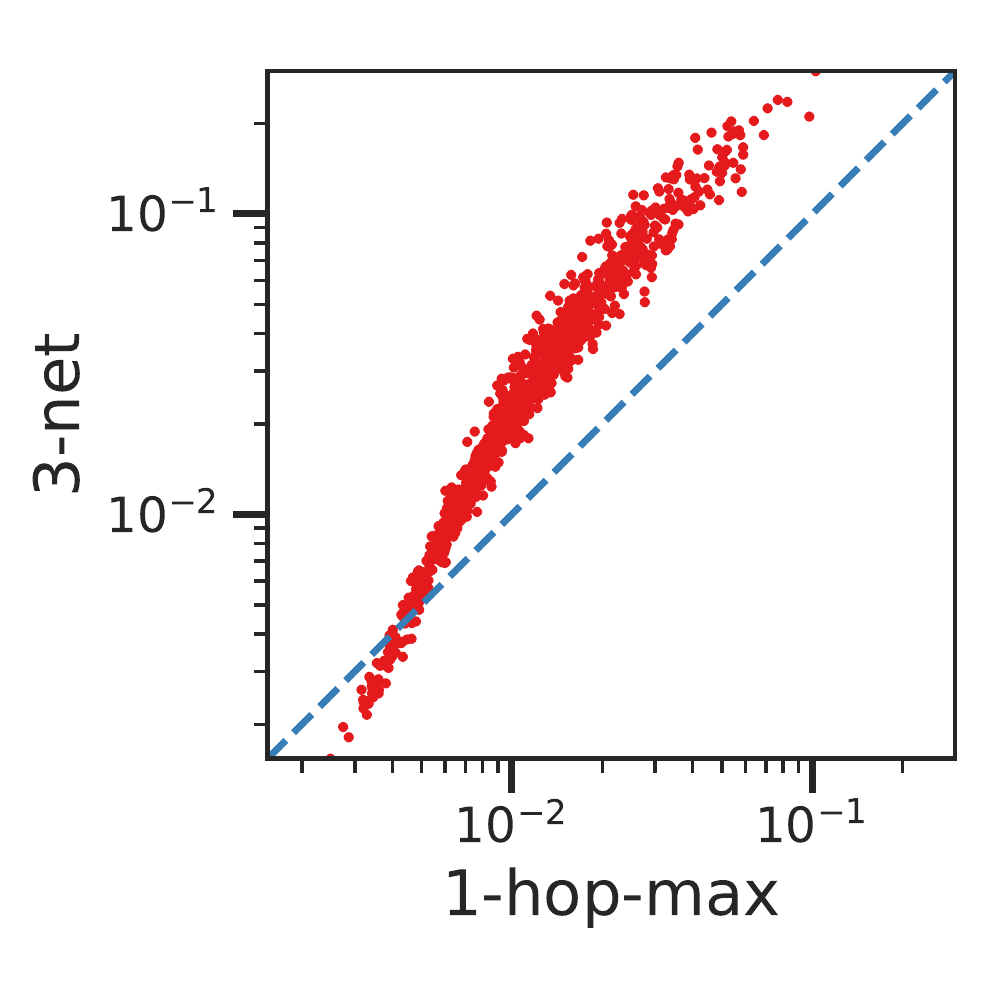} 
\end{minipage}
\\ \vspace{1em}
{\small \stanford} 
\vspace{1em}
\\
\begin{minipage}{0.25\linewidth}\centering
   {\small \hspace{2em} 3-net clustering  }\\
   \includegraphics[height=1.7in]{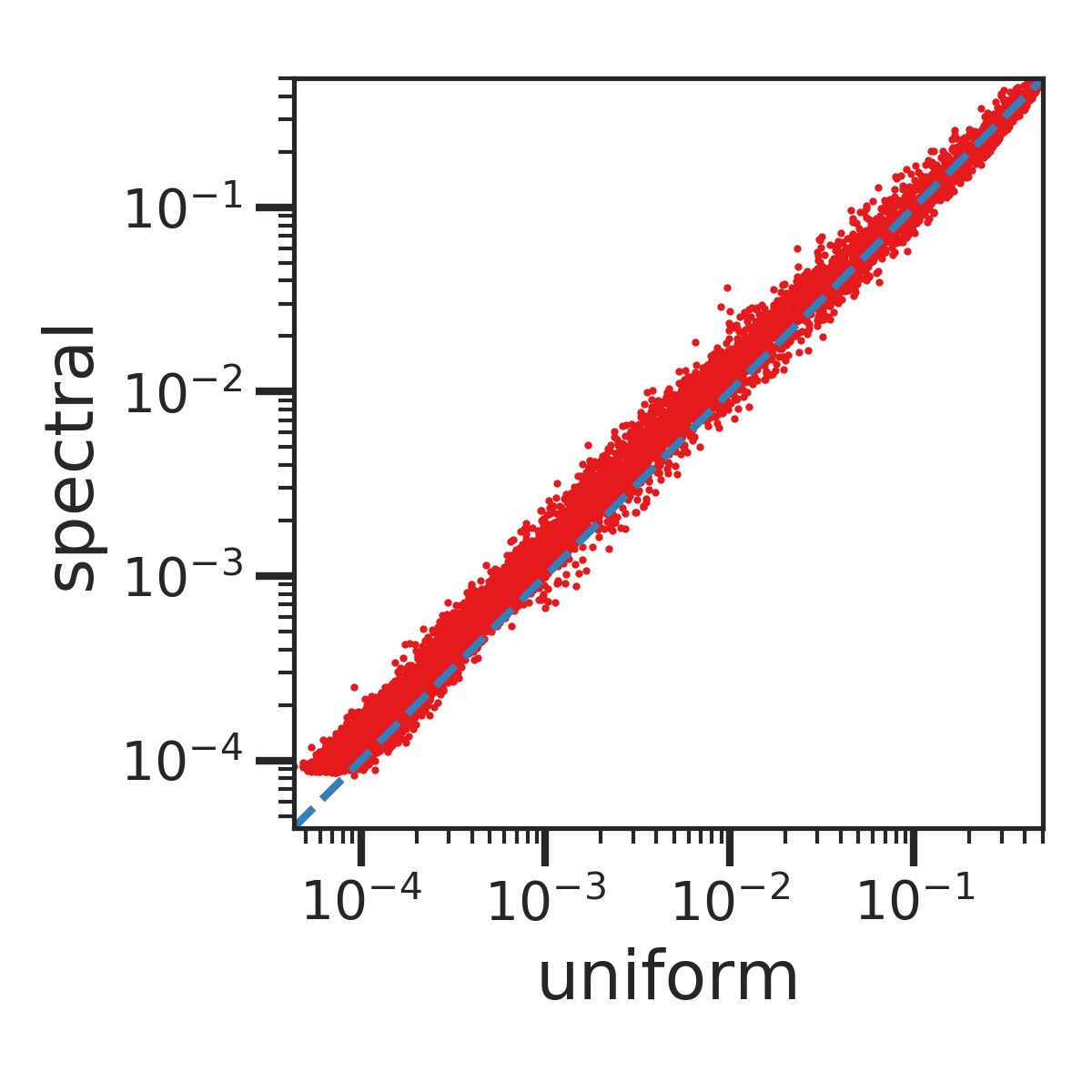} 
\end{minipage}
\begin{minipage}{0.25\linewidth}\centering
   {\small \hspace{2em} 1-hop-max clustering   }\\
   \includegraphics[height=1.7in]{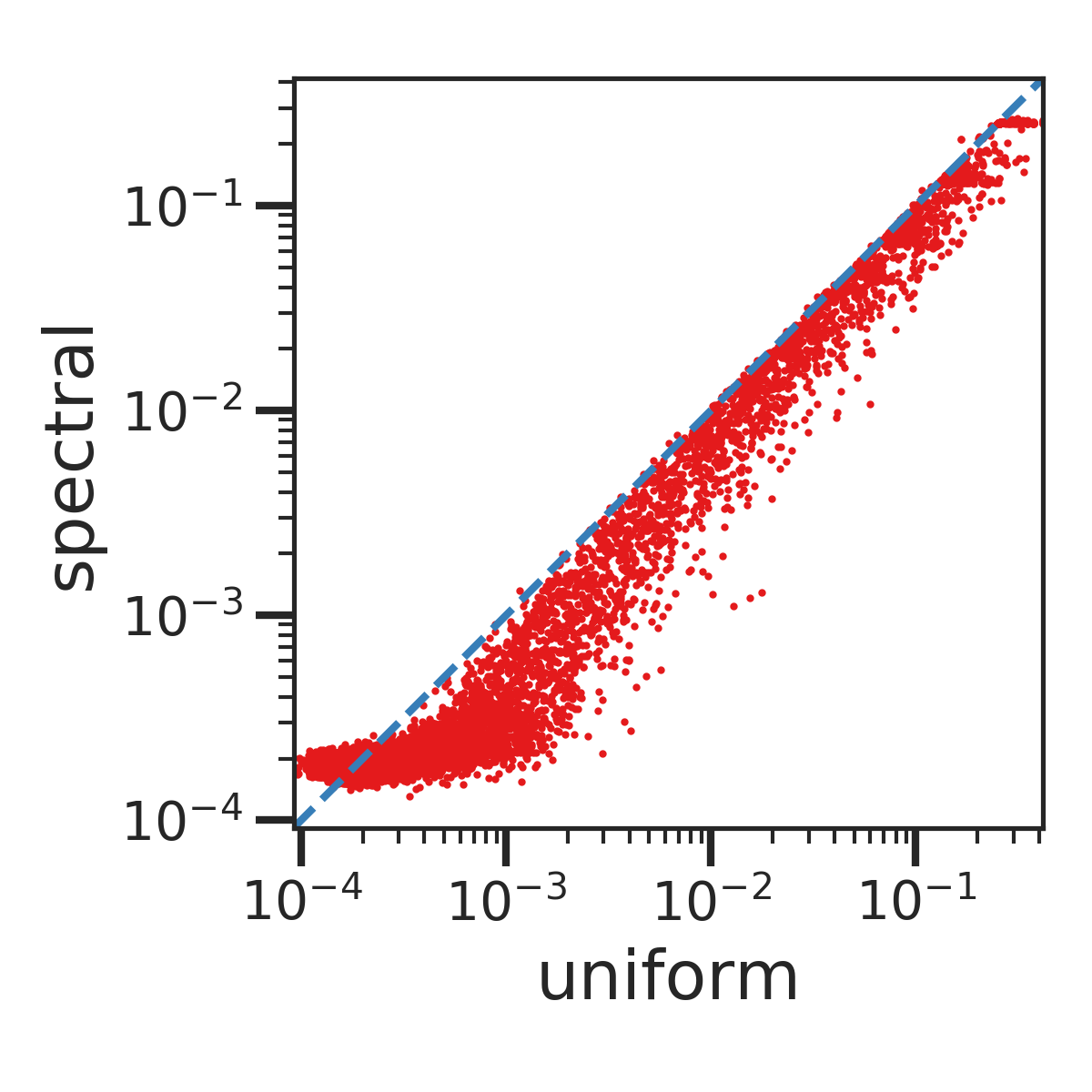} 
\end{minipage}
\begin{minipage}{0.25\linewidth}\centering
   {\small \hspace{2em} unweighted   } \\
   \includegraphics[height=1.7in]{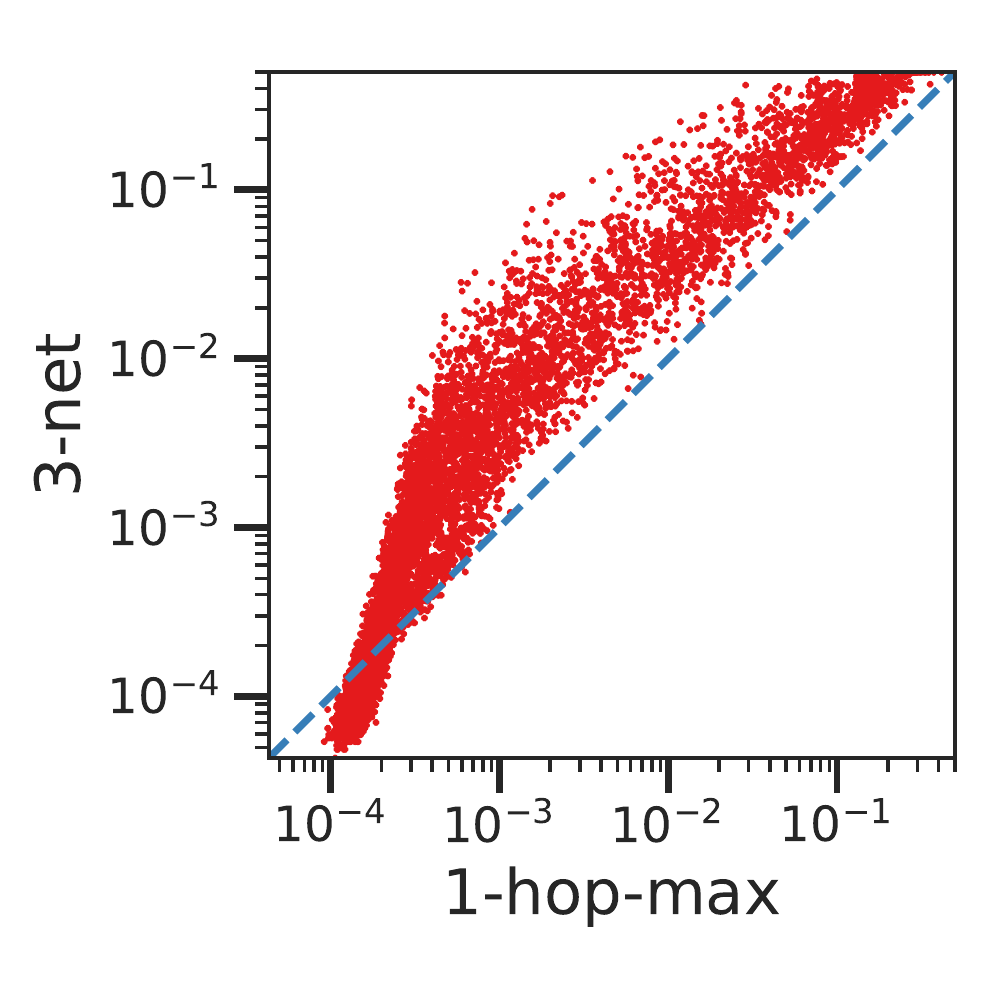} 
\end{minipage}
\caption[Comparing exposure probabilities under RGCR with different random clustering strategies]
{Scatterplot of the exposure probability $\prob{E_i ^\bOne \mid \mathcal P}$ of each node
in the \smallworld~network (first row) and \stanford~network (second row),
with different random clustering strategies. 
The blue dashed line marks the scenario when the exposure probability under two schemes are same.
In the first column, we compare the exposure probability at each node with uniform- (\ie, unweighted)
and spectral-weighted random 3-net clustering, and similar comparison for 1-hop-max clustering
is given in the second column. In the last column we compare the exposure probabilities
with unweighted 3-net and 1-hop-max clusterings.
}
\label{Fig:ProbVSProb}
\end{figure}

We further compare the exposure probability across different schemes in \Cref{Fig:ProbVSProb}.
In the first two columns, we examine how the spectral weighting affects nodes' exposure probabilities
associated the randomized 3-net and 1-hop-max clustering respectively. For 3-net, under both networks,
spectral weighting effectively increases the exposure probability of nodes whose probability 
is small under the unweighted 3-net, at a very small cost of decreasing some large exposure probabilities.
In contrast, for 1-hop-max clustering, even though spectral weighting can increase the very small
exposure probabilities seen in the unweighted scheme, it also significantly decreases the probability
of many other nodes.
Comparing 3-net clustering and 1-hop-max as in the last column, we observe that the exposure probabilities
under 3-net are mostly higher than under 1-hop-max, though the smallest exposure probability under 1-hop-max
is higher due to the improved lower bound theory in \Cref{Thm:prob_LB_improved}.

\subsection{HT estimator variance}\label{sec:sim_HT_RGCR}
We now examine the variance of $\hat \tau$ under the RGCR scheme
with each random clustering strategy. In both the \smallworld\ and
\stanford~networks, we consider 3-net and 1-hop-max clusterings, each under
uniform-, spectral-, and degree-weighting schemes. We also consider both
independent and complete randomization at the cluster level.

The variances are obtained from \Cref{Eq:var-mean_outcome,Eq:var-GATE,Eq:covar-mean_outcome},
the exact ground-truth variance (available in simulations).
In addition to the exposure probability of each node 
$\prob{E_i ^\bOne \mid \mathcal P}$, the variance formulae require 
the exposure probabilities of each pair of nodes, \ie, 
$\prob{E_i ^{\mathbf z_1} \cap E_j ^{\mathbf z_2} \mid \mathcal P}$
for any node pair $i, j \in V$ and $\mathbf z_1, \mathbf z_1 \in \{\bOne, \bZero\}$.
We also estimate these co-exposure probabilities using Monte Carlo estimation
with $Kn$ stratified samples, with $K = 128$ for the \smallworld~network
and $K = 16$ for the \stanford\ network.

\begin{figure}[t] 
   \centering
\begin{minipage}{0.45\linewidth}\centering
   \includegraphics[width=2.9in]{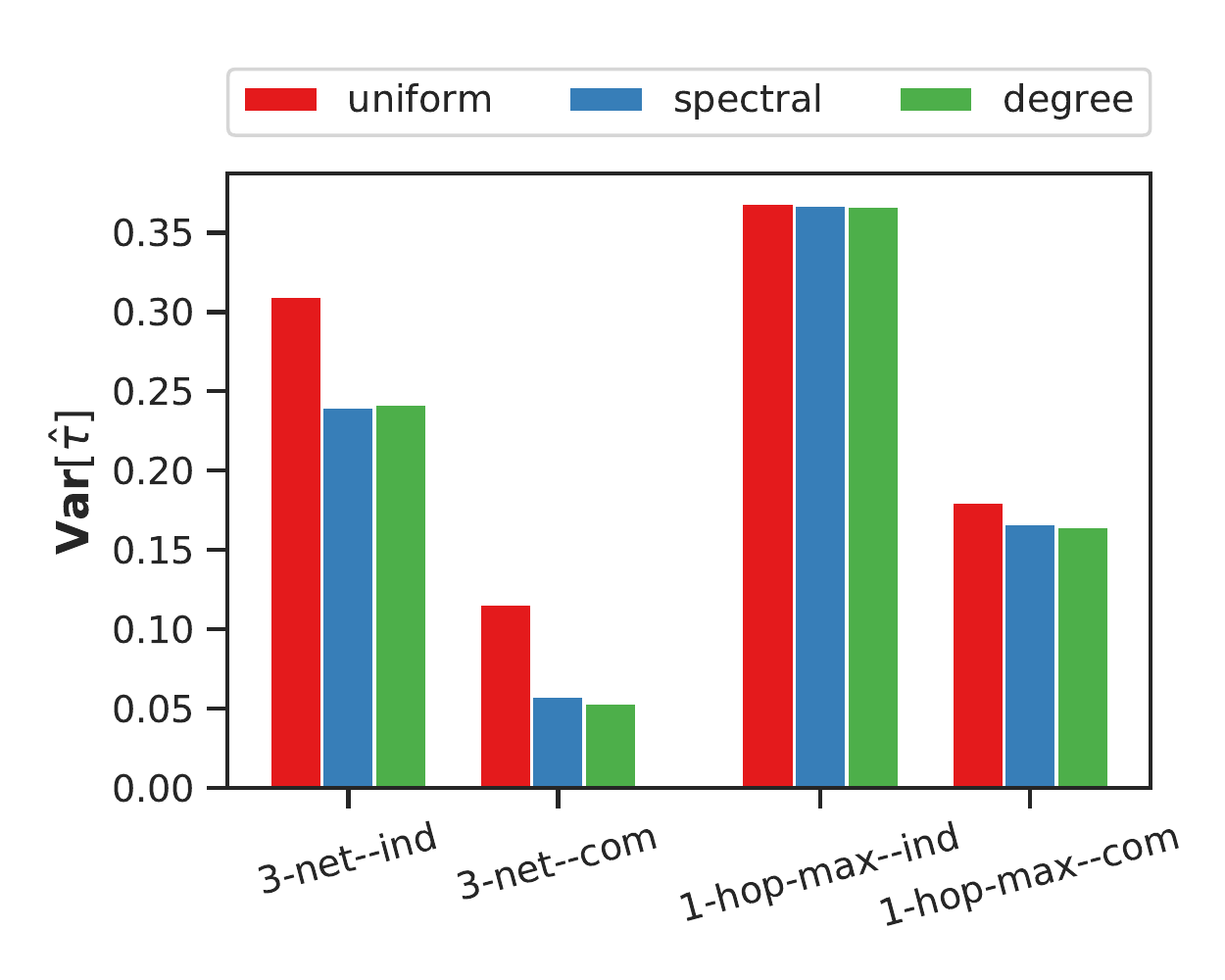} 
   \\{\footnotesize \smallworld   }
\end{minipage}
\begin{minipage}{0.45\linewidth}\centering
   \includegraphics[width=2.9in]{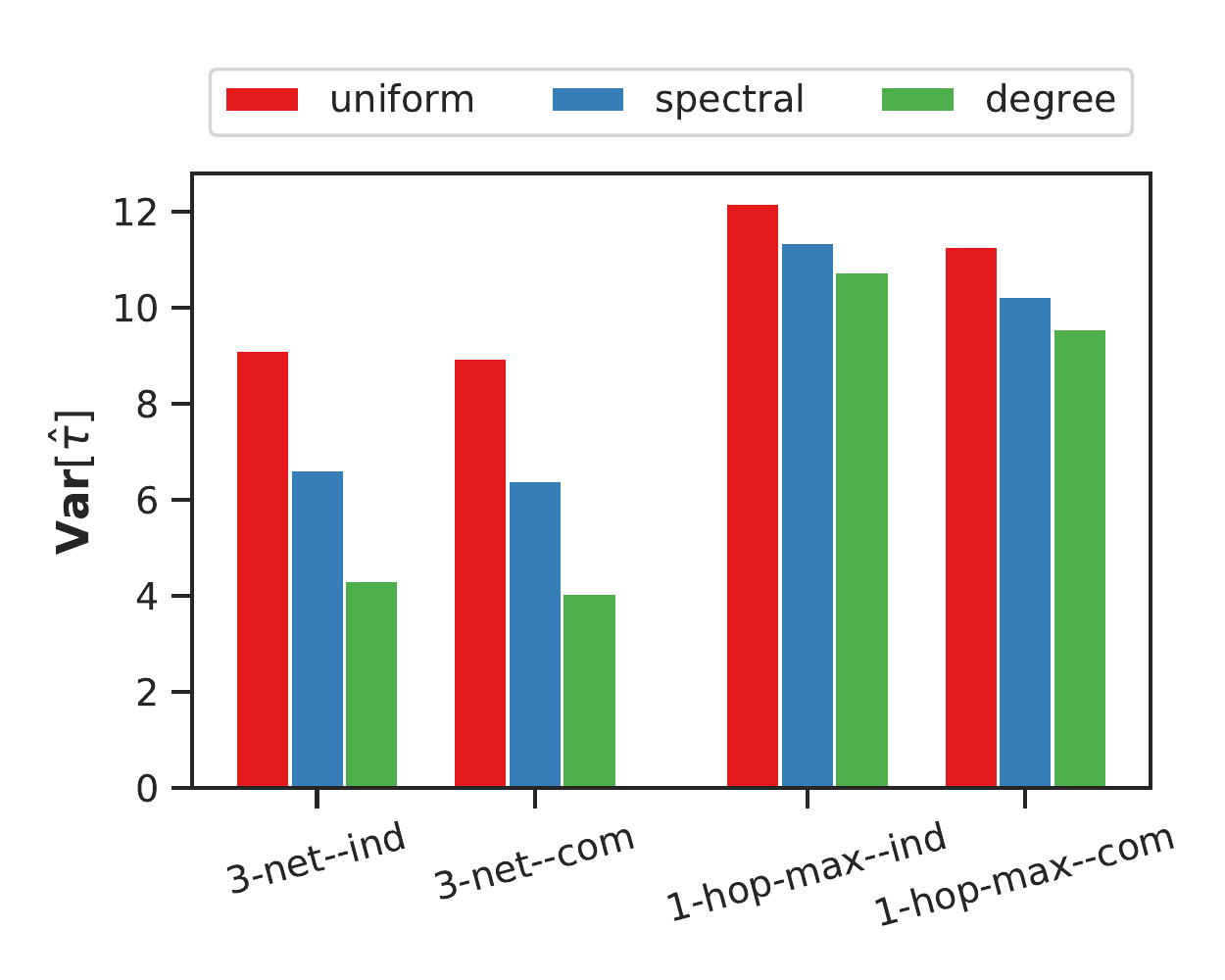} 
   \\{\footnotesize \stanford   }
\end{minipage}
\caption[Variance of $\hat\tau$ under RGCR with different random clustering strategies]
{Variance of the HT GATE estimator under the RGCR scheme with
various random clustering strategies.The suffix of each clustering method
distinguishes independent (-{}-ind) or complete (-{}-com) randomization at the
cluster level. The variance of these HT estimators under GCR, not shown,
are all dramatically higher (comparable to results in \Cref{Fig:VarDist}). }
\label{Fig:Var-rgcr}
\end{figure}

The results are shown in \Cref{Fig:Var-rgcr}, where we make four observations.
First, for both networks and every random clustering and weighting scheme, 
complete randomization yields lower variance than independent randomization.
Such variance reduction can be explained by the negative correlation introduced 
in the cluster-level assignment process, leading to larger values of 
$\prob{E_i ^\bOne \cap E_j ^\bZero \mid \mathcal P}$
and a positive $\cov{\hat \mu(\bOne), \hat \mu(\bZero)}$.
The reduction is less significant for the Facebook network,
which may be due to the growth structure (see Appendix~\ref{app:growth})
being quite different.

Second, in both networks and with both the 3-net and 1-hop-max
random clustering strategy, the variance with the spectral-weighted
scheme is lower than that of the unweighted version.
Such reduction can be explained by the increase in the lowest exposure probabilities
(see the first two columns of \Cref{Fig:ProbVSProb}).
This variance reduction is less significant for 1-hop-max clustering, which is
consistent with how the spectral-weighted scheme also decreases 
the exposure probabilities of most nodes (second column of \Cref{Fig:ProbVSProb}).

Third, we observe that degree-weighting usually gives lower variance
than the spectral-weighted scheme.
According to the discussion begun in \Cref{sec:weighted_opt_choice}, spectral weighting
achieves a uniform lower bound on the exposure probability of every node,
but this lower bound is less tight for nodes with smaller $\lvert B_2(i) \rvert$
(see \Cref{Fig:ProbVS2Ball}), and thus it is not unexpected that a weighting
favoring the nodes with large neighborhood, as degree-weighting does, 
would increases the minimum exposure probabilities of all nodes and reduce variance. 

We conclude that for HT estimators, RGCR with randomized 3-net and
1-hop-max are generally comparable.
Between the two, randomized 3-net usually yields a modestly lower variance. 
According to \Cref{Lem:1_hop_max-local_dep}, 1-hop-max 
has a local dependency property and thus the cross-node terms in the variance formulae
(\Cref{Eq:var-mean_outcome,Eq:var-GATE,Eq:covar-mean_outcome})
decay significantly in the distance between node pairs 
(and become zero when the distance is greater than 4).
In contrast, randomized 3-net has no local dependency guarantee,
and consequently we do not have a nontrivial theoretical upper bound
on the variance. 
However, in our simulations the cross-node terms are also small,
making the variance of the HT estimator under randomized 3-net
even lower than using 1-hop-max.

In summary, for the HT GATE estimator for the response model and networks we study:
\begin{itemize}
\item complete randomization yields lower variance
than independent randomization, 
\item randomized 3-net clustering yields lower variance than
1-hop-max,
\item spectral- and degree-weighting schemes yields lower variance than unweighted schemes.
\end{itemize}

\subsection{H\'ajek estimator bias and variance}
\label{sec:simulation_hajek}

Unlike the HT estimator, the \hajek~estimator does not have close-form formulae
to compute the bias, variance, or MSE, and thus we estimate
these quantities via simulation. We therefore briefly describe how we evaluate performance via simulation. 
For GCR, since the estimation performance is associated with the specific clustering is use,
we use the median bias, variance, and MSE across 1000 randomly generated clusterings.
Specifically for each clustering, we simulate the experimental procedure 
(assignment, outcome generation, and GATE estimation) and compute the sample bias, 
variance, and MSE and use as the proxy of the corresponding measure of the GCR scheme.
For RGCR, we simulate the experimental procedure (random clustering generation,
assignment, outcome generation, and GATE estimation) $50*n$ times
and analogously estimate each measures with the sample bias, sample variance, 
and sample MSE. 
As in \Cref{sec:properties_weighted}, here we use a similar idea of stratified sampling
 for the weighted 3-net clustering design:
for each node, there are 50 times when it is ranked first among all nodes
in generating the random clustering and guaranteed to be a seed node 
in the 3-net clustering.
In the analysis phase we use the exposure
probabilities estimated for RGCR in \Cref{sec:est_expo_prob}. For GCR
we use the exact exposure probabilities associated with the clustering in use.

\begin{figure}[t] 
   \centering
{\small \smallworld} \\
\begin{minipage}{0.3\linewidth}\centering
   \includegraphics[width=1.9in]{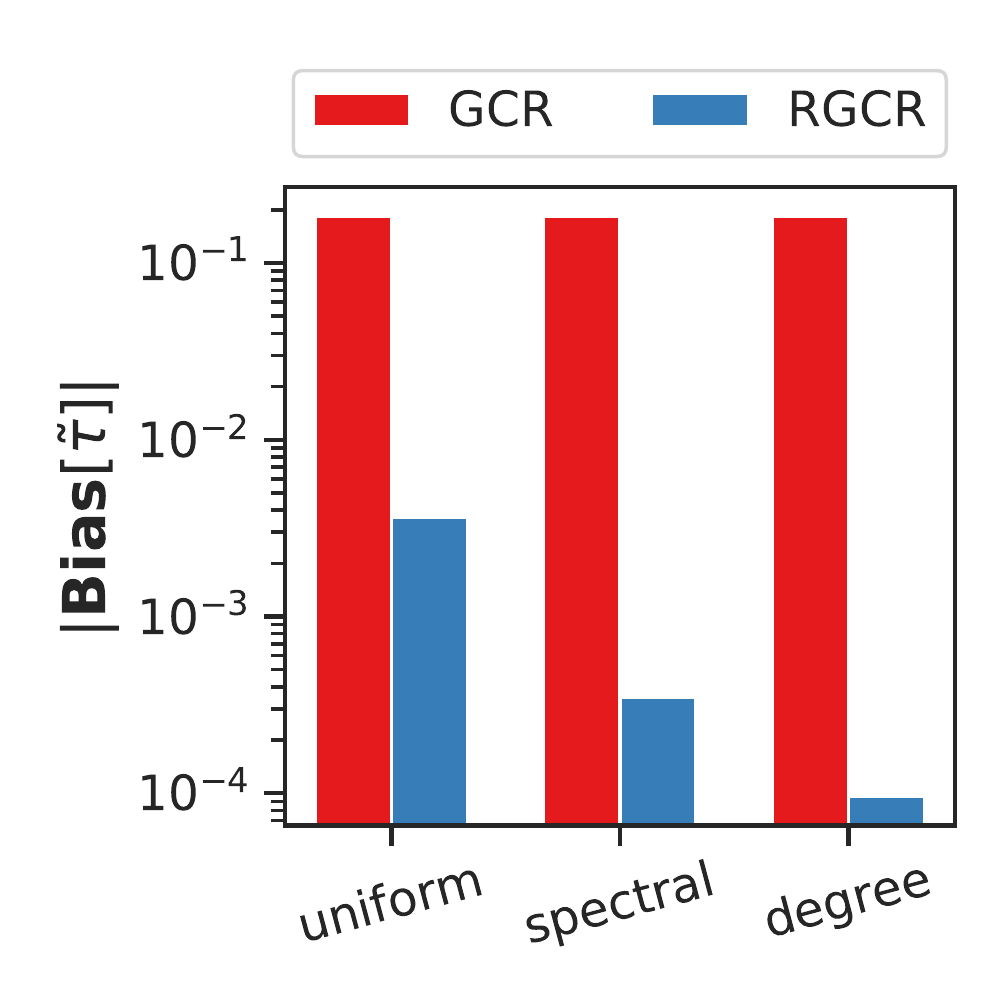} 
\end{minipage}
\begin{minipage}{0.3\linewidth}\centering
   \includegraphics[width=1.9in]{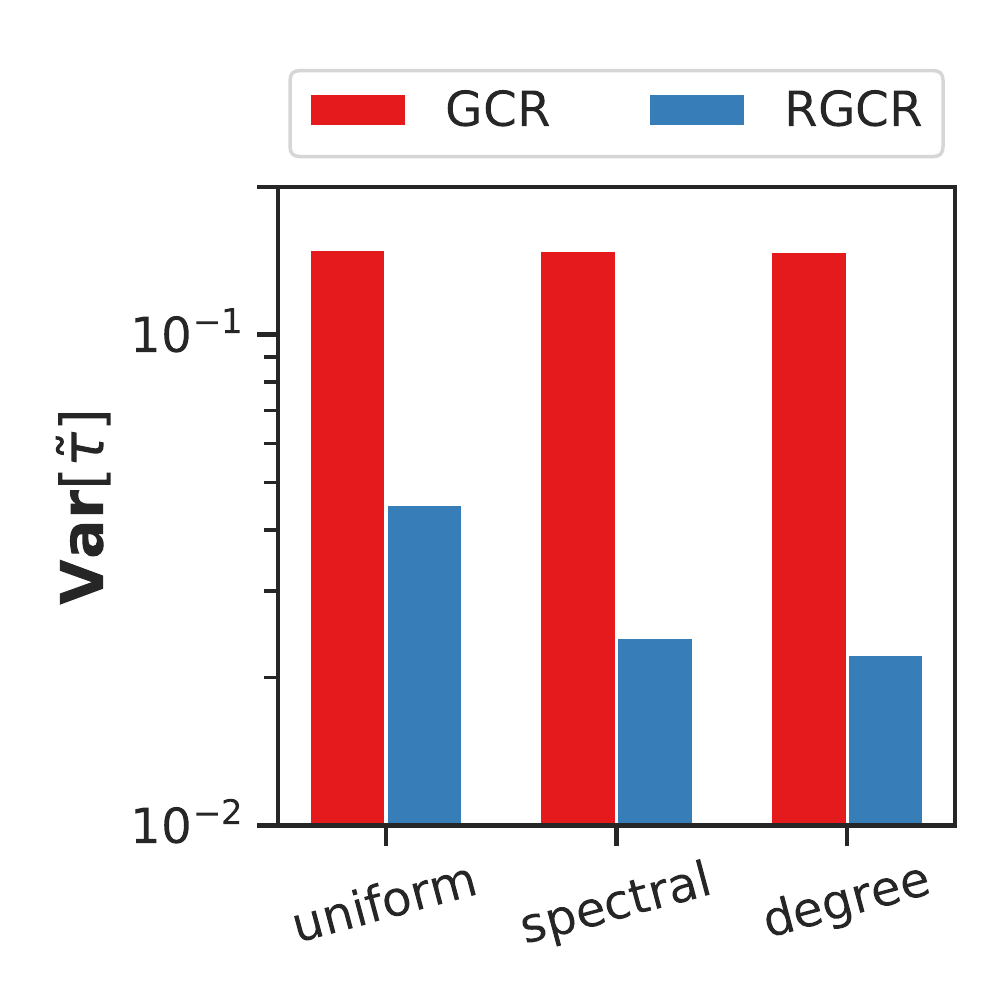} 
\end{minipage}
\begin{minipage}{0.3\linewidth}\centering
   \includegraphics[width=1.9in]{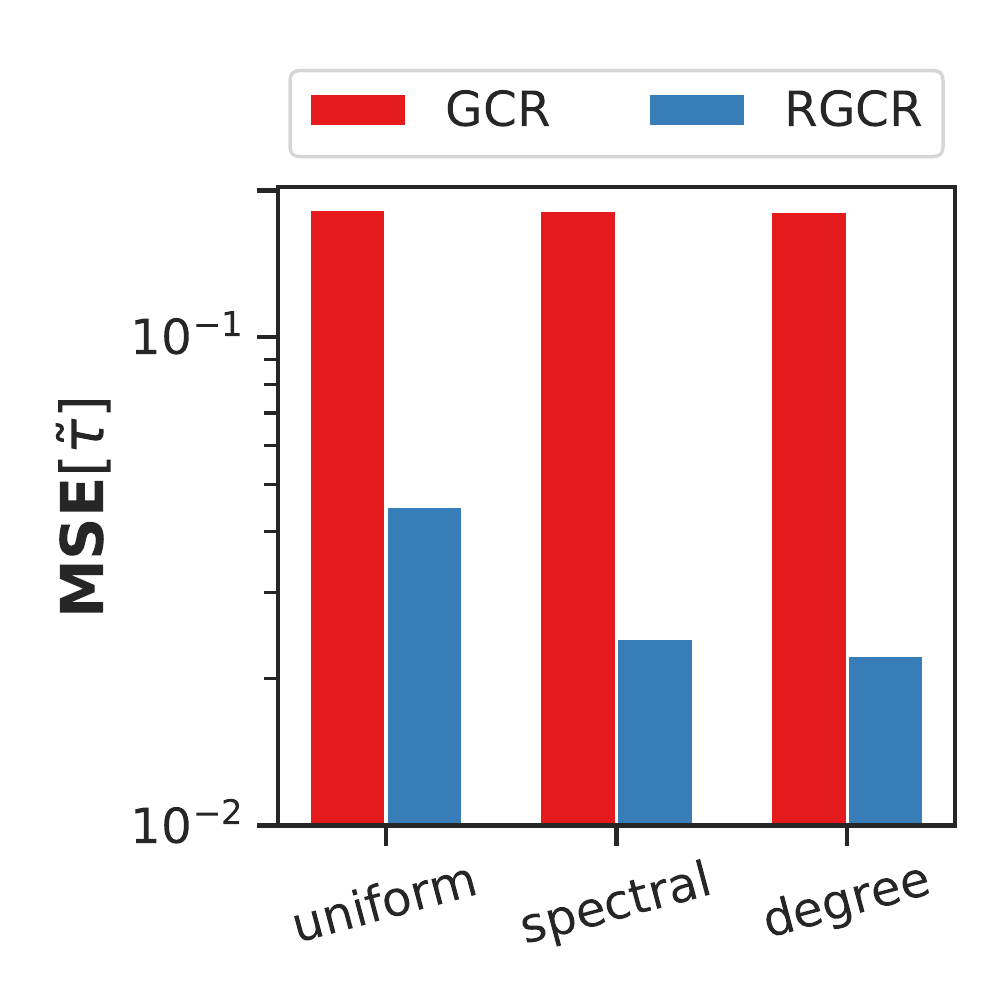} 
\end{minipage}
{\small \stanford} \\
\begin{minipage}{0.3\linewidth}\centering
   \includegraphics[width=1.9in]{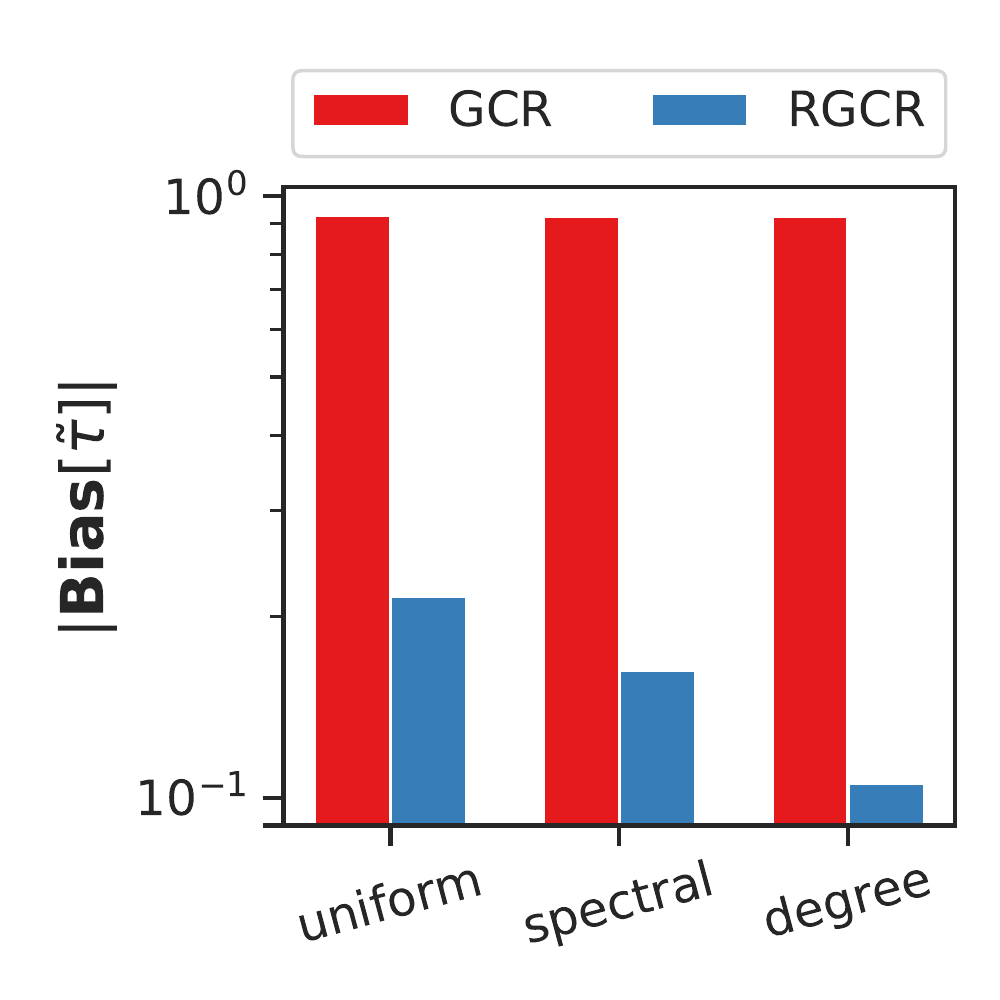} 
\end{minipage}
\begin{minipage}{0.3\linewidth}\centering
   \includegraphics[width=1.9in]{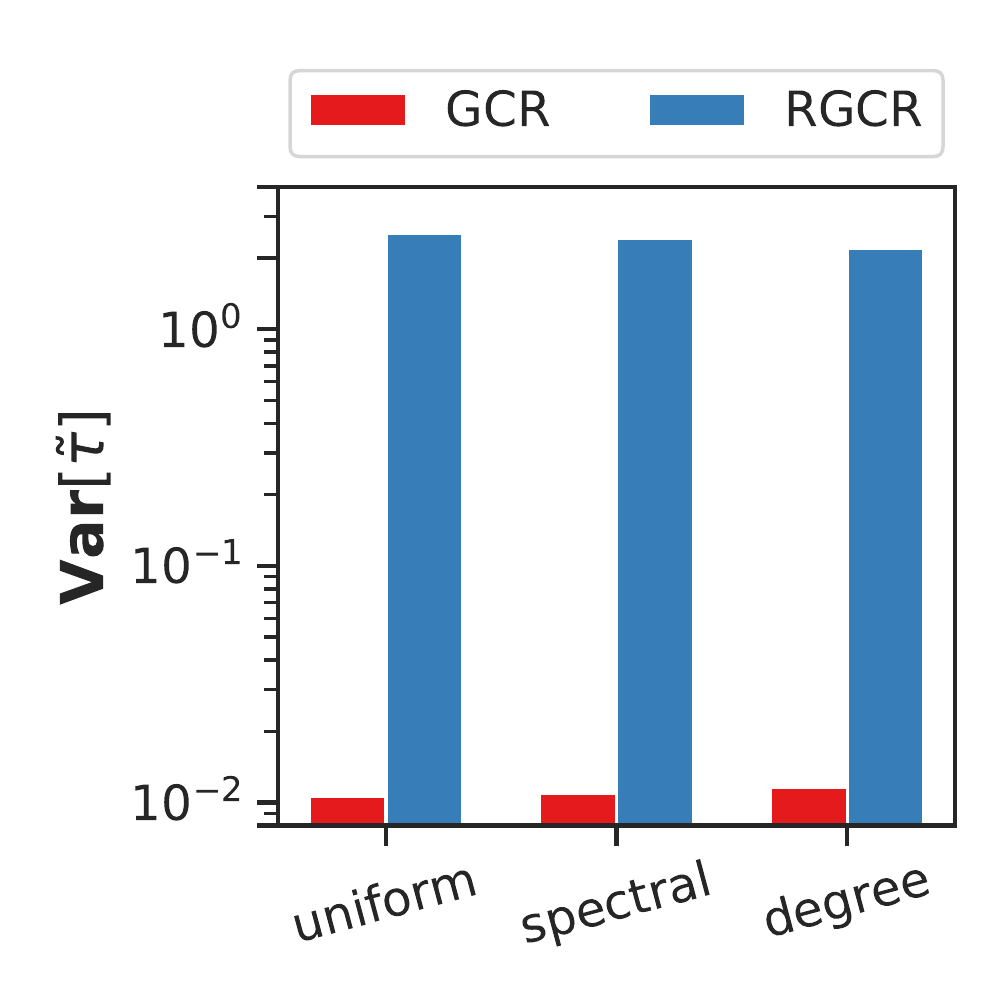} 
\end{minipage}
\begin{minipage}{0.3\linewidth}\centering
   \includegraphics[width=1.9in]{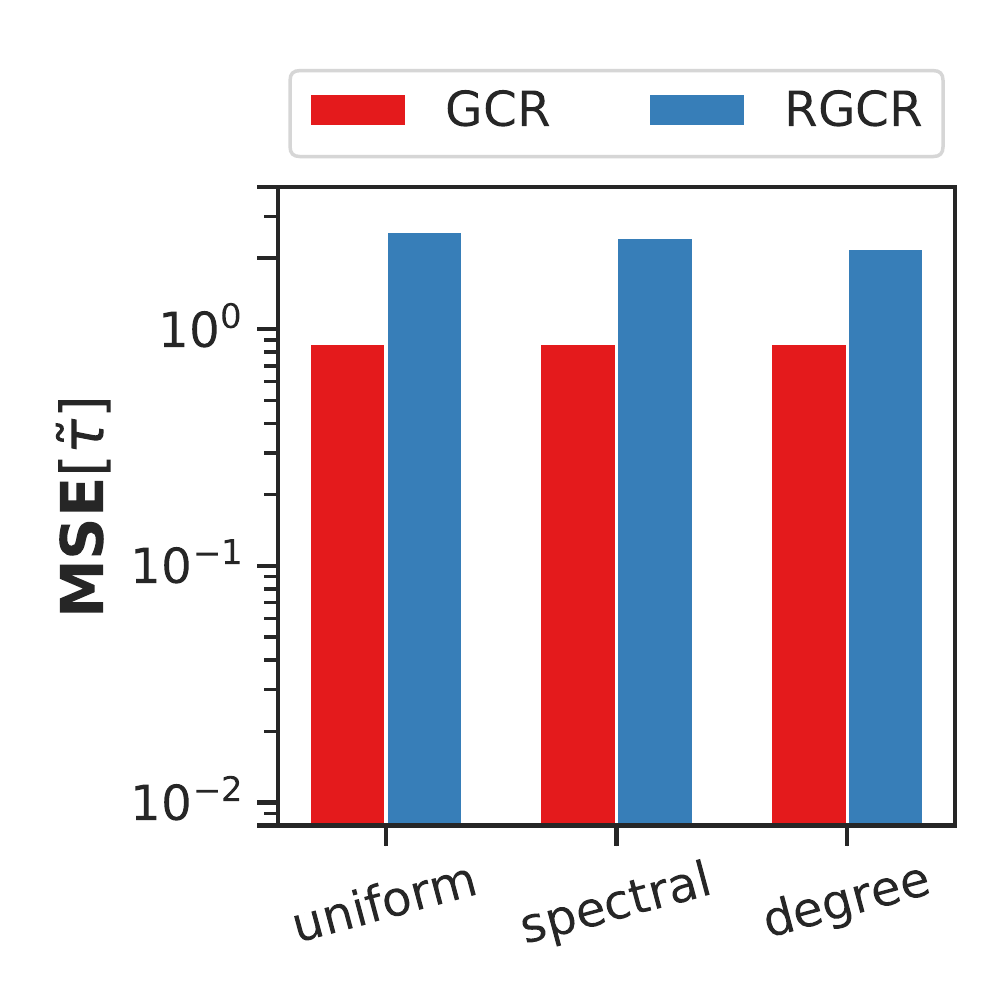} 
\end{minipage}
\caption{Bias, variance, and MSE of the \hajek~GATE estimator under GCR and RGCR
with 3-net clustering and 
independent randomization 
in the \smallworld~(first row) and \stanford~(second row) networks.
}
\label{Fig:Hajek-gcr}
\end{figure}

\Cref{Fig:Hajek-gcr} presents the bias, variance, and MSE of the \hajek~estimator under GCR and RGCR, focusing on independent randomization. 
We make two main observations. 
First, the bias of the \hajek~estimator under RGCR has been significantly reduced,
compared with GCR. Recall that in our response model (\Cref{sec:response_model}) the ground-truth GATE is $\tau = 1.0$. The bias is reduced from 15\% to less than 0.5\% in the 
\smallworld~network, and from 93\% to less than 20\% in the \stanford~network.

One can intuitively interpret this bias reduction as follows.
Under GCR, for nodes with exponentially small exposure probability (in the \stanford~network,
it can be lower than $10^{-50}$), they are almost never network exposed to treatment or control. As a result, their response $Y_i$ are almost never revealed in the weighted averaging procedure
of the \hajek~estimator. In every experiment execution, the exposure nodes are 
almost only those with large exposure probabilities, and mostly those with small degree.
A large degree node may also have a large exposure probability if it is at the center of a cluster; 
however, it is much less likely to be at the center of a cluster than a low-degree node, 
and even if it is exposed, its weight (the inverse of the exposure probability) is 
not larger than those low-degree exposed nodes.
Therefore, the analysis procedure in GCR is internally biased against the large degree nodes, 
favoring the low-degree nodes, which have smaller individual treatment effects $\tau_i$ in our response model. Consequently, we see how the \hajek~estimator can be severely biased downwards in such settings.

In contrast, for RGCR the exposed nodes
are not so strictly high exposure probabilities. For example, if a large degree
node is at the center of a cluster in a randomly generated clustering,
even though it has large \emph{conditional} exposure probability under
this clustering and thus being likely to be network exposure to treatment or control, 
its unconditional exposure probability can still be small. As a result, it is weighted
more heavily in the weighted average procedure of \hajek~estimation,
making the estimator value shift towards the response of large degree nodes
and thus less biased than that under GCR.

Alongside this understanding of \hajek\ bias, it is also expected to observe an increase
in variance from RGCR, vs.\ GCR, under \hajek~estimation.
In~\Cref{Fig:Hajek-gcr} we see variance reduction from RGCR in the Small World network
but an increased variance in the \stanford~network.
Under GCR, since the estimator value is dominated by the response of low degree nodes,
in our response model the response of low-degree nodes have a much narrower range than the whole population, 
resulting in low variance (but overwhelming bias, we repeat).

Finally, we also compare the bias and MSE of the RGCR scheme with different
random clustering strategies, which we also include complete randomization,
and the results are given in \Cref{Fig:Hajek-rgcr}.
In general, the benefits of complete randomization (over independent randomization) that we see for the HT estimator do not appear to carry over to \hajek\ estimation. 

In summary, comparing with the GCR scheme, the \hajek~estimator under the RGCR scheme
has significantly lower bias but may have larger variance. Examining the mean squared error (MSE) that trades off bias and variance, we see a lower MSE in the \smallworld~network from RGCR  (vs.\ GCR),
while we see a higher MSE in the \stanford~network under RGCR (vs.\ GCR).

\begin{figure}[t] 
   \centering
{\small \smallworld} \\
   \includegraphics[width=6in]{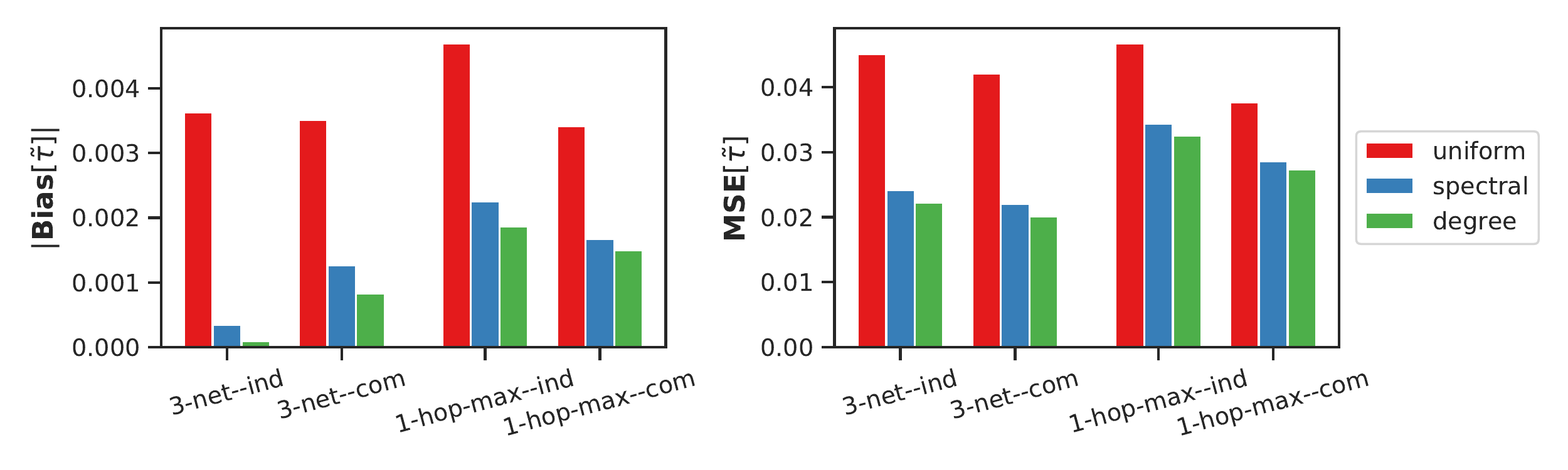} 
\vspace{-1em} \\
{\small \stanford} \\
   \includegraphics[width=6in]{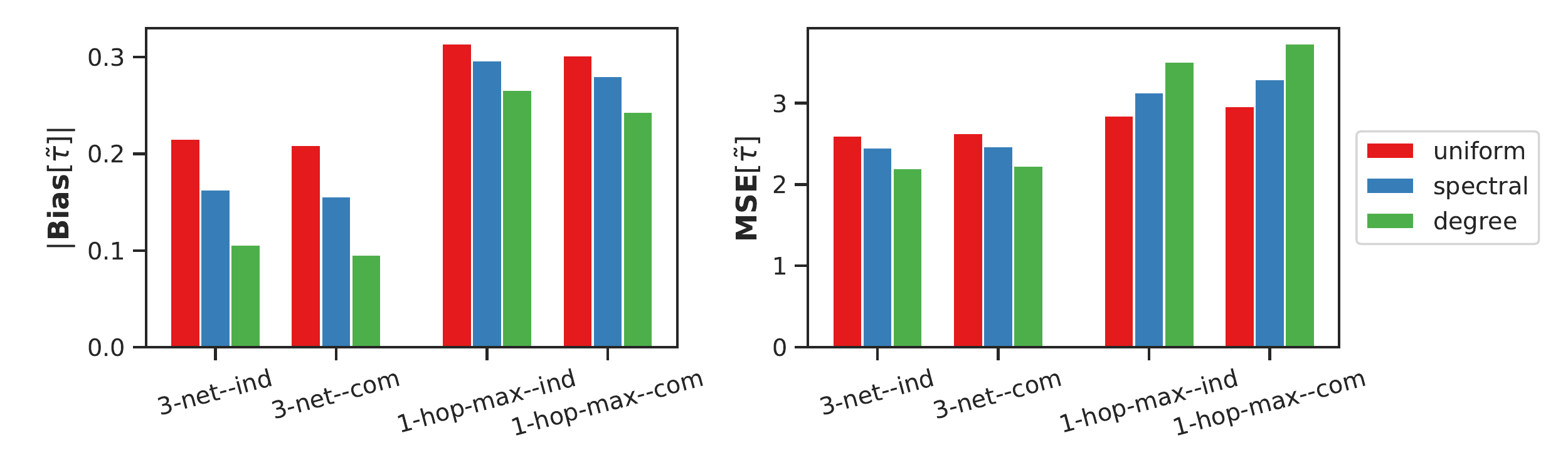} 
\caption{Bias and MSE of H\'ajek GATE estimator under RGCR with various clustering algorithms with
both independent and complete randomization.}
\label{Fig:Hajek-rgcr}
\end{figure}

\subsection{Variance, bias, and network size}\label{sec:sim_diff_n}
Here we examine  how the bias, variance, 
and mean squared error (MSE) change as a function of network size. 
Recall that the heavy-tailed small-world network we use throughout our earlier simulations
is based on a periodic two-dimensional lattice of side length 96 (thus $n = 96*96=9216$).
Here we consider a series of heavy-tailed small-world networks in this family.
Specifically, 
we generate a sequence of networks
based on lattices of size 16*16, 24*24, 32*32, 48*48, 64*64, and 96*96,
while the procedure (and parameters) for adding long-range edges
remains fixed. 
With this sequence of networks, in
\Cref{Fig:Hajek-gcr-SWx}
we repeat the above simulation procedures and compare the 
bias, variance, and MSE under the GCR and RGCR schemes. This analysis examines both the Horvitz--Thompson (HT) 
and \hajek\ estimators.

\begin{figure}[t] 
   \centering
\begin{minipage}{0.25\linewidth}\centering
   \includegraphics[height=1.7in]{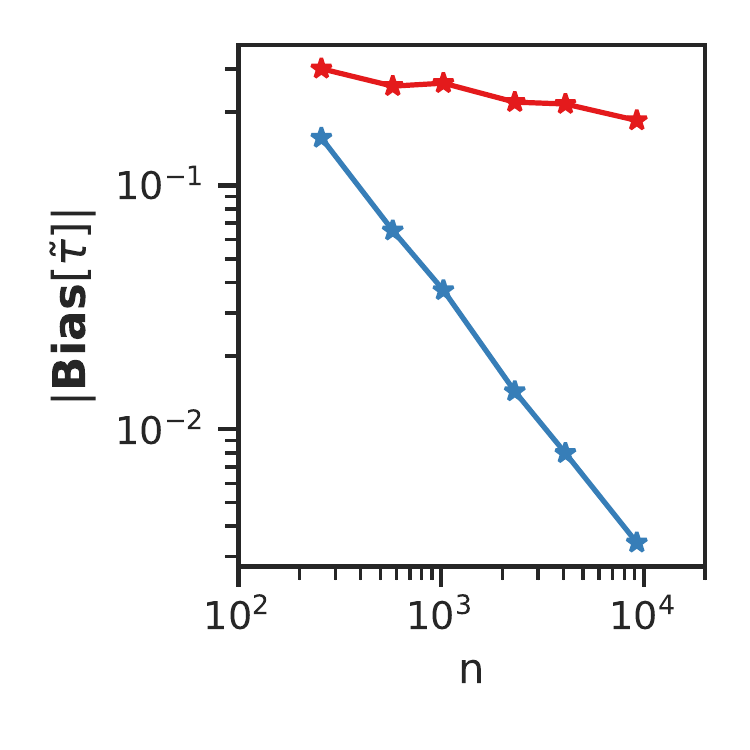} 
\end{minipage}
\begin{minipage}{0.25\linewidth}\centering
   \includegraphics[height=1.7in]{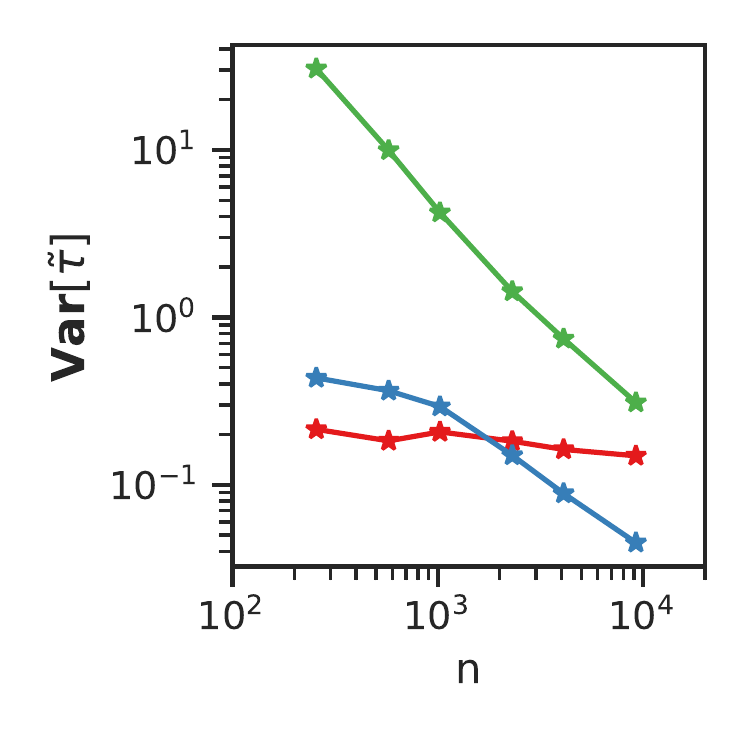} 
\end{minipage}
\begin{minipage}{0.4\linewidth}\centering
   \includegraphics[height=1.7in]{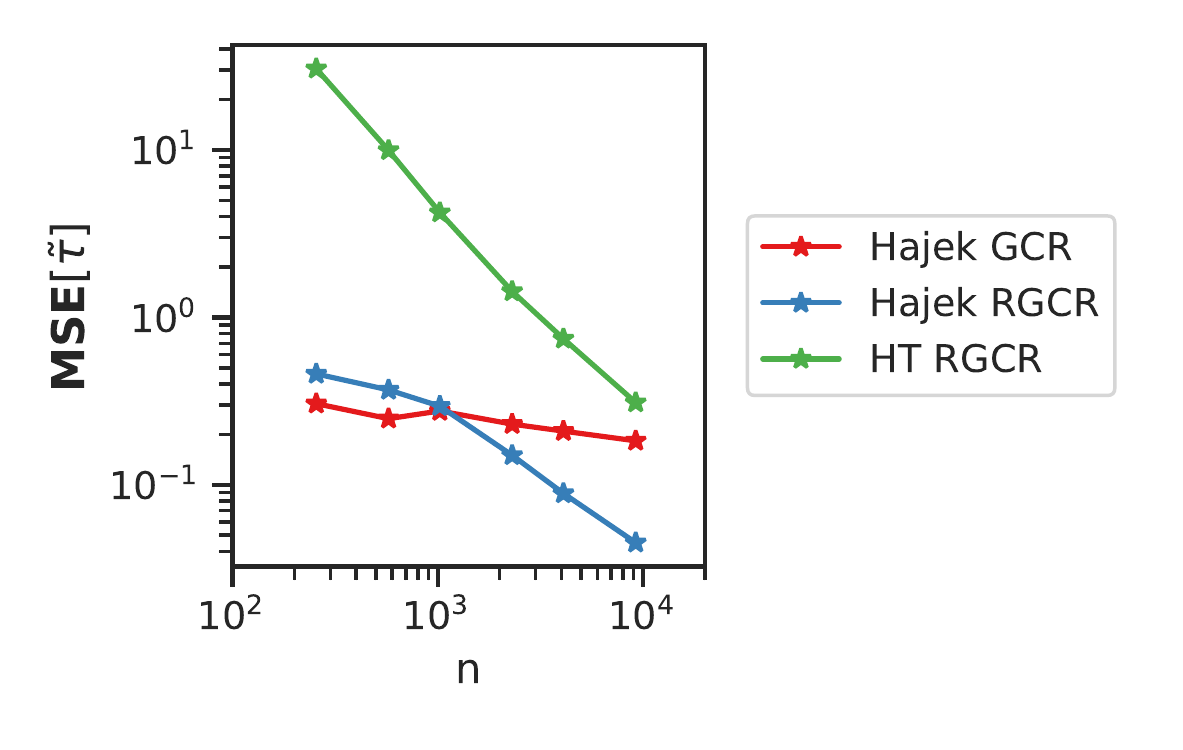} 
\end{minipage}
\vspace{-5mm}
\caption{Bias (absolute value), variance, and MSE of the \hajek~and HT GATE estimator under GCR and RGCR with unweighted 3-net clustering
in a sequence of increasingly-sized heavy-tailed small-world networks.
}
\label{Fig:Hajek-gcr-SWx}
\end{figure}

From the first plot in \Cref{Fig:Hajek-gcr-SWx}, we observe that the \hajek~estimator
under the RGCR scheme has consistently less bias than under the GCR scheme,
across the range of network sizes we study. Moreover, the bias decays with the network size at a much higher rate for RGCR than for GCR. 
Recall that the HT estimator is unbiased (under both GCR and RGCR).

From the second plot, we observe that the variance of the \hajek\ estimator under RGCR
also decays with higher rate than under GCR.
We can explain the slow decay of the variance under GCR
by observing that  as the number of nodes $n$ increases, so does the number of 
nodes with small exposure probabilities.
It then becomes more likely that more nodes with small exposure
probabilities have been network exposed.
Note that variance of response among these nodes are large,
and thus even though the expected number of exposed nodes 
increases proportional to $n$, 
due to the introduction of small exposure probability nodes, the variance of the \hajek~estimator
decays at a much smaller rate than $n^{-1}$.
In contrast, the empirical rate of variance decay of RGCR is close to $n^{-1}$.
Therefore, even though the variance under GCR might
be lower in small networks, for the reason explained in \Cref{sec:simulation_hajek}, 
the RGCR scheme can significantly reduce variance in large networks.
Meanwhile, the variance of the HT estimator under RGCR, while higher than that of the \hajek\ estimator with RGCR, also exhibits an empirical decay rate of $n^{-1}$, 
The variance of HT estimator under GCR is extremely high across all networks and is not presented in the plot.

Combining both bias and variance, we observe that the MSE of all three methods decays
with the network size, while the decay rate is notably faster for the estimators based on the RGCR scheme.
In summary, with a large interference network, the \hajek~estimator with RGCR is preferred.

\section{Conclusion}
\label{sec:conclusion}

We developed randomized graph cluster randomization (RGCR) as a scheme for 
the design and analysis of randomized experiments in the presence of interference. 
This scheme is an improvement on the graph clustering randomization (GCR) scheme
in that it is based on a distribution of random clusterings instead of a single fixed clustering,
with favorable consequences for the bias and variance of standard estimators. 
Compared to GCR, the RGCR scheme with proper random clustering generators enjoys significantly 
reduced variance for both the Horvitz--Thompson and H\'ajek estimator of the GATE, and also supports 
complete randomization.
We also discuss how the network drift pattern in nodes response, as is observed in real-world settings, 
plays an important role in the variance of GATE estimation, and propose a new response model exhibiting
homophily in the form of network drift in responses, facilitating a more careful analysis of 
realistic estimator performance.

\footnotesize
\bibliographystyle{abbrv}
\bibliography{paper_main}  

\newpage
\appendix
\section{Empirical study of social network growth rates}
\label{app:growth}

As described in the introduction, the theoretical analyses in this work are developed under either a bound on the maximum degree $\dmax$ of the interference graph or a stronger assumption of bounded geometry, assuming that the interference graph satisfy a restricted growth condition with coefficient $\kappa$.
Given the central role that bounded geometry plays in our theoretical analysis, 
in this appendix we present an empirical study on social network growth statistics. 

We use the Facebook100 datasets~\cite{traud2011comparing,traud2012social,jacobs2015assembling}, a collection of complete Facebook friendship networks at 100 American institutions collected and released in September 2005. 
The networks are quite diverse, most basically varying in size from $672$ to $>30,000$ nodes, which allows us to also understand how the growth statistics can vary with network size and other properties.
In~\Cref{tab:ball_size_both} we present growth statistics from a random subset of 25 networks from the collection, ordered by size $n$.

The average growth geometry of the full population of 100 networks in the FB100 collection is illustrated in \Cref{fig:growth-fb100}. The more fine-grained growth of the \smallworld\ and \stanford\ networks are illustrated in \Cref{fig:growth-sw-and-stanford}. 

We here give a concise summary of specific observations from \Cref{tab:ball_size_both} and these figures.
First, per \Cref{tab:ball_size_both}, diameter appears to be independent of network size. This is not surprising, as diameter is a fragile metric known to be sensitive to whiskers in the network. Second, the maximum degree, restricted growth coefficient, and network size all appear to be positively correlated.
Third, an observation that impacts how we interpret our theoretical results, the restrictive growth coefficients $\kappa$ are large and all above 100. They are typically 25\%-50\% of the max degree $\dmax$.

Looking closer at the results across netowrks in both \Cref{tab:ball_size_both} and \Cref{fig:growth-fb100}, regarding the average ball-size at each radius $r$, we see that for $1 \leq r \leq 3$, the normalized average ball size ratio decreases with $n$, while for $r \geq 4$,  the normalized average ball size saturates near $1$.
For the maximum ball-size at each radius $r$, for $r=1$, $\max|B_1|$ increases with $n$ while the ratio $\max|B_1| / n$ decreases with $n$. For $r = 2$, $\max|B_2|$ is usually more than 90\% of the nodes, and always at least 70\%. Finally, for $r \geq 3$,  $\max|B_r|$ is always more than 98\% of the nodes.

\begin{figure}[htbp] 
   \centering
\begin{minipage}{0.4\linewidth}\centering
   \includegraphics[width=2.5in]{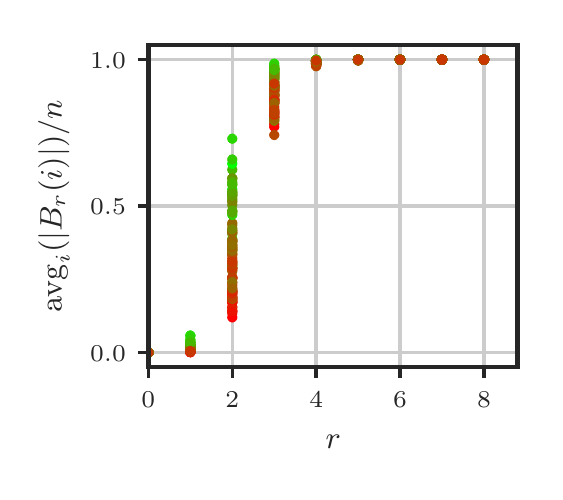} 
\end{minipage}
\begin{minipage}{0.4\linewidth}\centering
   \includegraphics[width=2.5in]{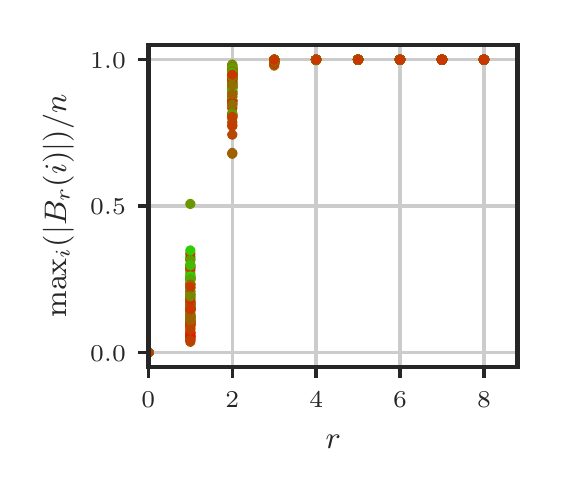} 
\end{minipage}
\caption{
Mean and max ball size at each radius $r$, for each of the 100 networks 
in the FB100 collection. The points are colored based on the size of the networks, $n$, 
with smaller networks green and larger networks red. 
}
\label{fig:growth-fb100}
\end{figure}

\begin{figure}[htpb] 
   \centering
{\small \smallworld} \\
   \includegraphics[width=4.8in]{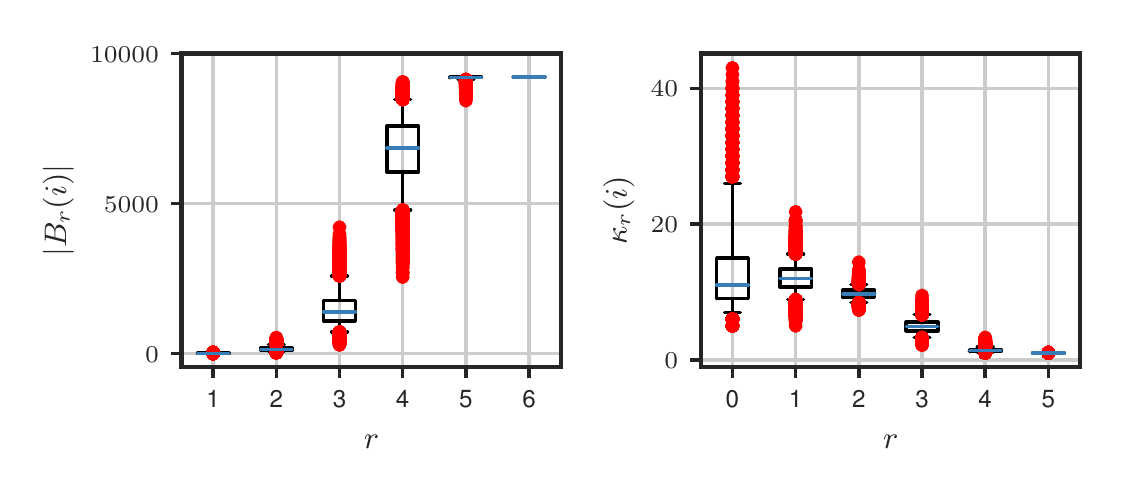} 
\\ 
{\small \stanford} \\
   \includegraphics[width=4.8in]{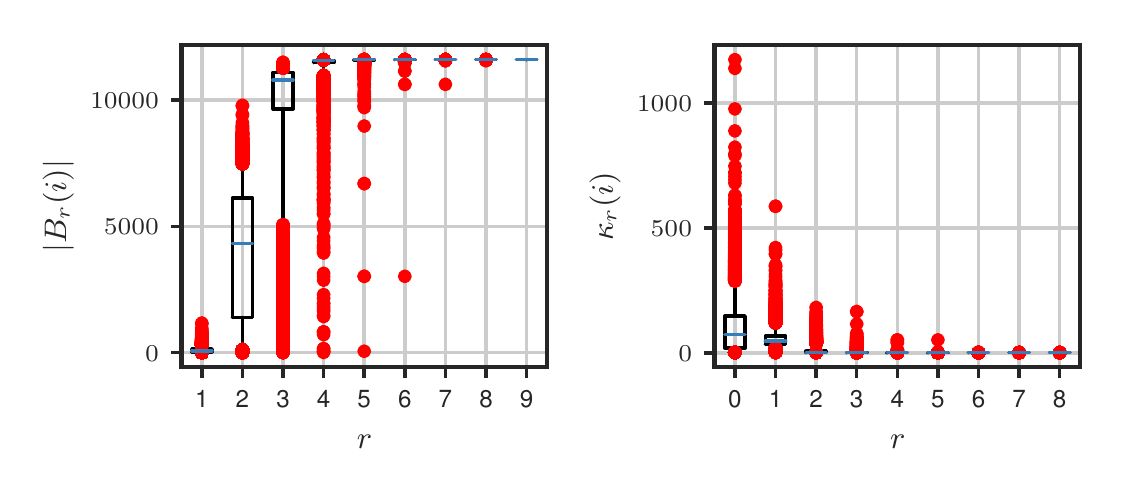} 
\caption{
The ball sizes $|B_r(i)|$ and local growth coefficient $\kappa_i = |B_{r+1}(i)| / |B_r(i)|$ 
of each node $i$  in the synthetic \smallworld\ and real-world \stanford\ network used for \Cref{sec:simulation}.
}
\label{fig:growth-sw-and-stanford}
\end{figure}

\newpage
\begin{landscape}
\begin{table}[p]
\centering
  \captionsetup{width=7in}
  \caption{Network statistics of 25 randomly selected networks from the Facebook-100 collection,
  listing the number of nodes $n$, number of edges $m$, diameter $D$, 
  average degree $\bar d = 2m/n$ maximum degree $\dmax$, and the restrictive
  growth coefficient $\kappa$. Additionally, we present the average
  ball size of each radius until $r=4$ and the maximum ball size until $r=3$, 
  both normalized by the number of nodes in the network.
 }
  \begin{tabular}{l   @{\hskip 15pt}   r r r r r r   c c c   c c c c }
    \toprule
  \multirow{3}{*}[6pt]{University} 
  & \multicolumn{1}{c}{\multirow{3}{*}[6pt]{$n$}}
  & \multicolumn{1}{c}{\multirow{3}{*}[6pt]{ $m$}} 
  & \multicolumn{1}{c}{\multirow{3}{*}[6pt]{$D$}}
  & \multicolumn{1}{c}{\multirow{3}{*}[6pt]{$\bar d$}} 
  & \multicolumn{1}{c}{\multirow{3}{*}[6pt]{$\dmax$}} 
  & \multicolumn{1}{c}{\multirow{3}{*}[6pt]{$\kappa$}}
  & \multicolumn{4}{c} {$\text{avg}_{i \in V} \lvert B_r(i) \rvert / n$} 
  & \multicolumn{3}{c} {$\max_{i \in V} \lvert B_r(i) \rvert / n$}\\
  \cmidrule(lr){8-11}   \cmidrule(lr){12-14}
  &  &  &  &  &  & & $r=1$ & $r=2$ & $r=3$ & $r=4$ & $r=1$ & $r=2$ & $r=3$ \\
    \midrule
    Caltech36 & 762 & 16651 & 6 & 43.70 & 248 & 102 & 0.0587 & 0.6448 & 0.9618 & 0.9986 & 0.3268 & 0.9357 & 0.9987 \\ 
Swarthmore42 & 1657 & 61049 & 6 & 73.69 & 577 & 149 & 0.0451 & 0.6594 & 0.9790 & 0.9991 & 0.3488 & 0.9523 & 0.9994 \\ 
Trinity100 & 2613 & 111996 & 6 & 85.72 & 404 & 148 & 0.0332 & 0.5735 & 0.9684 & 0.9980 & 0.1550 & 0.9529 & 0.9981 \\ 
Wellesley22 & 2970 & 94899 & 8 & 63.91 & 746 & 160 & 0.0219 & 0.4781 & 0.9221 & 0.9938 & 0.2515 & 0.9232 & 0.9963 \\ 
Pepperdine86 & 3440 & 152003 & 9 & 88.37 & 674 & 301 & 0.0260 & 0.5400 & 0.9420 & 0.9936 & 0.1962 & 0.9340 & 0.9968 \\ 
Mich67 & 3745 & 81901 & 7 & 43.74 & 419 & 157 & 0.0119 & 0.2945 & 0.8644 & 0.9906 & 0.1121 & 0.8179 & 0.9939 \\ 
Rice31 & 4083 & 184826 & 6 & 90.53 & 581 & 256 & 0.0224 & 0.5434 & 0.9675 & 0.9994 & 0.1425 & 0.9224 & 0.9995 \\ 
Wake73 & 5366 & 279186 & 9 & 104.06 & 1341 & 671 & 0.0196 & 0.5141 & 0.9561 & 0.9973 & 0.2501 & 0.9702 & 0.9983 \\ 
UChicago30 & 6561 & 208088 & 10 & 63.43 & 1624 & 813 & 0.0098 & 0.3375 & 0.8661 & 0.9832 & 0.2477 & 0.9218 & 0.9938 \\ 
UC64 & 6810 & 155320 & 8 & 45.62 & 660 & 282 & 0.0068 & 0.2103 & 0.7921 & 0.9777 & 0.0971 & 0.8026 & 0.9872 \\ 
WashU32 & 7730 & 367526 & 8 & 95.09 & 1794 & 898 & 0.0124 & 0.4244 & 0.9328 & 0.9957 & 0.2322 & 0.9578 & 0.9988 \\ 
Yale4 & 8561 & 405440 & 9 & 94.72 & 2517 & 1259 & 0.0112 & 0.4132 & 0.9082 & 0.9909 & 0.2941 & 0.9429 & 0.9961 \\ 
Georgetown15 & 9388 & 425619 & 11 & 90.67 & 1235 & 618 & 0.0098 & 0.3546 & 0.8946 & 0.9867 & 0.1317 & 0.8838 & 0.9923 \\ 
Northwestern25 & 10537 & 488318 & 9 & 92.69 & 2105 & 1053 & 0.0089 & 0.3624 & 0.9126 & 0.9936 & 0.1999 & 0.9503 & 0.9974 \\ 
Stanford3 & 11586 & 568309 & 9 & 98.10 & 1172 & 587 & 0.0086 & 0.3375 & 0.8529 & 0.9841 & 0.1012 & 0.8443 & 0.9906 \\ 
USF51 & 13367 & 321209 & 8 & 48.06 & 897 & 319 & 0.0037 & 0.1433 & 0.7428 & 0.9785 & 0.0672 & 0.7441 & 0.9915 \\ 
Northeastern19 & 13868 & 381919 & 9 & 55.08 & 968 & 393 & 0.0040 & 0.1721 & 0.8036 & 0.9850 & 0.0699 & 0.8017 & 0.9930 \\ 
UCSD34 & 14936 & 443215 & 9 & 59.35 & 2165 & 1083 & 0.0040 & 0.1877 & 0.8158 & 0.9868 & 0.1450 & 0.9129 & 0.9970 \\ 
UMass92 & 16502 & 519376 & 8 & 62.95 & 3684 & 1843 & 0.0039 & 0.2068 & 0.8621 & 0.9939 & 0.2233 & 0.9575 & 0.9994 \\ 
UConn91 & 17206 & 604867 & 8 & 70.31 & 1709 & 855 & 0.0041 & 0.2082 & 0.8754 & 0.9946 & 0.0994 & 0.9167 & 0.9980 \\ 
Auburn71 & 18448 & 973918 & 7 & 105.59 & 5160 & 2581 & 0.0058 & 0.3672 & 0.9528 & 0.9989 & 0.2798 & 0.9795 & 0.9998 \\ 
Maryland58 & 20829 & 744832 & 7 & 71.52 & 3784 & 1893 & 0.0035 & 0.2031 & 0.8631 & 0.9937 & 0.1817 & 0.9474 & 0.9989 \\ 
Wisconsin87 & 23831 & 835946 & 9 & 70.16 & 3484 & 1620 & 0.0030 & 0.1888 & 0.8607 & 0.9929 & 0.1462 & 0.9361 & 0.9985 \\ 
Indiana69 & 29732 & 1305757 & 8 & 87.84 & 1358 & 479 & 0.0030 & 0.1794 & 0.8624 & 0.9941 & 0.0457 & 0.8102 & 0.9960 \\ 
MSU24 & 32361 & 1118767 & 8 & 69.14 & 5267 & 2634 & 0.0022 & 0.1413 & 0.8222 & 0.9913 & 0.1628 & 0.9478 & 0.9989 \\ 

    \bottomrule
  \end{tabular}
  \label{tab:ball_size_both}
\end{table}

\end{landscape}

\newpage

\section{Proofs}
\label{app:proofs}



\subsection*{Proof of \Cref{Prp:Dep_u_max}}
\begin{proof}
In the first step of the algorithm (line 1--2), every node independently generates
a random number which can be executed in parallel. Therefore, the depth is $O(1)$
and the work is $O(n)$.
In the second step (line 3--4), each node $i$ computes the maximum of 
$\lvert B_1(i) \rvert = O(\degree{i})$ numbers, and parallel implementation of 
this max procedure requires $O(\log(\degree{i}))$ depth and $O(\degree{i})$
work~\cite{blelloch1996programming}. Moreover, note that the maximization task
at different nodes can also be executed in parallel, and thus in the second step, 
the total work is $\sum_i O(\degree{i}) = O(m)$ and the total depth is
$\max_i O(\log(\degree{i})) = O(\log (\dmax))$.
Combining both steps, the total work is $O(m)$ and total depth is $O(\log(\dmax))$.
\end{proof}

\subsection*{Proof of \Cref{Thm:prob_LB}}
\begin{proof}
We first show that, the probability that node $i$ is in the interior of a cluster
is lower bounded by $1 / \lvert B_2(i)\rvert$, \ie, 
\[
\prob{\forall j \in B_1(i), C_j = C_i \mid \mathcal G} \geq 1/\lvert B_2(i)\rvert.
\]
To this end, we consider a sufficient condition of this event, 
for the 3-net clustering and 1-hop-max clustering separately.
With 3-net clustering, in the first step when we generate a random ordering
of all nodes, if node $i$ ranks first among $B_2(i)$, 
then it must be outside the 2-hop neighborhood 
of every node ahead of itself in the ordering. Therefore, it is left unmarked  
and will be added as a seed, and thus its neighbors must belongs to 
the same cluster as $i$. The probability of this situation, \ie, node $i$
is ranked first amongst $B_2(i)$, is $1 / \lvert B_2(i)\rvert$ since we are
generating the orderings uniformly.
When the 1-hop-max clustering algorithm is used, in the first step where
every node independently generates a random number, 
if node $i$ generates the largest number amongst $B_2(i)$, then we have
$C_j = X_i$ for every $j \in B_1(i)$. This scenario happens with probability
$1 / \lvert B_2(i)\rvert$ since the random numbers generated at each node is \iid

Now we derive the results in the theorem. Conditioning on the event that
node $i$ is in the interior of a cluster, it is full-neighborhood exposed to 
the treatment condition if this cluster is assigned into the treatment group,
which happens with probability $p$. Therefore, by combining the result
in the previous paragraph, we have
\[
\prob{E_i ^\bOne \mid \mathcal G} \geq
\prob{\forall j \in B_1(i), C_j = C_i \mid \mathcal G} \cdot p 
\geq \frac{p}{\lvert B_2(i)\rvert} \geq \frac{p}{(1+\dmax)\kappa}.
\]
With the same reasoning, it can be easily verified that
$\prob{E_i ^{\bZero} \mid \mathcal P} \geq \frac{1-p}{\lvert B_2(i)\rvert} \geq \frac{1-p}{(1+\dmax)\kappa}$,
and this proof applies to both independent and complete randomization scenario.
\end{proof}

\subsection*{Proof of \Cref{Thm:prob_LB_improved}}
\begin{proof}
By symmetry, we only need to consider treatment (control is analogous).
Suppose node $i$ generates the $k$-th largest value in $B_2(i)$ for some $1 \leq k \leq \lvert B_2(i)\rvert$.
If $k \leq \degree i$, then with the generated clustering $\bc$, we have
$\prob{E_i ^\bOne \mid \bc} \geq p^{k}$ since node $i$ is adjacent to
at most $k$ clusters. If $k \geq 1 + \degree i$, then we have
$\prob{E_i ^\bOne \mid \bc} \geq p^{1 + \degree i}$, the trivial
lower bound when every node in $B_1(i)$ is assigned to a different cluster.
By combining the two scenarios, we have
\begin{eqnarray*}
\prob{E_i ^\bOne\mid \mathcal P} 
&=& \expect{ \prob{E_i ^\bOne \mid \bC} \mid \mathcal P}
\\
&\geq& \sum_{k=1} ^{\degree{i}} \frac{p^{k}}{\lvert B_2(i) \rvert}   
+\sum_{k=1+\degree i} ^{\lvert B_2(i)\rvert} \frac{p^{1+\degree i}}{\lvert B_2(i) \rvert} 
\\
&=& \frac{p}{\lvert B_2(i)\rvert} \cdot \left[\frac{1 - p^{\degree i}}{1 - p} +p^{\degree i} (\lvert B_2(i)\rvert - \degree i) \right],
\end{eqnarray*}
where we use the fact that the probability of node $i$ generating the $k$-th
largest number in $B_2(i)$ in the unweighted 1-hop-max clustering algorithm
is $1 / \lvert B_2(i)\rvert$.  As a final step, as $\lvert B_2(i) \rvert - \degree{i} \geq 1/(1-p)$
we conclude $\prob{E_i ^\bOne\mid \mathcal P}  \geq \frac{p}{\lvert B_2(i)\rvert} \cdot \frac1{1-p}$.
\end{proof}


\subsection*{Proof of \Cref{Thm:Hard-epsnet}}
\begin{proof}
Here we provide a proof for the specific case $r = 3$, corresponding to the
$3$-net clustering algorithm used in the original analysis of the graph cluster randomization scheme. 
The proof for other $r$ can be constructed analogously.

We present a polynomial-time reduction to network exposure probability computation 
from the minimum maximal distance-3 independent set (MD3IS) problem, 
which is known to be NP-complete~\cite{eto2014distance}.
This problem is as follows.

\begin{itemize}
\item Minimum Maximal Distance-3 Independent Set problem (decision version): Given a graph $G = (V, E)$ and an integer $K \leq \lvert V \rvert$, 
determine whether there is a maximal distance-3 independent set of size no greater than $K$, i.e., a subset of nodes  
$V_s \subseteq V$ such that
\begin{enumerate}   [(a)]
\item for any pair of nodes $u, v \in V_s$, their graph distance $\dist(u,v)\geq 3$ (\ie, a distance-3 independent set);
\item $V_s$ is not a subset of any other distance-3 independent set (\ie, maximal);
\item $\lvert V_s \rvert \leq K$.
\end{enumerate}
\end{itemize}

For any instance of the MD3IS problem with input $G = (V, E)$ and $K$, 
we construct the following instance of network exposure probability computation problem 
on $\tilde G = (\tilde V, \tilde E)$ around a node $i_0$, and $p = {1}/{(2\lvert V \rvert+1)!}$. We construct $\tilde G$ as follows, where $\sqcup$ is indicates a multi-set union. First, we make two copes if $V$ and add $i_0$ as an additional node. Let $E_1$ be edges connecting the corresponding nodes in $V$ and the copy $V^\prime$. Let $V$ have edges $E$ from the original graph, while let $V^\prime$ be a clique. Lastly, connect every node in $V^\prime$ to $i_0$. More formally:
\begin{itemize}
\item $\tilde V = V \sqcup V^\prime \sqcup \{i_0\}$ where there is a bijection $\phi$ between $V$ and $V^\prime$ (and consequently $\lvert V \rvert = \lvert V^\prime \rvert$).
\item $\tilde E = E \sqcup E_1 \sqcup E_2 \sqcup E_3$, where
\begin{itemize}
\item $E_1 = \{(u, \phi(u)) \mid u \in V\}$, \ie, connecting the every pair of corresponding nodes in $V$ and $V^\prime$.
\item $E_2 = \{(u^\prime, v^\prime) \mid u^\prime, v^\prime \in V^\prime, u^\prime \neq v^\prime\}$, \ie, connecting every pair of nodes in $V^\prime$.
\item $E_3 = \{(u^\prime, i_0)\}$, \ie, connecting node $i_0$ with every node in $V^\prime$.
\end{itemize}
\end{itemize}

Before connecting this exposure probability computation problem with the original MD3IS instance, we first present several
properties of the maximal distance-3 independent sets of $\tilde G$. 
The proofs are found at the end of this section.

\begin{lemma}   \label{Lem:11correspond}
For any node subset in the original graph $V_s \subset V$, it is a maximal distance-3 independent set of $\tilde G$ if and only if it is a maximal distance-3 independent set of $G$.
\end{lemma}

\begin{lemma}   \label{Lem:trivial_maximal_IS}
Any maximal distance-3 independent set of $\tilde G$, unless it is also a maximal distance-3 independent set of $G$, contains only a single-node $u^\prime \in V^\prime \sqcup \{i_0\}$.
\end{lemma}

\Cref{Lem:trivial_maximal_IS} illustrates the two types of maximal distance-3 independent set of $\tilde G$. With each of the types as the seed set in 3-net clustering, the following lemma states the conditional exposure probability of node $i_0$.
\begin{lemma}   \label{Lem:conditional_expo_prob}
For a random sample of 3-net clustering $c$ on $\tilde G$, let $V_s$ be the seed set, then
\begin{itemize}
\item with probability $\frac{\lvert V \rvert + 1}{2\lvert V \rvert+1}$, $V_s = \{u^\prime\}$ for some $u^\prime \in V^\prime \sqcup \{i_0\}$, and $\prob{E_{i_0} ^\bOne \mid c} = p$;
\item with probability $\frac{\lvert V \rvert}{2\lvert V \rvert+1}$, $V_s$ is a maximal distance-3 independent set of $G$, and $\prob{E_{i_0} ^\bOne \mid \bc} = p^{\lvert V_s \rvert}$.
\end{itemize}
\end{lemma}

Now we have the following key result connecting the exposure probability value to the MD3IS problem.
\begin{lemma}   \label{Lem:exposure_prob}
The exposure probability value of node $i_0$ under randomized 3-net clustering corresponds to the MD3IS problem instance as the following:
\begin{enumerate} [(1)]
\item If there exists a maximal distance-3 independent set of size $\leq K$, then 
$\prob{E_{i_0} ^\bOne} \geq \frac{n+1}{2n+1}\cdot p + p^{K+1}$.
\item If every maximal distance-3 independent set is of size $\geq K+1$, then 
$\prob{E_{i_0} ^\bOne} < \frac{n+1}{2n+1}\cdot p + \frac{1}{2} \cdot p^{K+1}$.
\end{enumerate}
\end{lemma}

Combining the results above, we show that exact computation of the exposure probability
solves the MD3IS instance. Suppose there is a polynomial algorithm such that,
for any graph $\tilde G$, node $i_0$, treatment probability $p$, and any precision $\epsilon > 0$, 
it outputs the treatment exposure probability of precision $\epsilon$ in time 
$poly(\sizeof(\tilde G), \sizeof(p), \log_2(1/\epsilon))$. 
Choosing $\epsilon = {p^{K+1}}/{4}$, we compare the output probability $P^{out}$ with $P^* = \frac{n+1}{2n+1}\cdot p +\frac{3}{4} p^{K+1}$.
If $P^{out} \geq P^*$, due to the precision $\epsilon$, we have 
\[
\prob{E_{i_0}^\bOne} \geq P^* - \epsilon = \frac{n+1}{2n+1}\cdot p + \frac{1}{2} \cdot p^{K+1}.
\]
Now according to scenario (2) in \Cref{Lem:exposure_prob},
the MD3IS instance must have a maximal distance-3 independent set of size $\leq K$.
If $P^{out} < P^*$, again due to the precision $\epsilon$, we have 
\[
\prob{E_{i_0}^\bOne} < P^* + \epsilon = \frac{n+1}{2n+1}\cdot p + p^{K+1},
\]
and thus the MD3IS instance cannot have a maximal distance-3 independent set of size $\leq K$ according to scenario (1) in \Cref{Lem:exposure_prob}.

Given this reduction, we must show that the reduction from the MD3IS is polynomial. 
The size of the MD3IS problem is $\sizeof(G) = \lvert V \rvert  + \lvert E \rvert$. 
For the constructed graph $\tilde G$, we have $\lvert \tilde V \rvert = 2 \lvert V \rvert + 1$,
and $\lvert \tilde E \rvert = \lvert E \rvert + \lvert V \rvert ^2 + 2 \lvert V \rvert$,
and thus constructing $\tilde G$ takes polynomial time and space.
In addition, we have
\[
\sizeof(p) = log_2(1/p) = \log_2((2\lvert V \rvert+1)!) = O(\lvert V \rvert \log(\lvert V \rvert)),
\]
and the log value of required precision
$\log(1 / \epsilon) = (K+1) \log_2(1/p) = O(K\lvert V \rvert \log(\lvert V \rvert))$
is also polynomial in the size of the MD3IS problem.
In summary, the reduction is a polynomial reduction.

\end{proof}

Before presenting the proof of \Cref{Lem:11correspond}, we first  
give an auxiliary result. 
\begin{lemma}   \label{Lem:Distance3Equivalent}
For any distinct nodes $u, v \in V$, we have
\begin{equation}   \label{Eq:Distance3Equivalent}
\dist[\tilde G](u, v) \geq 3 \qquad \Longleftrightarrow \qquad \dist[G](u, v)  \geq 3.
\end{equation}
\end{lemma}
\begin{proof}
``$\Longrightarrow$". Since all the edges in $G$ are preserved in $\tilde G$, we have $\dist[\tilde G](u, v) \leq \dist[G](u, v)$. Consequently, if $\dist[\tilde G](u, v) \geq 3$, we must have $\dist[G](u, v)  \geq 3$.

``$\Longleftarrow$". Note that node $u \in V$ is not directly connected to any node in
$V^\prime \sqcup \{i_0\}$ other than $\phi(u)$, and thus for
any path connecting nodes $u$ and $v$ through nodes in $V^\prime \sqcup \{i_0\}$, the path length is at least 3. Therefore, if $\dist[\tilde G](u, v) < 3$, the shortest path can only consists of nodes and edges in the original graph $G$, and thus $\dist[G](u, v) < 3$.
\end{proof}

\begin{proof}[Proof (\Cref{Lem:11correspond})]
Note that a corollary of \Cref{Lem:Distance3Equivalent} is the following:
\begin{itemize}
\item $V_s$ is a distance-3 independent set of $\tilde G$ if and only if it is also a distance-3 independent set of $G$.
\end{itemize}
which is the lemma result regarding only the independent set condition.

We first show the sufficiency in \Cref{Lem:11correspond}. If $V_s$ is a maximal distance-3 independent set of $G$, then according to the argument above, $V_s$ is also a distance-3 independent set of $\tilde G$, and we now show it is maximal, \ie, introducing any other node $w$ into $V_s$ would break the distance-3 independent set condition.
There are two scenarios: $w \in V$ and $w \in V^\prime \sqcup \{i_0\}$.
First, for any node $w \in V, w \notin V_s$, since $V_s$ is a maximal distance-3 independent set in $G$, there exists node $u \in V_s$ such that $\dist[G](u, w) \leq 2$, and thus $\dist[\tilde G](u, w) \leq 2$ according to \Cref{Lem:Distance3Equivalent}. 
Therefore, introducing node $v$ into $V_s$ makes $V_s$ no longer a distance-3 independent set of $\tilde G$.
Second, note that $V_s$ must be nonempty and denote $u$ as an arbitrary node therein. For any node $w \in V^\prime \sqcup \{i_0\}$, due to the length-2 path $(w, \phi(u), u)$, we have $\dist[\tilde G](u, w) \leq 2$ and thus one can not include node $v$ to $V_s$ while maintaining that it is a distance-3 independent set.

Next we show the necessity in \Cref{Lem:11correspond}. If $V_s \subset V$ is a maximal distance-3 independent set of $\tilde G$, again $V_s$ is a distance-3 independent set of $G$ and we only need to show its maximality.
Suppose it is not maximal, and there is another node $w \in V$, $w \notin V_s$ such that $V_s \sqcup\{w\}$ is a distance-3 independent set of $\tilde G$, then according to the argument above, $V_s \sqcup\{w\}$ is also a distance-3 independent set of $\tilde G$. This means that $V_s$ is not maximal in $\tilde G$, which creates a contradiction.
\end{proof}

\begin{proof}[Proof (\Cref{Lem:trivial_maximal_IS})]

Suppose $V_s$ is a maximal distance-3 independent set of $\tilde G$ but not of $G$. 
Note that $V_s$ must contain a node $u^\prime \in V^\prime \sqcup \{i_0\}$ because otherwise
with $V_s \subset V$, according to \Cref{Lem:11correspond}, $V_s$ must also be a maximal distance-3 independent set of $G$.
Now due to the fact that $V^\prime \sqcup \{i_0\}$ induces a complete graph, and thus any other node in $V^\prime \sqcup \{i_0\}$ cannot be included in $V_s$. Moreover, for any node $u \in V$, we have $\dist[\tilde G](u, u^\prime) \leq 2$ due to the path $(w, \phi(u), u)$, and thus $V_s$ cannot contain any node in $V$. In summary, $u^\prime$ is the only node in $V_s$.
\end{proof}

\begin{proof}[Proof (\Cref{Lem:conditional_expo_prob})]
In 3-net clustering, we use the randomized greedy algorithm to construct a maximal distance-3 independent set as the seed set.
With probability $\frac{\lvert V \rvert+1}{2\lvert V \rvert+1}$, the first randomly selected node is from $V^\prime \sqcup \{i_0\}$, and the seed set is a single-node set according to \Cref{Lem:trivial_maximal_IS}, and the conditional network exposure probability is $\prob{E_{i_0} ^\bOne \mid c} = p^1$.

Similarly, with probability $\frac{\lvert V \rvert}{2\lvert V \rvert+1}$, the first randomly selected node is from $V$,
and the seed set is a maximal distance-3 independent set of $G$. Now we show the conditional network exposure probability.
Note that in 3-net clustering, for any seed node $u \in V_s$, node $\phi(u) \in V^\prime$ 
belongs to the same cluster as $u$'s due to it being a directed neighbor of node $u$. 
Since $B_1(i_0) = V^\prime \sqcup \{i_0\}$, $B_1(i_0)$ has nonempty intersection
with all the $\lvert V_s \rvert$ clusters in this 3-net clustering, and thus 
$\prob{E_{i_0} ^\bOne \mid c} = p^{\lvert V_s \rvert}$. 
\end{proof}

\begin{proof}[Proof (\Cref{Lem:exposure_prob})]
For the first result, if the MD3IS problem has a maximal distance-3 independent set of size $\leq K$, then this set will be selected as the seed set
in 3-net clustering with probability no less than $\frac{1}{(2n+1)!}$. Recall that $p = {1}/{(2\lvert V \rvert+1)!}$ by construction. Thus
\[
\prob{E_{i_0} ^\bOne} \geq 
\frac{n+1}{2n+1}\cdot p + \frac{1}{(2n+1)!} \cdot p^{K}
= \frac{n+1}{2n+1}\cdot p + p^{K+1}
\]
where the first term comes from the case when a single-node maximal distance-3 independent set is used as seeds.

For the second result, if every maximal distance-3 independent set of size $\geq K+1$ in the MD3IS problem, then any class-2 maximal independent set of $\tilde G$ is of size $\geq K+1$, and thus
\[
\prob{E_{i_0} ^\bOne} \leq 
\frac{n+1}{2n+1}\cdot p + \frac{n}{2n+1} \cdot p^{K+1}
< \frac{n+1}{2n+1}\cdot p + \frac{1}{2} \cdot p^{K+1}.
\]
\end{proof}

\subsection*{Proof of \Cref{Thm:prob_MC_var}}
\begin{proof}
Unless otherwise stated, all the expectation and variance in this proof are taken 
conditioned on the random clustering distribution $\mathcal P$ in use;
thus to simplify the notation, we sometimes discard the conditional notation
in the expectation and variance symbol.

Let $\bC$ be a random clustering generated from 3-net or 1-hop-max algorithm,
we first show that
\begin{equation*}
\textstyle
\var[\bC \sim \mathcal P]{\frac{\prob{E_i ^{\bOne} \mid \bC} - \prob{E_i ^\bOne \mid \mathcal P}}{\prob{E_i ^\bOne \mid \mathcal P}}}
\leq \frac{\lvert B_2(i)\rvert}{p}.
\end{equation*}
Consider the Bernoulli random variable $\indic{E_i ^\bOne}$, and we have
\[
\var{\indic{E_i ^\bOne}} 
= \prob{E_i ^\bOne \mid \mathcal P} \cdot (1 - \prob{E_i ^\bOne \mid \mathcal P}) 
\leq \prob{E_i ^\bOne \mid \mathcal P}.
\]
Moreover, note that $\expect{\indic{E_i ^\bOne} \mid \bC} = \prob{E_i ^\bOne \mid \bC}$,
and consequently due to the law of total variance, we have
\[
\var[\bC \sim \mathcal P]{\prob{E_i ^\bOne \mid \bC}} = \var[\bC \sim \mathcal P]{\expect{\indic{E_i ^\bOne} \mid \bC}}
\leq \var{{\indic{E_i ^\bOne}} } \leq \prob{E_i ^\bOne \mid \mathcal P},
\]
and thus
\begin{equation}   \label{Eq:var_expo_prob_est}
\textstyle
\var[\bC \sim \mathcal P]{\frac{\prob{E_i ^{1} \mid \bC} - \prob{E_i ^\bOne\mid \mathcal P}}{\prob{E_i ^\bOne\mid \mathcal P}}}
\leq \frac{\prob{E_i ^\bOne\mid \mathcal P}}{\prob{E_i ^\bOne\mid \mathcal P}^2} 
= \frac{1}{\prob{E_i ^\bOne\mid \mathcal P}} 
\leq \frac{\lvert B_2(i)\rvert}{p},
\end{equation}
where the second inequality is due to \Cref{Thm:prob_LB}.

Now to prove the inequality in \Cref{Thm:prob_MC_var}, we note that
\[ \textstyle
\frac{\hat{\mathbb{P}}[E_i ^\bOne \mid \mathcal P] - \prob{E_i ^\bOne \mid \mathcal P}}{\prob{E_i ^\bOne \mid \mathcal P}}
=
\frac{ \left( \frac{1}{K} \sum_{k=1}^K \prob{E_i ^{\bOne} \mid \bc^{(k)}} \right) - \prob{E_i ^\bOne \mid \mathcal P}}{\prob{E_i ^\bOne \mid \mathcal P}}
= \frac{1}{K} \sum_{k=1}^K \left[ \frac{\prob{E_i ^{\bOne} \mid \bc^{(k)}} - \prob{E_i ^\bOne \mid \mathcal P}}{\prob{E_i ^\bOne \mid \mathcal P}} \right],
\]
and thus
\[ \textstyle
\var{\frac{\hat{\mathbb{P}}[E_i ^\bOne \mid \mathcal P] - \prob{E_i ^\bOne \mid \mathcal P}}{\prob{E_i ^\bOne \mid \mathcal P}} ~\middle|~ \mathcal P}
= \frac{1}{K} \var[\bC \sim \mathcal P]{  \frac{\prob{E_i ^{\bOne} \mid \bC} - \prob{E_i ^\bOne \mid \mathcal P}}{\prob{E_i ^\bOne \mid \mathcal P}}  }  \leq \frac{\lvert B_2(i)\rvert}{Kp}.
\]
Now note that the expected value of 
$\frac{\hat{\mathbb{P}}[E_i ^\bOne \mid \mathcal P] - \prob{E_i ^\bOne \mid \mathcal P}}{\prob{E_i ^\bOne \mid \mathcal P}}$
is zero due to the fact that $\hat{\mathbb{P}}[E_i ^\bOne \mid \mathcal P]$ is 
an unbiased estimator of $\prob{E_i ^\bOne \mid \mathcal P}$. Thus the squared error is the same as the variance.
\end{proof}

\subsection*{Proof of \Cref{Thm:var_restricted_growth_general}}
We first present the following result on the joint exposure probability
of a pair of nodes.
\begin{lemma}   \label{Lem:1_hop_max-joint_prob}
For the 1-hop-max random clustering algorithm, if $\dist(i, j) > 4$ for a pair
of nodes $i$ and $j$, then for $\bz = \bOne$ or $\bz = \bZero$ and
\begin{itemize}
\item independent randomization, we have 
$\prob{E_i ^\bz \cap E_j ^\bz \mid \mathcal P} = \prob{E_i ^\bz \mid \mathcal P} \cdot \prob{E_j ^\bz \mid \mathcal P}$;
\item complete randomization, we have
$\prob{E_i ^\bz \cap E_j ^\bz \mid \mathcal P} \leq \prob{E_i ^\bz \mid \mathcal P} \cdot \prob{E_j ^\bz \mid \mathcal P}$.
\end{itemize}
\end{lemma}
\begin{proof}
We first show that nodes $i$ and $j$ satisfying the above requirements 
can not be adjacent to the same cluster in
any clustering generated from 1-hop-max. If otherwise, then there exists nodes
$i^\prime \in B_1(i)$ and $j^\prime \in B_1(j)$ such that $C_{i^\prime} = C_{j^\prime}$,
and a node $k$ such that $i^\prime, j^\prime \in B_1(k)$ with $C_{i^\prime} = C_{j^\prime} = X_k$ (recall that in the 1-hop-max algorithm, $X_1,\cdots,X_n$ are the $\mathcal U(0, 1)$ samples and also the signifiers of the clusters),
we have $\dist(i, j) \leq 4$ due to the path $[i, i^\prime, k, j^\prime, j]$,
contradictory to our assumption that $\dist(i, j) > 4$.

Now we prove the results in the lemma.
Due to symmetry, it suffices to just prove for the case of $\bz = \bOne$, 
and we first analyze the independent randomization scenario. 
For any clustering $\bC$ generated from 1-hop-max,
since $i$ and $j$ are not adjacent to a same cluster in $\bC$, we have
$\prob{E_i ^\bOne \cap E_j ^\bOne  \mid \bC}
= \prob{E_i ^\bOne \mid \bC} \cdot \prob{E_i ^\bOne \mid \bC}$.
Moreover, since $\dist(i,j) > 4$, we have $B_2(i) \cap B_2(j) = \emptyset$,
and according to \Cref{Lem:1_hop_max-local_dep}, we have
$\bC_{B_1(i)}$ and $\bC_{B_1(j)}$ independent.
Combining both results, and note
the fact that $\prob{E_i ^\bOne \mid \bC} = \prob{E_i ^\bOne \mid \bC_{B_1(i)}}$
and $\prob{E_j ^\bOne \mid \bC} = \prob{E_j ^\bOne \mid \bC_{B_1(j)}}$, we have
\begin{eqnarray*}
\prob{E_i ^\bOne \cap E_j ^\bOne \mid \mathcal P} 
  &=& \expect{\prob{E_i ^\bOne \cap E_j ^\bOne \mid \bC}}
\\&=& \expect{\prob{E_i ^\bOne \mid \bC} \cdot \prob{E_i ^\bOne \mid \bC}}
\\&=& \expect{\prob{E_i ^\bOne \mid \bC_{B_1(i)}} \cdot \prob{E_i ^\bOne \mid \bC_{B_1(j)}}}
\\&=& \expect{\prob{E_i ^\bOne \mid \bC_{B_1(i)}}} \cdot \expect{\prob{E_i ^\bOne \mid \bC_{B_1(j)}}}
\\&=& \prob{E_i ^\bOne \mid \mathcal P} \cdot \prob{E_j ^\bOne \mid \mathcal P},
\end{eqnarray*}
where the last but one equality is due to $\bC_{B_1(i)}$ and $\bC_{B_1(j)}$
being independent.

We now prove the result for the complete randomization scenario. The proof
is almost identical to that of independent randomization, except for the fact that
$\prob{E_i ^\bOne \cap E_j ^\bOne  \mid \bC}
\leq \prob{E_i ^\bOne \mid \bC} \cdot \prob{E_i ^\bOne \mid \bC}$
which is an equality in the independent randomization. Here the joint probability
might be smaller due to the scenario if one of the clusters adjacent to node $i$
is paired with one adjacent to node $j$, and thus $\prob{E_i ^\bOne \cap E_j ^\bOne  \mid \bC} = 0$.
Formally, we have
\[
\prob{E_i ^\bOne \cap E_j ^\bOne \mid \mathcal P} 
= \expect{\prob{E_i ^\bOne \cap E_j ^\bOne \mid \bC}}
\leq \expect{\prob{E_i ^\bOne \mid \bC} \cdot \prob{E_i ^\bOne \mid \bC}}
= \prob{E_i ^\bOne \mid \mathcal P} \cdot \prob{E_j ^\bOne \mid \mathcal P}
\]
where the last equality is verified in the proof for independent randomization.
\end{proof}

With this lemma in hand, we now prove \Cref{Thm:var_restricted_growth_general}.

\begin{proof}[Proof (\Cref{Thm:var_restricted_growth_general})]
According to \Cref{Lem:1_hop_max-joint_prob}, for any pair of nodes
$i$ and $j$ such that $j \notin B_4(i)$, we have
$\frac{\prob{E_i ^\bOne \cap E_j ^\bOne}}{\prob{E_i ^\bOne} \prob{E_j ^\bOne}}- 1 \leq 0$
for both independent and complete randomization.
Now the variance of the mean outcome estimator, as formulated in
\Cref{Eq:var-mean_outcome}, satisfies
\begin{eqnarray*}
\var{\hat \mu _{\mathcal P}(\bOne)}
 &=& \textstyle \frac{1}{n^2}  \sum_{i=1}^n \left[ \left( \frac{1}{\prob{E_i ^\bOne \mid \mathcal P}} - 1 \right)Y_i(\bOne)^2
+  \sum_{ j \neq i} \left( \frac{\prob{E_i ^\bOne \cap E_j ^\bOne \mid \mathcal P}}{\prob{E_i ^\bOne \mid \mathcal P} \prob{E_j ^\bOne \mid \mathcal P}}- 1 \right) Y_i(\bOne)Y_j(\bOne) \right]
\\&\leq& \textstyle \frac{1}{n^2}  \sum_{i=1}^n \left[ \left( \frac{1}{\prob{E_i ^\bOne \mid \mathcal P}} - 1 \right)Y_i(\bOne)^2
+  \sum_{j \in B_4(i), j \neq i} \left( \frac{\prob{E_i ^\bOne \cap E_j ^\bOne \mid \mathcal P}}{\prob{E_i ^\bOne \mid \mathcal P} \prob{E_j ^\bOne \mid \mathcal P}}- 1 \right) Y_i(\bOne)Y_j(\bOne) \right]
\\&\leq& \textstyle \frac{1}{n^2}  \sum_{i=1}^n \left[ \left( \frac{1}{\prob{E_i ^\bOne \mid \mathcal P}} - 1 \right)Y_i(\bOne)^2
+  \sum_{j \in B_4(i), j \neq i} \left( \frac{1}{\prob{E_i ^\bOne \mid \mathcal P}}- 1 \right) Y_i(\bOne)Y_j(\bOne) \right]
\\&\leq&  \textstyle \frac{\bar Y^2}{n^2}  \sum_{i=1}^n \left[ \left( \frac{1}{\prob{E_i ^\bOne \mid \mathcal P}} - 1 \right) \cdot \lvert B_4(i) \rvert  \right]
\\&\leq& \textstyle \frac{\bar Y^2}{n^2} \sum_{i=1}^n \frac{\lvert B_4(i) \rvert}{\prob{E_i ^\bOne \mid \mathcal P}},
\end{eqnarray*}
where the second inequality is due to $\prob{E_i ^\bOne \cap E_j ^\bOne \mid \mathcal P} \leq \prob{E_j ^\bOne \mid \mathcal P}$. 
\end{proof}

\subsection*{Proof of \Cref{Thm:var_GATE_restricted_growth}}
\begin{proof}
Analogous to \Cref{Thm:var_restricted_growth_unweighted}, it can be verified that
$\var{\hat \mu _{\mathcal P}(\bZero) } \leq \frac{1}{n} \cdot  \bar Y^2 (d_{\max}+1) ^2 \kappa^4 (1-p)^{-1}$. 
Now according to the variance formula in \Cref{Eq:var-GATE}, we have
\begin{eqnarray*}
\var{\hat \tau _{\mathcal P}} 
  &=& \var{\hat \mu _{\mathcal P}(\bOne )} + \var{\hat \mu _{\mathcal P}(\bZero)} - 2 \cdot \cov{\hat \mu _{\mathcal P}(\bOne), \hat \mu _{\mathcal P}(\bZero)}
\\&\leq& \var{\hat \mu _{\mathcal P}(\bOne )} + \var{\hat \mu _{\mathcal P}(\bZero)} + 2 \cdot \sqrt{\var{\hat \mu _{\mathcal P}(\bOne )} \cdot \var{\hat \mu _{\mathcal P}(\bZero)}}
\\&\leq& \var{\hat \mu _{\mathcal P}(\bOne )} + \var{\hat \mu _{\mathcal P}(\bZero)} + \var{\hat \mu _{\mathcal P}(\bOne )} + \var{\hat \mu _{\mathcal P}(\bZero)}
\\&\leq& \frac{2}{n} \cdot  \bar Y^2 (d_{\max}+1) ^2 \kappa^4 (p^{- 1} + (1-p)^{-1}).
\end{eqnarray*}
where the second inequality is due to the mean inequality, and the last inequality
is due to the variance upper bound of $\hat \mu _{\mathcal P}(\bOne)$ and $\hat \mu _{\mathcal P}(\bZero)$ respectively.
\end{proof}


\subsection*{Proof of \Cref{Prp:beta_dist_prop}}
\begin{proof}
For result {\it(a)}, Since the PDF of $\beta(w, 1)$ distribution is $f(x) = wx^{w-1}$, we have
\[
\prob{X_j < X_i} = \int_{0} ^1 f(x_i) \prob{X_j \leq x_i}  dx_i = \int_{0} ^1 w_i x^{w_i - 1} \cdot x_i ^{w_j} dx_i = \frac{w_i}{w_i + w_j}.
\]
For result {\it (b)}, we have
\[
\prob{\max\{X_i, X_j\} \leq x} = \prob{X_i \leq x}\prob{X_j \leq x} = x^{w_i + w_j},
\]
which is the CDF of the $\beta(w_i + w_j, 1)$ distribution.
\end{proof}

\subsection*{Proof of \Cref{Thm:prob_LB_weighted}}

We first prove the following useful lemma.

\begin{lemma}   \label{Lem:ProbLargestNumber}
If every node generates $X_i \sim \beta(w_i, 1)$ independent, then we have 
$\prob{X_i \geq X_j, \forall j \in B_2(i)} = w_i / \left( \sum_{j \in B_2(i)} w_j\right)$
\end{lemma}
\begin{proof}
According to the second result in \Cref{Prp:beta_dist_prop}, the distribution
of $\bar X \triangleq \max([X_j]_{j \in B_2(i), j \neq i})$ is $\beta(S_i, 1)$ where
$S_i = \sum_{j \in B_2(i), j \neq i}w_j$. Moreover, we note that $\bar X$ and $X_i$
are independent, and thus $\prob{X_i \geq \bar X} = w_i / (S_i + w_i)$ according to
the first result in \Cref{Prp:beta_dist_prop}. This result is equivalent to the lemma statement.
\end{proof}

\begin{proof}[Proof (\Cref{Thm:prob_LB_weighted})]
Analogous to the proof of \Cref{Thm:prob_LB}, one only need to show that
the probability of node $i$ being in the interior of a cluster is lower bounded by
${w_i}/\left({\sum_{j \in B_2(i)} w_j} \right)$.

For 1-hop-max clustering, similarly a sufficient condition for node $i$
being in the interior of a clustering is if node $i$
generates the largest number in $B_2(i)$, \ie, $X_i \geq X_j$ for any
$j \in B_2(i)$. By \Cref{Lem:ProbLargestNumber} we have its probability being
${w_i}/\left({\sum_{j \in B_2(i)} w_j} \right)$, and which
is a lower bound on the probability of node $i$ being in the interior of a cluster.

For 3-net clustering, again we consider the scenario when node $i$ is ranked
first among nodes in $B_2(i)$, a sufficient condition of node $i$ being a seed node
and thus in the interior of a cluster. According to the procedure in \Cref{Alg:3_net-w},
this is equivalent to node $i$ generating the largest number in $B_2(i)$. 
According to \Cref{Lem:ProbLargestNumber} , that probability
is ${w_i}/\left({\sum_{j \in B_2(i)} w_j} \right)$.
\end{proof}

\subsection*{Proof of \Cref{Thm:prob_MC_var_weighted}}
The proof is almost the same as that of \Cref{Thm:prob_MC_var},
except that in \Cref{Eq:var_expo_prob_est} we apply the lower bound
result of \Cref{Thm:prob_LB_weighted}.

\subsection*{Proof of \Cref{Thm:UniformOptimal}}
\begin{proof}
Note that $j \in B_2(i)$ is equivalent as $\dist(i, j) \leq 2$, and thus
\begin{eqnarray*}
\bar H(\bw) 
&=& \frac{1}{p} \sum_{i=1} ^n \frac{\sum_{j \in B_2(i) }w_j}{w_i}
= \frac1 p \sum_{(i, j): \dist(i,j) \leq 2} \left( \frac{w_j}{w_i} + \frac{w_i}{w_j}\right) 
\\
&\geq& \frac 1 p \sum_{(i, j): \dist(i,j) \leq 2} 2\cdot \sqrt{\frac{w_j}{w_i} \cdot \frac{w_i}{w_j}}
= \frac 1 p \sum_i \lvert B_2(i) \rvert,
\end{eqnarray*}
where the inequality applies the AM--GM Inequality, which holds as equality
if all nodes have the same weight.
\end{proof}

\subsection*{Proof of \Cref{Thm:UnbiasedHajek}}
\begin{proof}
With a constant treatment effect $\tau$, we have
\[
\tilde \mu(\bOne) 
= \frac{ \sum_{i=1}^n {\indic{E_i ^\bOne}}/{\prob{E_i ^\bOne}}\cdot (Y_i(\bZero) + \tau)}{ \sum_{i=1}^n {\indic{E_i ^\bOne}}/{\prob{E_i ^\bOne}}}
= \tau + \frac{ \sum_{i=1}^n {\indic{E_i ^\bOne}}/{\prob{E_i ^\bOne}}\cdot Y_i(\bZero)}{ \sum_{i=1}^n {\indic{E_i ^\bOne}}/{\prob{E_i ^\bOne}}},
\]
and consequently
\[
\expect{\tilde \mu(\bOne) } -\tau
=  \expect{\frac{ \sum_{i=1}^n {\indic{E_i ^\bOne}}/{\prob{E_i ^\bOne}}\cdot Y_i(\bZero)}{ \sum_{i=1}^n {\indic{E_i ^\bOne}}/{\prob{E_i ^\bOne}}}}
= \expect{\frac{ \sum_{i=1}^n {\indic{E_i ^\bZero}}/{\prob{E_i ^\bZero}}\cdot Y_i(\bZero)}{ \sum_{i=1}^n {\indic{E_i ^\bZero}}/{\prob{E_i ^\bZero}}}}
= \expect{\tilde \mu(\bZero) }
\]
where the second equality is due to the symmetry of network exposure to treatment and control,
more specifically, the joint distribution of $\{ \indic{E_i ^\bOne} \}_{i=1}^n$ is the same as
that of $\{ \indic{E_i ^\bZero} \}_{i=1}^n$. 

Since $\expect{\tilde \mu(\bOne) } = \expect{\tilde \mu(\bZero) } + \tau$, then $\expect{\tilde \tau} =\expect{\tilde \mu(\bOne) } - \expect{\tilde \mu(\bZero) }= \tau$.
\end{proof}


\subsection*{Proof of \Cref{Thm:RingNetworkVariance}}

\begin{proof}

Note that in any oracle $k$-partition of a cycle, each cluster contains two nodes on the boundary.
Therefore, for any node $i$, the probability of being on the boundary of a random such partition
is $2k / n = o(1)$, and thus as $n \rightarrow \infty$,
it is almost surely between nodes in the same cluster and 
\begin{equation}   \label{Eq:ExpoProbRing}
\prob{E_i ^\bOne \mid \mathcal P} = \prob{E_i ^\bZero \mid \mathcal P} \rightarrow \frac{1}{2}.
\end{equation}

For any node pair $i$ and $j$, define their \emph{angle distance} as 
\[
\delta(\alpha_i, \alpha_j) \triangleq \min\{\lvert \alpha_i - \alpha_j \rvert, 2\pi - \lvert \alpha_i - \alpha_j \rvert\},
\]
a quantity in $[0, \pi]$ which is zero if and only if $i = j$.
Note that if $C_i = C_j$, \ie, if $i$ and $j$ belongs to a same cluster in a oracle $k$-partition, 
we must have $\delta(\alpha_i, \alpha_j) < \frac{2\pi}{k}$. On the contrary, if $\delta(\alpha_i, \alpha_j) < \frac{2\pi}{k}$,
the probability of them belonging to a same cluster in a random oracle $k$-partition is
$1 - \delta(\alpha_i, \alpha_j) / \frac{2\pi}{k}$. Therefore, for any node pair $i$ and $j$, we have
\begin{equation}   \label{Eq:ProbRingSameCluster}
\prob{C_i = C_j \mid \mathcal P} = \indic{ \delta(\alpha_i, \alpha_j) < 2\pi/k } \cdot \left[1 - \frac{k \delta(\alpha_i, \alpha_j)}{2\pi} \right].
\end{equation}

\xhdr{Variance under independent randomization}
We start with computing the joint exposure probabilities $\prob{E_i ^\bOne \cap E_j ^\bOne \mid \mathcal P}$ and 
$\prob{E_i ^\bOne \cap E_j ^\bZero \mid \mathcal P}$. Note that in the limit of $n \rightarrow \infty$, 
the probability of either node being on a boundary vanishes, and thus we have
\[
\prob{E_i ^\bOne \cap E_j ^\bOne \mid \bC} \rightarrow \left\{
\begin{array}{ll}
1/2  &  \text{if $C_i = C_j$}   \\
1/4  &  \text{if $C_i \neq C_j$}   
\end{array}
\right.,
\qquad
\prob{E_i ^\bOne \cap E_j ^\bZero \mid \bC} \rightarrow \left\{
\begin{array}{ll}
0  &  \text{if $C_i = C_j$}   \\
1/4  &  \text{if $C_i \neq C_j$}   
\end{array}
\right..
\]
By combining with \Cref{Eq:ProbRingSameCluster}, we have
\begin{equation}   \label{Eq:Joint11ExpoProbRing}
\prob{E_i ^\bOne \cap E_j ^\bOne \mid \mathcal P} \rightarrow \frac{1}{2} \cdot \prob{C_i = C_j \mid \mathcal P} + \frac{1}{4} \cdot \prob{C_i \neq C_j \mid \mathcal P}
= \left\{
\begin{array}{ll}
\frac{1}{2} - \frac{k \delta(\alpha_i, \alpha_j)}{8\pi}  &  \text{if $\delta(\alpha_i, \alpha_j) < \frac{2\pi}{k}$}   \\
\\
\frac{1}{4}  &  \text{if $\delta(\alpha_i, \alpha_j) \geq \frac{2\pi}{k}$}   
\end{array}
\right.
\end{equation}
and
\begin{equation} 
\prob{E_i ^\bOne \cap E_j ^\bZero \mid \mathcal P} \rightarrow \frac{1}{4} \cdot \prob{C_i \neq C_j \mid \mathcal P}
= \left\{
\begin{array}{ll}
\frac{k \delta(\alpha_i, \alpha_j)}{8\pi}  &  \text{if $\delta(\alpha_i, \alpha_j) < \frac{2\pi}{k}$}   \\
\\
\frac{1}{4}  &  \text{if $\delta(\alpha_i, \alpha_j) \geq \frac{2\pi}{k}$}   
\end{array}
\right.
.
\end{equation}

Now we compute the variance of the mean outcome HT estimators given these probabilities. Note that the variance, as given in \Cref{Eq:var-mean_outcome}, is equivalent to
\[
\var{\hat \mu(\bz)} = \frac{1}{n^2} \sum_{i=1}^n  \sum_{j = 1} ^n \left(\frac{\prob{E_i ^\bz \cap E_j ^\bz}}{\prob{E_i ^\bz} \prob{E_j ^\bz}}- 1 \right)Y_i(\bz)Y_j(\bz) 
\]
due to the fact that $\prob{E_i ^\bz \cap E_i ^\bz} = \prob{E_i ^\bz}$. Therefore,
\begin{eqnarray*}
 \var{\hat \mu(\bOne)} 
&=&\textstyle \frac{1}{n^2} \sum_{i=1}^n \sum_{j=1}^n \left(\frac{\prob{E_i ^\bOne \cap E_j ^\bOne}}{\prob{E_i ^\bOne}\prob{E_j ^\bOne}} - 1\right) Y_i(E_i ^\bOne)Y_j(E_j ^\bOne)
\\&=&\textstyle \frac{1}{n^2} \sum_{i=1}^n \sum_{j: \delta(\alpha_i, \alpha_j) < \frac{2\pi}{k}}\left(1 - \frac{k \delta(\alpha_i, \alpha_j)}{2\pi} + o(1)\right) (a + b \sin \alpha_i + \tau)(a + b \sin \alpha_j + \tau)
\\& & \textstyle + \frac{1}{n^2} \sum_{i=1}^n \sum_{j: \delta(\alpha_i, \alpha_j) \geq \frac{2\pi}{k}}\left( 4 \cdot \frac{1}{4} - 1 + o(1)\right) (a + b \sin \alpha_i + \tau)(a + b \sin \alpha_j + \tau)
\\&=&\textstyle \frac{1}{n^2} \sum_{i=1}^n \sum_{j: \delta(\alpha_i, \alpha_j) < \frac{2\pi}{k}}\left(1 - \frac{k \delta(\alpha_i, \alpha_j)}{2\pi}\right) (a + b \sin \alpha_i + \tau)(a + b \sin \alpha_j + \tau)
\\& & \textstyle + \frac{1}{n^2} \sum_{i=1}^n \sum_{j: \delta(\alpha_i, \alpha_j) \geq \frac{2\pi}{k}}\left( 4 \cdot \frac{1}{4} - 1\right) (a + b \sin \alpha_i + \tau)(a + b \sin \alpha_j + \tau) + o(1),
\end{eqnarray*}
where the third equality is due to the average of $n^2$ $o(1)$ terms being $o(1)$.
We then take a limit corresponding to Riemann integration and obtain:
\begin{eqnarray*}
& & \var{\hat \mu(\bOne)} 
\\&\to& \frac{1}{4\pi^2} \int_{0} ^{2\pi}(a + b \sin \alpha_i + \tau) \cdot \left(\int_{\alpha_i - \frac{2\pi}{k}} ^{\alpha_i + \frac{2\pi}{k}} \left(1 - \frac{k \lvert \alpha_i - \alpha_j \rvert}{2\pi}\right) (a + b \sin \alpha_j + \tau) d\alpha_j \right) d\alpha_i
\\&=& \frac{1}{4\pi^2} \int_{ 0} ^{2\pi}(a + b \sin \alpha_i + \tau) \cdot \left(\int_{0} ^{\frac{2\pi}{k}} \left(1 - \frac{k \delta}{2\pi}\right) (2a + 2\tau + b \sin(\alpha_i + \delta) + b \sin(\alpha_i - \delta)) d\delta \right) d\alpha_i
\\&=& \frac{1}{4\pi^2} \int_{ 0} ^{2\pi}(a + b \sin \alpha_i + \tau) \cdot \left(\int_{0} ^{\frac{2\pi}{k}} \left(1 - \frac{k \delta}{2\pi}\right) (2a + 2\tau + 2b \sin\alpha_i \cos\delta) d\delta \right) d\alpha_i
\\&=& \frac{1}{4\pi^2} \int_{ 0} ^{2\pi}(a + \tau + b \sin \alpha_i) \cdot \left((a + \tau) \cdot \frac{2\pi}{k} + \frac{bk(1-\cos(2\pi / k))}{\pi}\sin \alpha_i \right) d\alpha_i
\\&=& \frac{(a + \tau)^2}{k} + \frac{b^2 k (1 - \cos(2\pi / k))}{4 \pi^2}.
\end{eqnarray*}
Analogously, it can be shown that
\begin{eqnarray*}
\var{\hat \mu(\bZero)} &\rightarrow&  \frac{a^2}{k} + \frac{b^2 k (1 - \cos(2\pi / k))}{4 \pi^2}, \\
\cov{\hat \mu(\bOne), \hat \mu(\bZero)} &\rightarrow& \frac{a(a+\tau)}{k} +  \frac{b^2 k (1 - \cos(2\pi / k))}{4 \pi^2}.
\end{eqnarray*}
Consequently
\[
\var{\hat \tau} = \var{\hat \mu(\bOne)} + \var{\hat \mu(\bZero)} - 2 \cov{\hat \mu(\bOne), \hat \mu(\bZero)} \rightarrow  \frac{(2a+\tau)^2}{k} + \frac{b^2 k}{\pi^2} (1 - \cos({2\pi}/{k})).
\]

\xhdr{Variance under complete randomization}
We also first compute the joint exposure probabilities $\prob{E_i ^\bOne \cap E_j ^\bOne}$. 
If $C_i = C_j$ we have $\prob{E_i ^\bOne \cap E_j ^\bOne \mid C} = 1/2$ and $\prob{E_i ^\bOne \cap E_j ^\bZero \mid C} = 0$.
When $C_i \neq C_j$, \ie, nodes $i$ and $j$ belongs to different clusters, there are two scenarios:
the two clusters are assigned together and oppositely into treatment and control, 
or the two clusters are assigned independently.
Under the first scenario, which happens with probability $\frac{1}{k-1}$ conditional on $C_i \neq C_j$,
$E_i ^\bOne \cap E_j ^\bOne$ is not possible;
under the second scenario which happens with probability $\frac{k-2}{k-1}$ conditional on $C_i \neq C_j$,
$E_i ^\bOne \cap E_j ^\bOne$ happens when both clusters are assigned into treatment 
and thus the conditional probability is $1/4$ as $n \rightarrow \infty$.
Therefore, we have
\[
\prob{E_i ^\bOne \cap E_j ^\bOne \mid C_i \neq C_j} \rightarrow \frac{1}{k-1} \cdot 0 + \frac{k-2}{k-1} \cdot \frac{1}{4},
\]
and thus
\[
\prob{E_i ^\bOne \cap E_j ^\bOne} \rightarrow \frac{1}{2} \cdot \prob{C_i = C_j} + \frac{k-2}{4(k-1)} \cdot \prob{C_i \neq C_j}
= \left\{
\begin{array}{ll}
\frac{1}{2} - \frac{k \delta(\alpha_i, \alpha_j)}{8\pi}\cdot \frac{k}{k-1}  &  \text{if $\delta(\alpha_i, \alpha_j) < \frac{2\pi}{k}$}   \\
\\
\frac{1}{4}\cdot \frac{k-2}{k-1}  &  \text{if $\delta(\alpha_i, \alpha_j) \geq \frac{2\pi}{k}$} .  
\end{array}
\right.
\]
Similarly, one can show that
\[
\prob{E_i ^\bOne \cap E_j ^\bZero}
\rightarrow \left\{
\begin{array}{ll}
\frac{k \delta(\alpha_i, \alpha_j)}{8\pi}\cdot \frac{k}{k-1}  &  \text{if $\delta(\alpha_i, \alpha_j) < \frac{2\pi}{k}$}   \\
\\
\frac{1}{4}\cdot \frac{k}{k-1}  &  \text{if $\delta(\alpha_i, \alpha_j) \geq \frac{2\pi}{k}$} .  
\end{array}
\right.
\]

With these exposure probabilities, we have
\begin{eqnarray*}
& & \var{\hat \mu(\bOne)} 
\\ &=&\textstyle \frac{1}{n^2} \sum_{i=1}^n \sum_{j=1}^n \left(\frac{\prob{E_i ^\bOne \cap E_j ^\bOne}}{\prob{E_i ^\bOne}\prob{E_j ^\bOne}} - 1\right) Y_i(E_i ^\bOne)Y_j(E_j ^\bOne)
\\&=&\textstyle \frac{1}{n^2} \sum_{i=1}^n \left( \sum_{j: \delta(\alpha_i, \alpha_j) < \frac{2\pi}{k}}\left(1 - \frac{k \delta(\alpha_i, \alpha_j)}{2\pi} \cdot \frac{k}{k-1}\right) (a + b \sin \alpha_i + \tau)(a + b \sin \alpha_j + \tau) \right.
\\& &\qquad\qquad\quad \textstyle + \left. \sum_{j: \delta(\alpha_i, \alpha_j) \geq \frac{2\pi}{k}}\left(\frac{k-2}{k-1} - 1\right) (a + b \sin \alpha_i + \tau)(a + b \sin \alpha_j + \tau) \right) + o(1)
\\& \to & \frac{1}{4\pi^2} \int_{ 0} ^{2\pi}(a + b \sin \alpha_i + \tau) \cdot \left[\int_{0} ^{\frac{2\pi}{k}} \left(1 - \frac{k \delta}{2\pi}\cdot \frac{k}{k-1}\right) (2a + 2\tau + b\sin(\alpha_i + \delta) + b\sin(\alpha_i - \delta)) d\delta 
\right.
\\& &\qquad\qquad\qquad\qquad\qquad\qquad \left. +
\int_{\frac{2\pi}{k}} ^{\pi} \left(-\frac{1}{k-1}\right) (2a + 2\tau + b\sin(\alpha_i + \delta) + b\sin(\alpha_i - \delta)) d\delta
\right] d\alpha_i
\\& = & \frac{1}{4\pi^2} \int_{ 0} ^{2\pi}(a + b \sin \alpha_i + \tau) \cdot \left[\int_{0} ^{\frac{2\pi}{k}} \left(1 - \frac{k \delta}{2\pi}\cdot \frac{k}{k-1}\right) (2a + 2\tau + 2b\sin\alpha_i \cos \delta) d\delta 
\right.
\\& &\qquad\qquad\qquad\qquad\qquad\qquad \left. +
\int_{\frac{2\pi}{k}} ^{\pi} \left(-\frac{1}{k-1}\right) (2a + 2\tau + 2b\sin\alpha_i \cos\delta) d\delta
\right] d\alpha_i
\\&=& \frac{1}{4\pi^2} \int_{ 0} ^{2\pi}(a + \tau + b \sin \alpha_i) \cdot \left[
(a +\tau)\frac{2 \pi (k-2)}{k(k-1)} - \frac{2b}{k-1} \sin(2\pi/k) \sin\alpha_i + \frac{bk(1-\cos(2\pi / k))}{\pi}\cdot \frac{k}{k-1} \sin \alpha_i \right.
\\& &\qquad\qquad\qquad\qquad\qquad\qquad \left.
-(a +\tau)\frac{2 \pi (k-2)}{k(k-1)} + \frac{2b}{k-1} \sin(2\pi/k) \sin\alpha_i\right] d\alpha_i
\\&=& \frac{1}{4\pi^2} \int_{ 0} ^{2\pi}(a + \tau + b \sin \alpha_i) \cdot \frac{bk(1-\cos(2\pi / k))}{\pi}\cdot \frac{k}{k-1} \sin \alpha_i d\alpha_i
\\&=& \frac{b^2 k^2 (1 - \cos(2\pi / k))}{4 \pi^2(k-1)}.
\end{eqnarray*}
Analogously, it can be shown that
\begin{eqnarray*}
\var{\hat \mu(\bZero)} &\rightarrow&  \frac{b^2 k^2 (1 - \cos(2\pi / k))}{4 \pi^2(k-1)}, \\
\cov{\hat \mu(\bOne), \hat \mu(\bZero)} &\rightarrow& \frac{b^2 k^2 (1 - \cos(2\pi / k))}{2 \pi^2(k-1)}.
\end{eqnarray*}
Consequently
\[
\var{\hat \tau} = \var{\hat \mu(\bOne)} + \var{\hat \mu(\bZero)} - 2 \cov{\hat \mu(\bOne), \hat \mu(\bZero)} \rightarrow  \frac{b^2 k^2 (1 - \cos(2\pi / k))}{\pi^2(k-1)}.
\]

\xhdr{Variance with increasing number of clusters}
In the end we show the variance when the number of clusters $k$ increases.
Due the fact that $1 - \cos(x) \sim x^2 / 2$ as $x \rightarrow 0$, under the independent randomization scheme, we have
\[
\var{\hat \tau} \rightarrow  \frac{(2a+\tau)^2}{k} + \frac{b^2 k}{\pi^2} (1 - \cos({2\pi}/{k})) 
\sim \frac{(2a+\tau)^2}{k} + \frac{b^2 k}{\pi^2} \cdot \frac{2 \pi^2}{k^2} 
= [(2a+\tau)^2 + 2b^2]\cdot \Theta(1/k).
\]
Under the complete randomization scheme, we have
\[
\var{\hat \tau} \rightarrow \frac{b^2k^2}{\pi^2(k-1)} (1 - \cos({2\pi}/{k})) 
\sim  \frac{b^2k^2}{\pi^2(k-1)} \cdot \frac{2\pi^2}{k^2}
=2b^2 \cdot \Theta(1/k).
\]

\end{proof}

\end{document}